\documentclass{article}

\usepackage{times}
\usepackage{helvet}
\usepackage{courier}
\usepackage{amsmath}
\usepackage{amsthm}
\usepackage{amsfonts}
\usepackage{enumitem}
\usepackage{array}
\usepackage{rotating}
\usepackage{algorithm}
\usepackage[noend]{algorithmic}
\usepackage{authblk}
\usepackage{fullpage}

\def \N {\mathbb{N}}

\def \G {\mathbb{G}}
\def \M {\mathbb{M}}

\def \ie {i.e.}
\def \eg {e.g.}
\def \etal {et al.}
\def \etc {\textit{etc.}}

\def \vs {\widehat{v}}
\def \F {F}
\def \Add {A^*}
\def \ps {\pi}
\def \dist {\lambda}

\newtheorem{definition}{Definition}

\newtheorem{theorem}{Theorem}

\newtheorem{observation}{Observation}

\DeclareMathOperator*{\argmax}{arg\,max}
\DeclareMathOperator*{\argmin}{arg\,min}

\begin{document}

\title{Hiding in Multilayer Networks}

\author[a,b]{Marcin Waniek}
\author[b,*]{Tomasz P. Michalak}
\author[a,*]{Talal Rahwan}

\renewcommand*{\Affilfont}{\normalsize}

\affil[a]{Computer Science, New York University, Abu Dhabi, UAE}
\affil[b]{Institute of Informatics, University of Warsaw, Warsaw, Poland}
\affil[*]{To whom correspondence should be addressed:  tpm@mimuw.edu.pl, tr72@nyu.edu}

\date{}

\maketitle

\begin{abstract}
Multilayer networks allow for modeling  complex relationships, where individuals are embedded in multiple social networks at the same time. Given the ubiquity of such relationships, these networks have been increasingly gaining attention in the literature. This paper presents the first analysis of the robustness of centrality measures against strategic manipulation in multilayer networks.
More specifically, we consider an ``evader'' who strategically chooses which connections to form in a multilayer network in order to obtain a low centrality-based ranking---thereby reducing the chance of being highlighted as a key figure in the network---while ensuring that she remains connected to a certain group of people. We prove that determining an optimal way to ``hide'' is NP-complete and hard to approximate for most centrality measures considered in our study. Moreover, we empirically evaluate a number of heuristics that the evader can use. Our results suggest that the centrality measures that are functions of the entire network topology are more robust to such a strategic evader than their counterparts which consider each layer separately.
\end{abstract}

\section{Introduction}
\label{sec:introduction}
Owing to several incidents in the past few years, most notably those concerning the American presidential elections of 2016, the general public has become increasingly concerned with the privacy and security of their online activities \cite{persily20172016}. Experts, however, had been warning about such potential risks long ago. For instance,  Mislove~\etal~\cite{Mislove:2010} famously showed that, by coupling
the social network of a given Facebook user with publicly-known attributes of some other users, it is possible to infer otherwise-private information about that user.
Worryingly, this is true not only for typically innocuous data, but also for potentially-sensitive confidential information such as political preferences (as demonstrated in the case of Cambridge Analytica), or even sexual orientation \cite{kitchin2016ethics}.

Various proposals on how to deal with such privacy challenges have already been put forward. Among those proposals is the
General Data Protection Regulation, implemented in May 2018, which is perhaps the most well-known attempt to use state-enforced, legal instruments \cite{EU:2016:gdpr}. On the other hand, there have been a plethora of algorithmic solutions for privacy protection \cite{Lane:et:al:2014,Kearns:et:al:2016}. Perhaps the most well-known such solutions come from the network anonymization and de-anonymization literature \cite{Zhou:2008,Narayanan:Shmatikov:2009,Kayes:Iamnitchi:2015}, which studies the problem faced by a \textit{data trustee} who publishes anonymized network data to be analyzed for various purposes. In this literature, the responsibility of protecting the privacy of the network members lies solely on the shoulders of the data trustee, while the network members are implicitly assumed to be passive in this regard. In contrast, a recent body of work studies ways in which the network members can themselves protect their own privacy by acting strategically to evade various tools from the social network analysis toolkit \cite{Michalak:et:al:2017}. In this context, three fundamental classes of tools have been considered: (1) centrality measures, (2) community detection algorithms; and (3) link prediction algorithms. More specifically, Waniek~\etal~\cite{waniek2018hiding,waniek2017construction} studied how key individuals in a social network could rewire the network to avoid being highlighted by centrality measures while maintaining their own influence within the network. The authors also studied how a group of individuals could avoid being identified by community detection algorithms. Furthermore, Yu~\etal~\cite{yu2018target}, Waniek~\etal~\cite{Waniek:et:al:2019}, and Zhou~\etal~\cite{zhou2019attacking} studied how to hide one's sensitive relationships from link prediction algorithms. 

The aforementioned literature on the strategic behaviour of network members demonstrates that it is indeed possible to develop reasonably effective heuristics to escape detection by fundamental network analysis tools. However, the main limitation of this literature is that it focuses only on standard, single-layered networks. In contract, people often interact with each other via a complicated pattern of relationships, thereby creating multiple subsystems, or ``layers'', of connectivity. This is even more so nowadays when many of us belong to multiple social media platforms simultaneously. Furthermore, multilayer networks are increasingly being recognized not only in the context of human interactions, but also in many natural and engineered systems~\cite{DeDomenico:et:al:2013a}. For instance, to travel from one point to another in many urban transportation networks, one can choose between a road subnetwork (car or taxis), bus or tram subnetwork, subway subnetwork, local train subnetwork, bike subnetwork, footpath subnetwork, or any combination thereof. Each such subnetwork has its own distinct characteristics, which become difficult, or even impossible, to account for if modelled as a single layer due to the interdependencies between the different layers.
The theoretical and empirical analysis of multilayer networks has recently attracted significant attention (see the work by Kivel{\"a}~\etal~\cite{kivela2014multilayer} for a comprehensive review). This new body of research is primarily driven by the fact that, due to the much more complex nature of multilayer networks, many results for singlelayer networks become obsolete.
 
Motivated by these observations, we present in this paper the first analysis of how to protect ones' privacy against centrality measures in multilayer networks. Specifically, we consider an evader who wishes to connect to a certain group of individuals, without being highlighted by centrality measures as a key member in the multilayer network. To this end, the evader has to strategically choose at which layer(s) to connect to those individuals. We prove that the corresponding optimization problem is NP-complete and hard to approximate for most centrality measures considered in our study. Furthermore, we empirically evaluate a number of heuristic algorithms that the evader can use. The results of this evaluation suggest that the centrality measures that are functions of the entire network topology are more robust to such a strategic evader than their counterparts which consider each layer separately.

\section{Preliminaries}
\label{sec:preliminaries}


\subsection{Basic Network Notation and Definitions}

Let $G =(V, E) \in \G$ denote a \textit{simple (single-layer) network}, where $V$ is the set of $n$ nodes and $E \subseteq V \times V$ the set of edges.
We denote an edge between nodes $v$ and $w$ by~$(v,w)$.

In this paper we consider \textit{multilayer networks}, \ie, networks where edges can represent different types of relations.
We will denote a multilayer network by $M=(V_L,E_L,V,L) \in \M$, where $V$ is the set of nodes, $L$ is the set of layers (\ie, types of relations), $V_L \subseteq V \times L$ is the set of occurrences of nodes in layers (e.g., having $(v,\alpha)\in V_L$ means that node $v$ appears in layer $\alpha$), and $E_L \subseteq V_L \times V_L$ is the set of edges.
We will denote an occurrence of node $v$ in layer $\alpha$ by $v^\alpha$.
Note that $V = \{v : \exists_{\alpha \in L} v^\alpha \in V_L\}$.
Let $V^\alpha$ be the set of nodes occurring in layer $\alpha$, \ie, $V^\alpha=\{v \in V : v^\alpha \in V_L\}$, and let $G^\alpha$ denote the simple network consisting of all the nodes and edges in layer $\alpha$, \ie, $G^\alpha = (V^\alpha, \{(v,w): (v^\alpha,w^\alpha) \in E_L\})$. 

We focus on \textit{undirected} networks, \ie, we do not discern between edges $(v^\alpha,w^\beta)$ and $(w^\beta,v^\alpha)$. Moreover, we do not consider self-loops, \ie, $\forall_{v^\alpha \in V_L}(v^\alpha,v^\alpha) \notin E_L$.
Multilayer network allow for \textit{inter-layer edges}, which are edges between two layers; they may connect two different nodes, or may connect two occurrences of the same node. We restrict our attention to networks with \textit{diagonal couplings}, \ie, networks where every inter-layer edge connects two occurrences of the same node, i.e., $\forall_{(v^\alpha,w^\beta) \in E_L} \alpha \neq \beta \rightarrow v = w$.

Notice that, in some literature, multilayer networks with diagonal couplings are called \textit{multiplex} networks. However, it is also typically assumed that the multiplex networks are node-aligned (\ie, every node occurs in every layer), which is not the case in our setting. Hence, we will use the more general term ``\textit{multilayer networks}''. For a comprehensive discussion of the nomenclature, see Kivel{\"a}~\etal~\cite{kivela2014multilayer}.

A path in a simple network is an ordered sequence of nodes in which every two consecutive nodes are connected by an edge.
A path in a multilayer network is an ordered sequence of node occurrences in which every two consecutive occurrences are connected by an edge.
The length of a path is the number of edges in that path.
The set of all shortest paths between a pair of nodes, $v,w \in V$ will be denoted by $\ps_G(v,w)$.
The distance between a pair of nodes $v,w \in V$ is the length of a shortest path between them, and is denoted by $\dist_G(v,w)$.
We assume that if there does not exist a path between $v$ and $w$ then $\dist_G(v,w)=\infty$.
In a multilayer network we consider distance between $v$ and $w$ to be the shortest distance between an occurrence of $v$ in any layer $\alpha$ and an occurrence of $w$ in any layer $\beta$ (possibly $\alpha \neq \beta$).

For any node, $v\in V$, in a simple network, $G$, we denote by $N_G(v) = \{w \in V : (v,w) \in E\}$ the set of neighbors of $v$ in $G$. Similarly, given a multilayer network $M$, we write $N_M(v) = \{w \in V : (v^\alpha,w^\beta) \in E_L\}$.
Finally, we denote by $N^\alpha_M(v)$ the set of neighbors of $v$ in layer $\alpha$, \ie, $N^\alpha_M(v) = \{w \in V : (v^\alpha,w^\alpha) \in E_L\}$.
We will often omit the network itself from the notation whenever it is clear from the context, \eg, by writing $\dist(v,w)$ instead of $\dist_G(v,w)$.

\subsection{Centrality Measures}

A centrality measure~\cite{bavelas1948mathematical} is a function that expresses the importance of a given node in a given network.
Arguably, the best-known centrality measures are \textit{degree}, \textit{closeness} and \textit{betweenness}.

\textit{Degree centrality}~\cite{shaw1954group} assumes that the importance of a node is proportional to the number of its neighbors, i.e., the degree centrality of node $v$ in network $G$ is:
$$
	c_{degr}(G,v) = |N_G(v)|.
$$

\textit{Closeness centrality}~\cite{beauchamp1965improved} quantifies the importance of a node in terms of shortest distances from this node to all other nodes in the network. Formally, the closeness centrality of node $v$ in network $G$ can be expressed as:
$$
	c_{clos}(G,v) = \sum_{w \in V \setminus \{v\}}\frac{1}{\dist_G(v,w)}.
$$

\textit{Betweenness centrality}~\cite{anthonisse1971rush,freeman1977set} states that, if we consider all the shortest paths in the network, then the more such paths traverse through a given node (it is often stated that the node \textit{controls} such paths), the more important the role of that node in the network. More formally, the betweenness centrality of node $v \in V$ in network $G$ is:
$$
	c_{betw}(G,v) =
		\sum_{w,u \in V \setminus \{v\}}
			\frac {|\{p\in\ps_G(w,u) : v\in p\}|} {|\ps_G(w,u)|}.
$$

The definitions of degree and closeness centrality can be generalized to multilayer networks using the definitions of neighbors and distance for multilayer networks (see above).
As for the betweenness centrality of node $v$ in a multilayer network $M$, it grows with the number of occurrences of $v$ on the shortest paths between pairs of other nodes:
$$
	c_{betw}(M,v) =
		\sum_{w,u \in V \setminus \{v\}}
			\frac {|\{(v^\alpha,p): v^\alpha \in p, p \in \ps_M(w,u)\}|} {|\ps_M(w,u)|}
$$
To avoid any potential confusion, the measures that are designed for simple networks will be referred to as ``\textit{local} centrality measures'', since they can be applied to only a single layer. Conversely, the measures that are designed for multilayer networks will be referred to as a  ``\textit{global} centrality measures'', since they take all layers into consideration. 

\section{Theoretical Analysis}\label{sec:definition-theoretical}

In this section we formally define our computational problems and then move on to analyse them.

\subsection{Definitions of Computational Problems}

We  define the decision problems before defining the corresponding optimization problems. Here, the ``group of contacts'' is the set of individuals to whom the evader wishes to connect while remaining hidden from centrality measures. 

\subsection{Decision Problems} We will define two different decision versions of this problem, starting with the global version.

\begin{definition}[Multilayer Global Hiding]
This problem is defined by a tuple, $(M,\vs,\F,c,d)$, where $M=(V_L, E_L, V, L)$ is a multilayer network, $\vs \in V$ is the evader, $\F \subset V$ is the group of contacts, $c$ is a centrality measure, and $d \in \N$ is a safety margin.
The goal is to identify a set of edges to be added to the network, $\Add \subseteq \{(\vs^\alpha,v^\alpha) : v \in \F \land \vs^\alpha \in V_L \land v^\alpha \in V_L\}$, such that in the resulting network $\widehat{M}=(V_L, E_L \cup \Add, V, L)$ the evader is connected with every contact in at least one layer and there are at least $d$ nodes with a centrality score greater than that of the evader, \ie:
$$
\forall_{v \in \F} \exists_{\alpha \in L} (\vs^\alpha, v^\alpha) \in \Add,
$$
$$
\exists_{W \subset V} \left( |W| \geq d \land \forall_{v \in W} c(\widehat{M},v) > c(\widehat{M},\vs) \right).
$$
\end{definition}

We say that ``$\vs$ is hidden'' when there are at least $d$ nodes whose centrality is greater than that of $\vs$.

\begin{definition}[Multilayer Local Hiding]
This problem is defined by a tuple, $(M,\vs,\F,c,\left(d^\alpha\right)_{\alpha \in L})$, where $M=(V_L, E_L, V, L)$ is a multilayer network, $\vs \in V$ is the evader, $\F \subset V$ is the group of contacts, $c$ is a centrality measure, and $d^\alpha \in \N$ is a safety margin for layer $\alpha \in L$.
The goal is to identify a set of edges to add, $\Add \subseteq \{(\vs^\alpha,v^\alpha) : v \in \F \land \vs^\alpha \in V_L \land v^\alpha \in V_L\}$, such that in the resulting network $\widehat{M}=(V_L, E_L \cup \Add, V, L)$ the evader is connected with every contact in at least one layer and for each layer $\alpha$ the network $G^\alpha$ contains at least $d^\alpha$ nodes with a centrality score greater than that of the evader, \ie:
$$
\forall_{v \in \F} \exists_{\alpha \in L} (\vs^\alpha, v^\alpha) \in \Add,
$$
$$
\forall_{\alpha \in L} \exists_{W \subset V^{\alpha}} \left( |W| \geq d^\alpha \land \forall_{v \in W} c(\widehat{M}^{\alpha},v) > c(\widehat{M}^{\alpha},\vs) \right).
$$
\end{definition}

We say that ``$\vs$ is hidden in $\alpha$'' if there are at least $d^\alpha$ nodes with centrality in layer $\alpha$ greater than that of $\vs$ in $\alpha$.

In the global version of the problem we assume that the seeker is able to observe and analyze the entire multilayer network using centrality measures, hence the evader's goal is to minimize her centrality ranking in the network as a whole.
On the other hand, the local version of the problem models situations where the seeker analyzes only one of the layers, \eg, if the seeker gains access to the email communication network, but not to the phone-call network. In such situations, the evader's goal is to attain an adequate level of safety in each layer separately.

The approach to hiding represented by the two problems differs from the one developed for simple networks by Waniek~\etal~\cite{waniek2017construction,waniek2018hiding}.
Their hiding algorithms focus on choosing which edge(s) to add or remove from the single layer, often causing the evader to lose the direct connection with some of the neighbors.
The algorithms presented in our paper focus on choosing the layer in which to maintain the connection, and allow the evader to keep direct links with all contacts.
Notice that this approach cannot be applied to simple networks, as there is only one way to have a direct link between the evader and every contact in a single layer.

\subsection{Optimization Problems}
We now define the corresponding optimization problems. They take into consideration a situation when it is impossible to connect the evader with all the contacts.

\begin{definition}[Maximum Multilayer Global Hiding]
This problem is defined by a tuple, $(M,\vs,\F,c,d)$, where $M=(V_L, E_L, V, L)$ is a multilayer network, $\vs \in V$ is the evader, $\F \subset V$ is the group of contacts, $c$ is a centrality measure, and $d \in \N$ is a safety margin.
The goal is then to identify a set of edges to be added to the network, $\Add \subseteq \{(\vs^\alpha,v^\alpha) : v \in \F \land \vs^\alpha \in V_L \land v^\alpha \in V_L\}$, such that in the resulting network $\widehat{M}=(V_L, E_L \cup \Add, V, L)$ the evader is connected with as many contacts as possible, while there are at least $d$ nodes with a centrality score greater than that of the evader.
\end{definition}

\begin{definition}[Maximum Multilayer Local Hiding]
This problem is defined by a tuple, $(M,\vs,\F,c,\left(d^\alpha\right)_{\alpha \in L})$, where $M=(V_L, E_L, V, L)$ is a multilayer network, $\vs \in V$ is the evader, $\F \subset V$ is the group of contacts, $c$ is a centrality measure, and $d^\alpha \in \N$ is a safety margin for layer $\alpha \in L$.
The goal is then to identify a set of edges to be added to the network, $\Add \subseteq \{(\vs^\alpha,v^\alpha) : v \in \F \land \vs^\alpha \in V_L \land v^\alpha \in V_L\}$, such that in the resulting network $\widehat{M}=(V_L, E_L \cup \Add, V, L)$ the evader is connected with as many contacts as possible, while for each layer $\alpha$ the network $G^\alpha$ contains at least $d^\alpha$ nodes with a centrality score greater than that of the evader.
\end{definition}

Intuitively, the goal is to connect the evader with as many contacts as possible, while keeping the evader hidden.

\begin{table}[t]
\caption{Summary of our computational complexity results.}
\smallskip
\centering
\smallskip
\begin{tabular}{lcc}
Centrality	& Multilayer Global Hiding	& Multilayer Local Hiding \\
\hline
Degree		& P							& NP-complete	\\
Closeness	& NP-complete				& NP-complete	\\
Betweenness	& NP-complete				& NP-complete	\\
\end{tabular}
\label{tab:complexity}
\end{table}

\subsection{Complexity Analysis}

The complexity results for both the global and local versions of the problem are listed below (see Table~\ref{tab:complexity} for a summary).

\begin{observation}
\label{thrm:npc-degree-global}
The problem of Multilayer Global Hiding is in P given the degree centrality measure.
In fact, for a given problem instance either any $\Add$ that connects $\vs$ with all contacts is a solution, or there are no solutions at all.
\end{observation}

\begin{proof}
Any valid solution to the problem $\Add$ must connect the evader $\vs$ with all contacts.
Therefore, after the addition of $\Add$ the degree centrality of $\vs$ is $|\F|$, while the degree centrality of every contact increases by $1$.
Hence, the degree centrality ranking in the network does not depend on the choice of layers in which $\vs$ gets connected with its contacts.
\end{proof}

\begin{theorem}
\label{thrm:npc-closeness-global}
The problem of Multilayer Global Hiding is NP-complete given the closeness centrality measure.
\end{theorem}

\begin{proof}

The problem is trivially in NP, since after the addition of a given $\Add$ the closeness centrality ranking can be computed in polynomial time.

Next, we prove that the problem is NP-hard.
To this end, we show a reduction from the NP-complete problem of \textit{Exact 3-Set Cover}.
The decision version of this problem is defined by a set of subsets $S=\{S_1, \ldots, S_{m}\}$ of universe $U=\{u_1,\ldots,u_{3k}\}$, such that $\forall_i |S_i|=3$.
The goal is to determine whether there exist $k$ pairwise disjoint elements of $S$ the sum of which equals $U$.

Given an instance of the problem of \textit{Exact 3-Set Cover}, let us construct a multilayer network, $M=(V_L, E_L, V', L)$, as follows (Figure~\ref{fig:npc-closeness-global} depicts an instance of this network):

\begin{itemize}[leftmargin=*]
\item \textbf{The set of nodes $V'$:}
For every $u_i \in U$ we create a node $u_i$, as well as $3$ nodes $w_{i,1},w_{i,2},w_{i,3}$.
We will denote the set of all nodes $u_i$ by $U$, and the set of all nodes $w_{i,j}$ by $W$.
We also create the evader node $\vs$, the node $v'$, and the following four sets of nodes:
\begin{enumerate}
\item $A=\{a_1,\ldots,a_{m}\}$; \item $B=\{b_1,\ldots, b_{2k+2m}\}$; \item $B'=\{b'_1,\ldots, b'_{k+2m+1}\}$;
\item $B''=\{b''_1,\ldots, b''_{2k+m-1}\}$.
\end{enumerate}

\item \textbf{The set of layers $L$:}
For every $S_i \in S$ we create a layer $\alpha_i$.
We also create an additional layer $\beta$.

\item \textbf{The set of occurrences of nodes in layers $V_L$:}
Node $u_j \in U$ appears in layer $\alpha_i$ if and only if $u_j \in S_i$.
Node $w_{i,j} \in U$ appears only in layer $\alpha_i$.
The evader $\vs$, as well as all nodes in $A$ appear in every layer $\alpha_i$.
Node $v'$, as well as all nodes in $B$, $B'$, and $B''$ appear only in layer $\beta$.

\item \textbf{The set of edges $E_L$:}
For every node that appears in multiple layers, we connect all occurrences of this node in a clique.
For node $u_j$ in layer $\alpha_i$ we connect it with node $a_i$.
In every layer $\alpha_i$ we connect all nodes in $A$ into a clique. Moreover, we connect every node $b_i$ with node $v'$, and connect every node $b'_i$ with node $b_i$. Finally, we connect every node $b''_i$ with node $b'_i$.
\end{itemize}

\begin{figure*}[t]
\centering
\includegraphics[width=.8\linewidth]{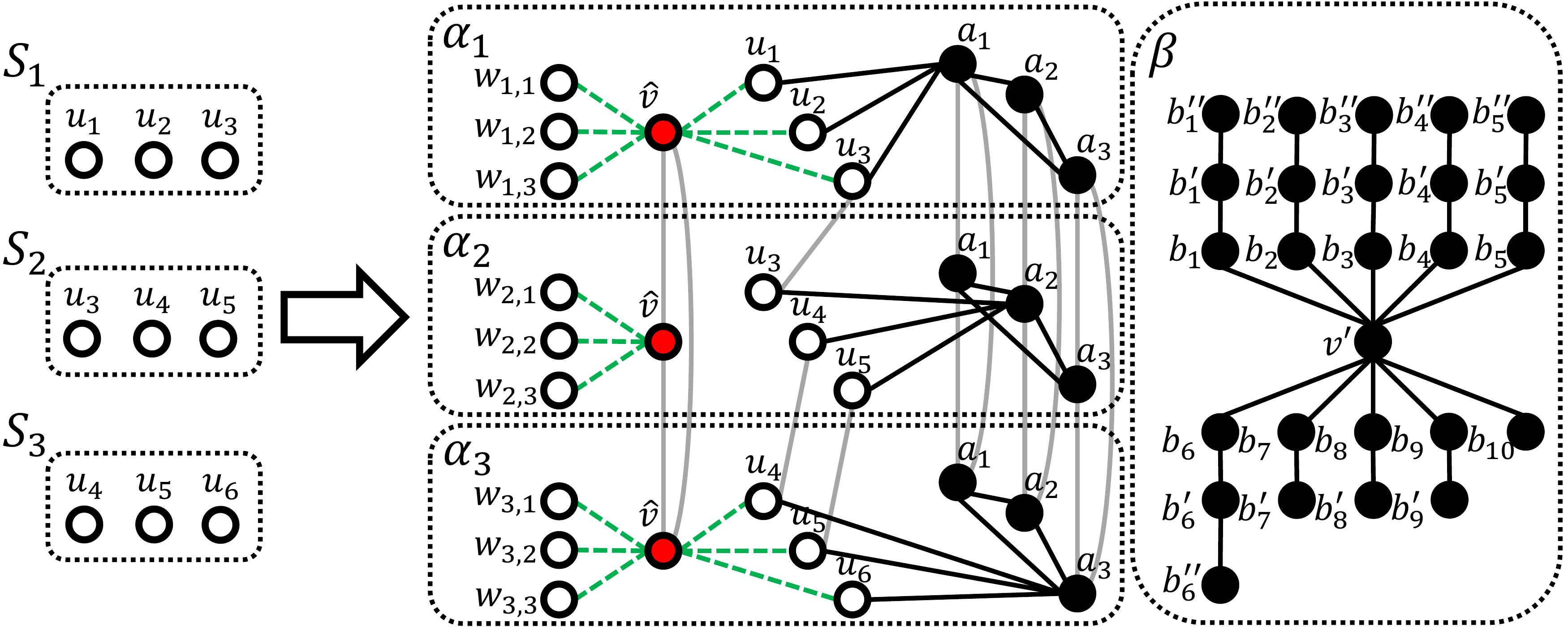}
\caption{An illustration of the network used in the proof of Theorem~\ref{thrm:npc-closeness-global}.
Edges connecting occurrences of the same node in different layers are highlighted in grey.
The red node represents the evader, while the white nodes represent the contacts.
Dashed (green) edges represent the solution to this problem instance.}
\label{fig:npc-closeness-global}
\end{figure*}

Now, consider the following instance of the problem of Multilayer Global Hiding, $(M,\vs,\F,c,d)$, where:

\begin{itemize}[leftmargin=*]
\item $M$ is the multilayer network we just constructed;
\item $\vs$ is the evader;
\item $\F=U \cup W$ is the set of contacts;
\item $c$ is the closeness centrality measure;
\item $d = 1$.
\end{itemize}

Next, let us analyze the closeness centrality values of nodes in the network.
Notice that every node $w_{i,j}$ appears only in a single layer $\alpha_i$, hence $\vs$ has to connect with $w_{i,j}$ in layer $\alpha_i$.
Assume that the evader $\vs$ has connections with nodes in $U$ in exactly $x$ layers, \ie, $x=|\{\alpha_i \in L : \exists u_j (\vs^{\alpha_i},u_j^{\alpha_i}) \in \Add\}|$.
We then have:
\begin{itemize}[leftmargin=*]
\item $c_{clos}(\vs)= 3k + 3m + \frac{x}{2} + \frac{m-x}{3} \geq 3k + 3\frac{1}{3}m$ as $\vs$ is a neighbor of $3k$ nodes in $U$ and $3m$ nodes in $W$, while for any $a_i \in A$ the distance between $a_i$ and $\vs$ is $2$ if $\vs$ is connected with any $u_j$ in layer $\alpha_i$ and $3$ otherwise;
\item $c_{clos}(u_i) \leq 1 + m + \frac{3k-1}{2} + \frac{3m}{2} = 1\frac{1}{2}k + 2\frac{1}{2}m + \frac{1}{2} < c_{clos}(\vs)$ as $u_i$ is a neighbor of $\vs$ and at most $m$ nodes in $A$, while the distance to all other nodes is at least $2$;
\item $c_{clos}(a_i) \leq 3 + m - 1 + \frac{1}{2} + \frac{3k-3}{2} + \frac{3m}{3} = 1\frac{1}{2}k + 2m +1 < c_{clos}(\vs)$ as $a_i$ is a neighbor of $3$ nodes from $U$ and all other $m-1$ nodes in $A$, while the distance to $\vs$ and all other nodes in $U$ is $2$, and the distance to all nodes is $W$ is at least $3$;
\item $c_{clos}(w_{i,j}) < c_{clos}(\vs)$ as for any other node $v$ we have $\lambda(w_{i,j},v)= \lambda(\vs,v)+1$, since the shortest paths between $w_{i,j}$ and all other nodes go through $\vs$;
\item $c_{clos}(v') = 2k + 2m + \frac{k + 2m + 1}{2} + \frac{2k+m-1}{3} = 3k + 3m + \frac{k}{2} + \frac{m-k}{3} + \frac{1}{6}$ as $v'$ is a neighbor of all $2k+2m$ nodes in $B$, the distance to all $k+2m+1$ nodes in $B'$ is $2$, while the distance to all $m-k+1$ nodes in $B''$ is $3$.
\end{itemize}

We have shown that all nodes in $A$, $U$, and $W$ have smaller closeness centrality than $\vs$.
It is easy to check that $v'$ has greater closeness centrality than all other nodes occurring in layer $\beta$.
Hence, $\vs$ is hidden if and only if $v'$ has greater closeness centrality than $\vs$.
This is true when:
$$
3k + 3m + \frac{x}{2} + \frac{m-x}{3} < 3k + 3m + \frac{k}{2} + \frac{m-k}{3} + \frac{1}{6}
$$
which can be simplified to $x < k + 1$.
Since both $x$ and $k$ are in $\N$ this is equivalent to $x \leq k$.
Therefore, $\vs$ is hidden if and only if it has connections with nodes in $U$ in at most $k$ layers.

Now we will show that if there exists a solution to the given instance of the Exact 3-Set Cover problem, then there also exists a solution to the constructed instance of the Multilayer Global Hiding problem.
Let $S^*$ be an exact cover of $U$.
In layer $\alpha_i$ we connect $\vs$ with all nodes from $W$ that occurr in this layer.
For every $S_i \in S^*$ we connect $\vs$ with $u_j \in S_i$ in layer $\alpha_i$.
This way, $\vs$ becomes connected to all $3k$ contacts from $U$, since all the sets in $S^*$ are pairwise disjoint.

To complete the proof, we have to show that if there exists a solution $\Add$ to the constructed instance of the Multilayer Global Hiding problem, then there also exists a solution to the given instance of the Exact 3-Set Cover problem.
We have shown above that if $\vs$ is hidden, then it is connected to nodes in $U$ in at most $k$ layers from $\{\alpha_1, \ldots, \alpha_m\}$.
However, since $\vs$ must be connected with all $3k$ nodes in $U$ in order for $\Add$ to be a correct solution, then $\{S_i: \exists u_j (\vs^{\alpha_i},u_j^{\alpha_i}) \in \Add\}$ is a solution to the given instance of the Exact 3-Set Cover problem.
This concludes the proof.
\end{proof}

\begin{figure*}[t]
\centering
\includegraphics[width=.8\linewidth]{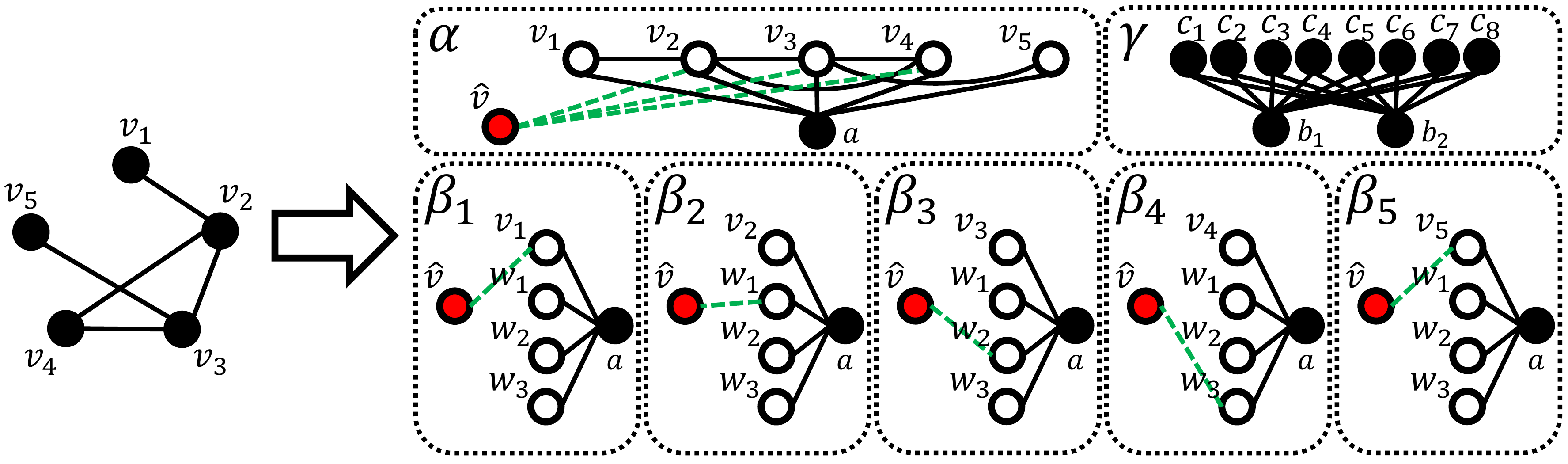}
\caption{An illustration of the network used in the proof of Theorem~\ref{thrm:npc-betweenness-global}.
The red node represents the evader, while the white nodes represent the contacts.
Dashed (green) edges represent the solution to this problem instance.}
\label{fig:npc-betweenness-global}
\end{figure*}

\begin{theorem}
\label{thrm:npc-betweenness-global}
The problem of Multilayer Global Hiding is NP-complete given the betweenness centrality measure.
\end{theorem}

\begin{proof}

The problem is trivially in NP, since after the addition of a given $\Add$ the betweenness centrality rankings can be computed in polynomial time.

Next, we prove that the problem is NP-hard.
To this end, we show a reduction from the NP-complete problem of \textit{Finding $k$-Clique}.
The decision version of this problem is defined by a simple network, $G=(V,E)$, and a constant, $k \in \N$.
The goal is to determine whether there exist $k$ nodes in $G$ that form a clique.

Given an instance of the problem of \textit{Finding $k$-Clique}, defined by $k$ and a simple network $G=(V,E)$, let us construct a multilayer network, $M=(V_L, E_L, V', L)$, as follows (Figure~\ref{fig:npc-betweenness-global} depicts an instance of this network):

\begin{itemize}[leftmargin=*]
\item \textbf{The set of nodes $V'$:}
For every node, $v_i \in V$, we create a node $v_i$.
Additionally, we create the evader node $\vs$, node $a$, and the following three sets of nodes:
\begin{enumerate}
\item $B=\{b_1,b_2\}$;
\item $W=\{w_1,\ldots,w_k\}$;
\item $C=\{c_1,\ldots,c_{n+k}\}$.
\end{enumerate}

\item \textbf{The set of layers $L$:}
We create a layer $\alpha$, a layer $\gamma$, as well as $n$ layers $\beta_1, \ldots, \beta_n$.

\item \textbf{The set of occurrences of nodes in layers $V_L$:}
Node $\vs$ and node $a$ appear in layer $\alpha$ and all layers $\{\beta_1, \ldots, \beta_n\}$.
Each node $v_i$ appears in layer $\alpha$ and $\beta_i$.
Nodes in $W$ appear in all layers $\{\beta_1, \ldots, \beta_n\}$.
Nodes in $B$ and $C$ appear only in layer $\gamma$.

\item \textbf{The set of edges $E_L$:}
In layer $\alpha$ we create an edge between two nodes $v_i,v_j \in V$ if and only if this edge was present in $G$.
In every layer where $a$ appears we connect it with all occurring nodes from $V$ and $W$.
Finally, we connect every node $c_i$ with both $b_1$ and $b_2$.
\end{itemize}

Now, consider the following instance of the problem of Multilayer Local Hiding, $(M,\vs,\F,c,\left(d^\alpha\right)_{\alpha \in L})$, where:

\begin{itemize}[leftmargin=*]
\item $M$ is the multilayer network we just constructed;
\item $\vs$ is the evader;
\item $\F= V \cup W$ is the set of contacts;
\item $c$ is the betweenness centrality measure;
\item $d = 2n+2k+3$ is the safety margin.
\end{itemize}

Notice that, since $d = 2n+2k+3$, all other nodes must have greater betweenness centrality than the evader in order for $\vs$ to be hidden.
Notice also that the betweenness centrality of every node $c_i$ is $\frac{1}{n+k}$. Moreover, after adding $\Add$ all nodes other than $\vs$ have non-zero betweenness centrality.
If $\vs$ gets connected to at least two nodes from $\F$ that are not connected to each other, then $\vs$  controls one of at most $n+k-1$ shortest path between them (other paths can only go through nodes in $V\cup W \cup \{a\}$) and thus the betweenness centrality of $\vs$ is at least $\frac{1}{n+k-1}$.
Therefore, in order to get hidden, $\vs$ cannot control any shortest paths in the network.
This implies that, if $\vs$ is hidden then all nodes that are connected to $\vs$ in layer $\alpha$ must form a clique, and also implies that in every layer $\beta$, the evader $\vs$ can be connected to at most one node (otherwise $\vs$ controls one of the shortest paths between its two neighbors without an edges between them).

Now we will show that if there exists a solution to the given instance of the Finding $k$-Clique problem, then there also exists a solution to the constructed instance of the Multilayer Global Hiding problem.
Let $V^*$ be a group of $k$ nodes forming a clique in $G$.
Let us create $\Add$ by connecting $\vs$ to nodes from $V^*$ in layer $\alpha$.
Now, we connect every $v_i \in V \setminus V^*$ to $\vs$ in layer $\beta_i$.
In the remaining layers from $\{\beta_1, \ldots, \beta_n\}$ (corresponding to elements $v_i \in V^*$) we connect $\vs$ to all nodes in $W$. 
As argued above, for such $\Add$, the evader $\vs$ is hidden, hence $\Add$ is a solution to the constructed instance of the Multilayer Global Hiding problem.

To complete the proof we have to show that if there exists a solution $\Add$ to the constructed instance of the Multilayer Global Hiding problem, then there also exists a solution to the given instance of the Finding $k$-Clique problem.
As argued above, in each layer $\beta_i$ the evader $\vs$ can be connected to at most one node.
Since all $k$ nodes from $W$ appear only in layers from $\{\beta_1, \ldots, \beta_n\}$, the evader $\vs$ can be connected to at most $n-k$ nodes from $V$ in layers from $\{\beta_1, \ldots, \beta_n\}$.
Therefore, $\vs$ has to have at least $k$ neighbors from $V$ in layer $\alpha$
As shown above, in order for $\vs$ to be hidden in $\alpha$, all of its neighbors must form a clique.
Hence, the neighbors of $\vs$ in layer $\alpha$ form a clique in $G$.
This concludes the proof.
\end{proof}

\begin{theorem}
\label{thrm:npc-degree-local}
The problem of Multilayer Local Hiding is NP-complete given the degree centrality measure.
\end{theorem}

\begin{proof}

The problem is trivially in NP, since after the addition of a given $\Add$ the degree centrality rankings for all layers can be computed in polynomial time.

Next, we prove that the problem is NP-hard.
To this end, we show a reduction from the NP-complete problem of \textit{Exact 3-Set Cover}.
The decision version of this problem is defined by a set of subsets $S=\{S_1, \ldots, S_{m}\}$ of universe $U=\{u_1,\ldots,u_{3k}\}$, such that $\forall_i |S_i|=3$.
The goal is to determine whether there exist $k$ pairwise disjoint elements of $S$ the sum of which equals $U$.

Given an instance of the problem of \textit{Exact 3-Set Cover}, let us construct a multilayer network, $M=(V_L, E_L, V', L)$, as follows (Figure~\ref{fig:npc-degree-local} depicts an instance of this network):

\begin{itemize}[leftmargin=*]
\item \textbf{The set of nodes $V'$:}
For every element, $u_i \in U$, we create a node $u_i$.
We also create $2(m-k)$ nodes $w_1,\ldots,w_{2(m-k)}$.
Additionally, we create the evader node $\vs$ and three nodes $a_1,a_2,a_3$.
We will denote the set of all nodes $a_i$ as $A$, the set of all nodes $u_i$ as $U$, and the set of all nodes $w_i$ as $W$.
\item \textbf{The set of layers $L$:}
For every $S_i \in S$ we add a layer $\alpha_i$.
\item \textbf{The set of occurrences of nodes in layers $V_L$:}
Node $u_j \in U$ appears in layer $\alpha_i$ if and only if $u_j \in S_i$.
The evader $\vs$, as well as all nodes in $A$ and $W$, appear in all layers.
\item \textbf{The set of edges $E_L$:}
In every layer we connect every node $u_j \in U$ occurring in this layer to every node in $A$.
\end{itemize}

Now, consider the following instance of the problem of Multilayer Local Hiding, $(M,\vs,\F,c,\left(d^\alpha\right)_{\alpha \in L})$, where:

\begin{itemize}[leftmargin=*]
\item $M$ is the multilayer network we just constructed;
\item $\vs$ is the evader;
\item $\F=U \cup W$ is the set of contacts;
\item $c$ is the degree centrality measure;
\item $d^{\alpha_i} = 3$ for every $\alpha_i \in L$.
\end{itemize}

Next, let us consider what are the sets of edges that can be added between the evader $\vs$ and the contacts $\F$ in each layer, so that the evader is hidden.
In every layer $\alpha_i$ the nodes in $A$ as well as the nodes $u_j \in S_i$ have degree $3$, while all other nodes have degree $0$.
We can connect $\vs$ to any two or less contacts and $\vs$ will still be hidden.
If we connect the evader to three contacts, they have to be nodes in $S_i$ (as these are the only nodes that potentially can have degree greater than $3$).
We cannot connect $\vs$ to more than three contacts and still have $\vs$ hidden.

Now we will show that if there exists a solution to the given instance of the Exact 3-Set Cover problem, then there also exists a solution to the constructed instance of the Multilayer Local Hiding problem.
Let $S^*$ be an exact cover of $U$.
For every $S_i \in S^*$ we connect $\vs$ to every $u_j \in S_i$ in layer $\alpha_i$.
This way, $\vs$ becomes connected to all $3k$ contacts from $U$, since all the sets in $S^*$ are pairwise disjoint.
For every $S_i \notin S^*$ we connect $\vs$ to two nodes from $W$ in layer $\alpha_i$ (since there are $m-k$ such layers, we can connect $\vs$ to all $2(m-k)$ contacts from $W$ this way).

To complete the proof we have to show that if there exists a solution to the constructed instance of the Multilayer Local Hiding problem, then there also exists a solution to the given instance of the Exact 3-Set Cover problem.
Let $x$ be the number of layers from $\{\alpha_1, \ldots, \alpha_m\}$ in which $\vs$ has at most two neighbors, and let $m-x$ be the number of layers from $\{\alpha_1, \ldots, \alpha_m\}$ where $\vs$ has exactly three neighbors.
Since $\vs$ has to be connected to all $3k+2(m-k)$ contacts, we have $2x+3(m-x) \geq 3k+2(m-k)$, which gives us $x \leq m-k$.
However, since $\vs$ can connect to nodes from $W$ in layer $\alpha_i$ if and only if it connects to at most two nodes in $\alpha_i$, we also have $2x \geq 2(m-k)$.
Hence, we have $x=m-k$, \ie, $\vs$ is connected with all nodes from $W$ in $m-k$ layers from $\{\alpha_1, \ldots, \alpha_m\}$.
Therefore, in the remaining $k$ layers from $\{\alpha_1, \ldots, \alpha_m\}$, the evader $\vs$ has to connect to all $3k$ nodes from $U$.
Since the evader cannot connect to more than three nodes in any layer $\alpha_i$, all these sets of neighbors from $U$ have to be disjoint, thus forming the solution to the given instance of the Exact 3-Set Cover problem.
This concludes the proof.
\end{proof}

\begin{figure*}[t]
\centering
\includegraphics[width=.7\linewidth]{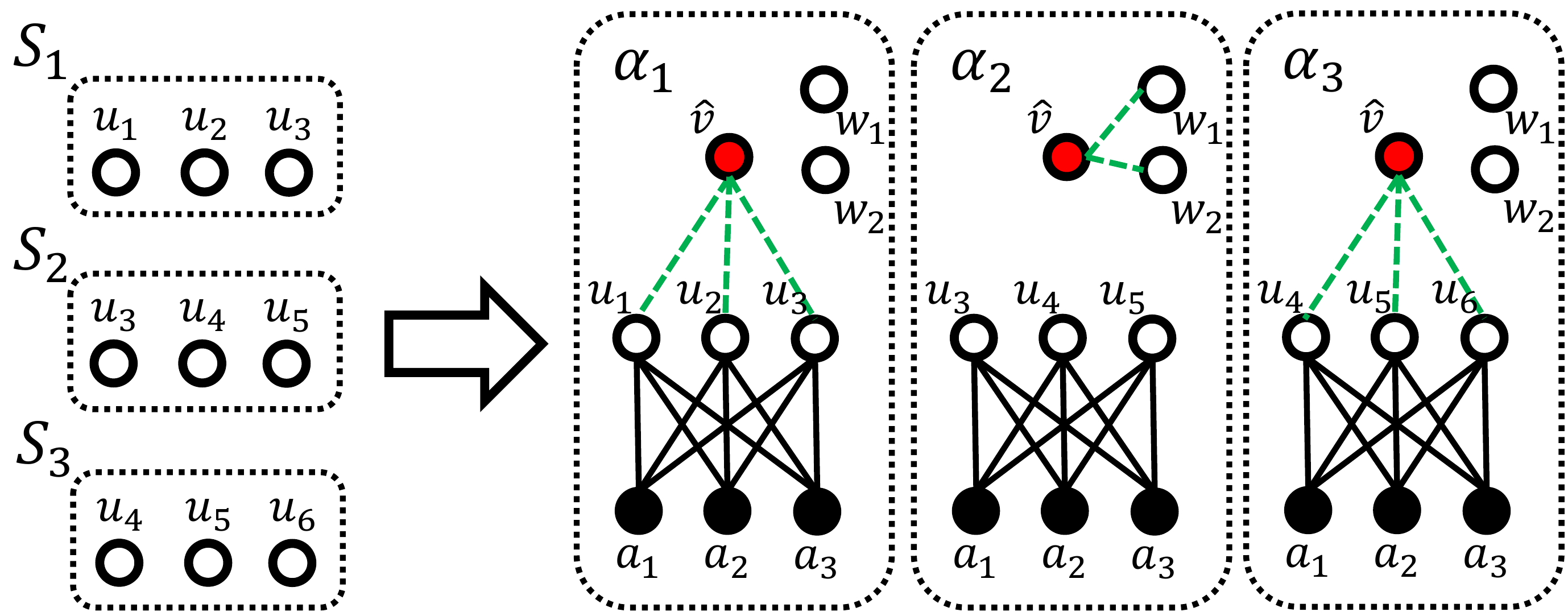}
\caption{An illustration of the network used in the proof of Theorem~\ref{thrm:npc-degree-local}.
The red node represents the evader, while the white nodes represent the contacts.
Dashed (green) edges represent the solution to this problem instance.}
\label{fig:npc-degree-local}
\end{figure*}

\begin{theorem}
\label{thrm:npc-closeness-local}
The problem of Multilayer Local Hiding problem is NP-complete given the closeness centrality measure.
\end{theorem}

\begin{proof}

The problem is trivially in NP, since after the addition of a given $\Add$ the closeness centrality rankings for all layers can be computed in polynomial time.

Next, we prove that the problem is NP-hard.
To this end, we show a reduction from the NP-complete problem of \textit{Exact 3-Set Cover}.
The decision version of this problem is defined by a set of subsets $S=\{S_1, \ldots, S_{m}\}$ of universe $U=\{u_1,\ldots,u_{3k}\}$, such that $\forall_i |S_i|=3$.
The goal is to determine whether there exist $k$ pairwise disjoint elements of $S$ the sum of which equals $U$.

Given an instance of the problem of \textit{Exact 3-Set Cover}, let us construct a multilayer network, $M=(V_L, E_L, V', L)$, as follows (Figure~\ref{fig:npc-closeness-local} depicts an instance of this network):

\begin{itemize}[leftmargin=*]
\item \textbf{The set of nodes $V'$:}
For every $u_i \in U$ we create a node $u_i$. In addition, we create the nodes $w_1,\ldots,w_{2(m-k)}$ and $a_1,\ldots,a_{2(m-k)}$.
Finally, we create the evader node $\vs$ and $5$ nodes $c_1,\ldots,c_5$.
We will denote the set of all nodes $a_i$ by $A$, the set of all nodes $c_i$ by $C$, the set of all nodes $u_i$ by $U$, and the set of all nodes $w_i$ by $W$.
\item \textbf{The set of layers $L$:}
For every $S_i \in S$ we add a layer $\alpha_i$.
\item \textbf{The set of occurrences of nodes in layers $V_L$:}
Node $u_j \in U$ appears in layer $\alpha_i$ if and only if $u_j \in S_i$.
The evader $\vs$, as well as all nodes in $A$, $C$, and $W$, appear in all layers.
\item \textbf{The set of edges $E_L$:}
In all layers we connect every node $w_i$ with the node $a_i$, and we create edges $(c_1,c_2),(c_1,c_3),(c_1,c_4),(c_4,c_5)$.
\end{itemize}

\begin{figure*}[t]
\centering
\includegraphics[width=.7\linewidth]{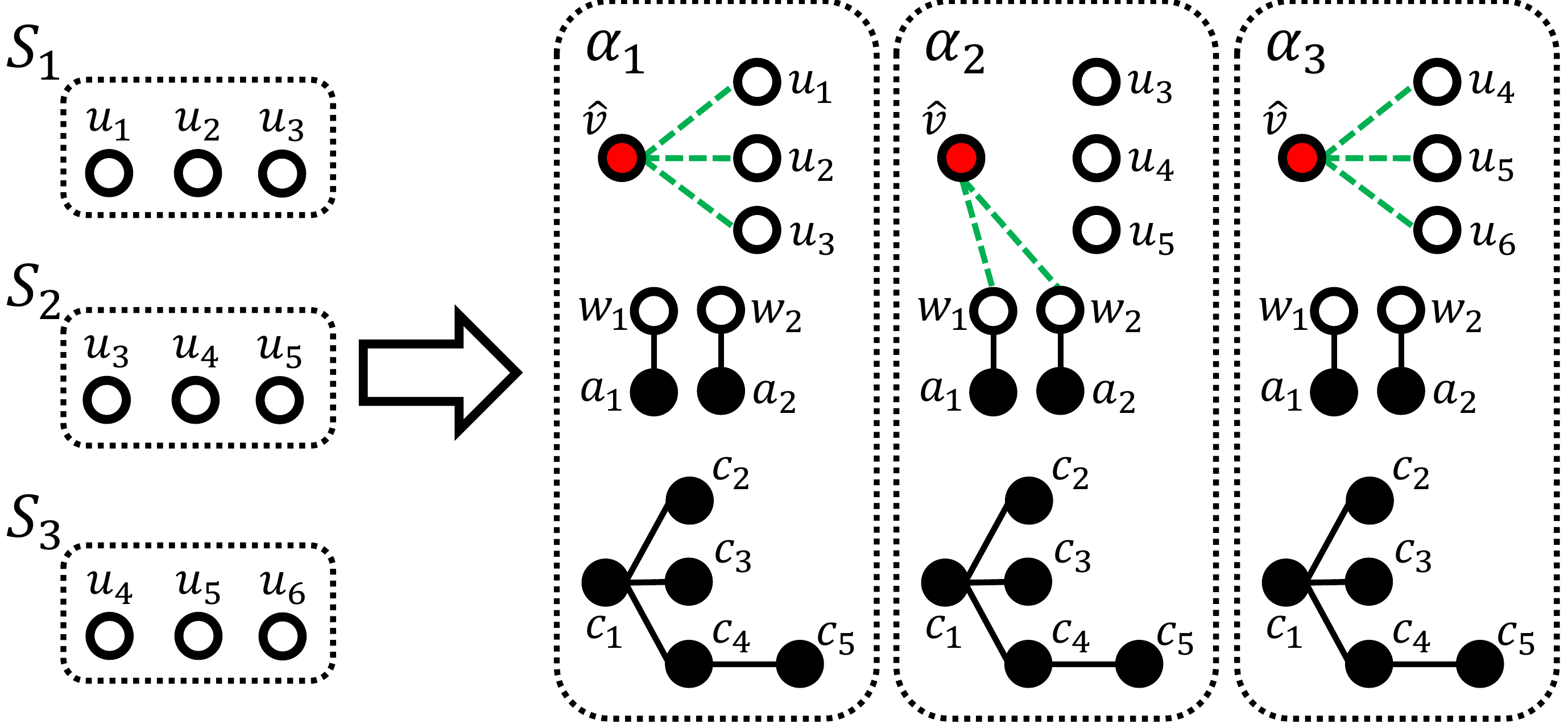}
\caption{An illustration of the network used in the proof of Theorem~\ref{thrm:npc-closeness-local}.
The red node represents the evader, while white the nodes represent the contacts.
Dashed (green) edges represent the solution to this problem instance.}
\label{fig:npc-closeness-local}
\end{figure*}

Now, consider the following instance of the problem of Multilayer Local Hiding, $(M,\vs,\F,c,\left(d^\alpha\right)_{\alpha \in L})$, where:

\begin{itemize}[leftmargin=*]
\item $M$ is the multilayer network we just constructed;
\item $\vs$ is the evader;
\item $\F=U \cup W$ is the set of contacts;
\item $c$ is the closeness centrality measure;
\item $d^{\alpha_i} = 1$ for every $\alpha_i \in L$.
\end{itemize}

Next, let us consider what are the sets of edges that can be added between the evader $\vs$ and the contacts $\F$ in each layer, so that the evader is hidden.
Notice that closeness centrality of the node $c_1$ is $3\frac{1}{2}$ and it is not affected by the edges added to $\vs$.
Assume that we connect node $\vs$ with $x$ nodes from $U$ and $y$ nodes from $W$.
We then have the following (for easier comparison we express the centrality values as fractions with the common denominator 6):
\begin{itemize}[leftmargin=*]
\item $c_{clos}(\vs)=x+\frac{3y}{2}=\frac{6x+9y}{6}$;
\item $c_{clos}(w_i)=\frac{x}{2}+\frac{5y}{6}+\frac{7}{6}=\frac{3x+5y+7}{6}$ if $w_i \in N(\vs)$;
\item $c_{clos}(c_1)=\frac{7}{2}=\frac{21}{6}$;
\end{itemize}
No other node can have greater closeness centrality than $\vs$.
We can connect $\vs$ with at most two of any of the contacts, as node $c_1$ will still have greater closeness centrality.
If we want to connect $\vs$ with three contacts, these contacts have to be nodes from $U$.
If $x+y=3$ and $y>0$, or if $x+y>3$, then the closeness centrality of $\vs$ is the highest in the network, meaning that $\vs$ is not hidden.

Now we will show that if there exists a solution to the given instance of the Exact 3-Set Cover problem, then there also exists a solution to the constructed instance of the Multilayer Local Hiding problem.
Let $S^*$ be an exact cover of $U$.
For every $S_i \in S^*$ we connect $\vs$ to every $u_j \in S_i$ in layer $\alpha_i$.
This way, $\vs$ becomes connected to all $3k$ contacts from $U$, since all the sets in $S^*$ are pairwise disjoint.
For every $S_i \notin S^*$ we connect $\vs$ to two nodes from $W$ in layer $\alpha_i$ (since there are $m-k$ such layers, we can connect $\vs$ to all $2(m-k)$ contacts from $W$ this way).

To complete the proof, we have to show that if there exists a solution $\Add$ to the constructed instance of the Multilayer Local Hiding problem, then there also exists a solution to the given instance of the Exact 3-Set Cover problem.
Let $z$ be the number of layers from $\{\alpha_1, \ldots, \alpha_m\}$ where $\vs$ has at most two neighbors, and let $z-x$ be the number of layers from $\{\alpha_1, \ldots, \alpha_m\}$ where $\vs$ has exactly three neighbors.
Since we have to connect $\vs$ to all $3k+2(m-k)$ contacts, we have $2z+3(m-z) \geq 3k+2(m-k)$, which gives us $z \leq m-k$.
However, since $\vs$ can connect to nodes from $W$ in layer $\alpha_i$ if and only if it connects to at most two nodes in $\alpha_i$, we also have $2z \geq 2(m-k)$.
Hence, we have $z=m-k$, \ie, $\vs$ connects to all nodes from $W$ in $m-k$ layers from $\{\alpha_1, \ldots, \alpha_m\}$.
Therefore, in the remaining $k$ layers from $\{\alpha_1, \ldots, \alpha_m\}$, the evader $\vs$ has to connect with all $3k$ nodes from $U$.
Since the evader cannot connect to more than three nodes in any layer $\alpha_i$, all these sets of neighbors from $U$ have to be disjoint, thus forming a solution to the given instance of the Exact 3-Set Cover problem.
This concludes the proof.
\end{proof}

\begin{theorem}
\label{thrm:npc-betweenness-local}
The problem of Multilayer Local Hiding is NP-complete given the betweenness centrality measure.
\end{theorem}

\emph{Proof of Theorem~\ref{thrm:npc-betweenness-local}:}
The problem is trivially in NP, since after the addition of a given $\Add$ the betweenness centrality rankings for all layers can be computed in polynomial time.

Next, we prove that the problem is NP-hard.
To this end, we show a reduction from the NP-complete problem of \textit{Finding $k$-Clique}.
The decision version of this problem is defined by a simple network, $G=(V,E)$, and a constant, $k \in \N$.
The goal is then to determine whether there exist $k$ nodes in $G$ that form a clique.

Let us assume that $k < n-1$ (if this assumption does not hold then the solution can be computed in polynomial time). Furthermore, let us assume that $G$ is connected (if this does not hold, the problem can be considered separately for each connected component).
Given an instance of the problem of \textit{Finding $k$-Clique} where $k < n-1$, and given a simple network $G=(V,E)$, let us construct a multilayer network, $M=(V_L, E_L, V', L)$, as follows (Figure~\ref{fig:npc-betweenness-local} depicts an instance of this network):

\begin{figure}[t]
\centering
\includegraphics[width=.7\linewidth]{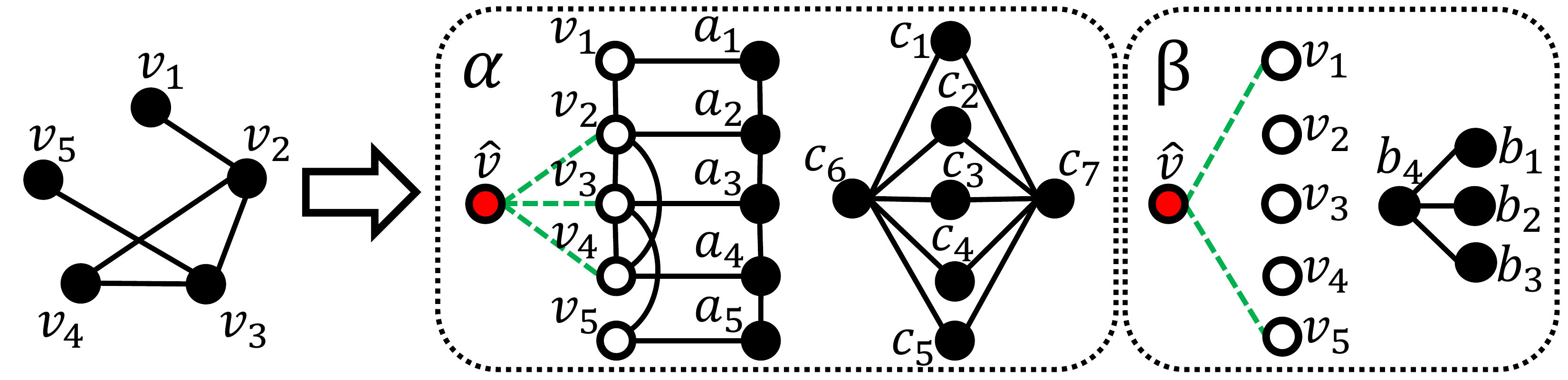}
\caption{An illustration of the network used in the proof of Theorem~\ref{thrm:npc-betweenness-local}.
The red node represents the evader, while the white nodes represent the contacts.
Dashed (green) edges represent the solution to this problem instance.}
\label{fig:npc-betweenness-local}
\end{figure}

\begin{itemize}[leftmargin=*]
\item \textbf{The set of nodes $V'$:}
This consists of the following sets of nodes: $V=\{v_1,\ldots,v_n\}$, $A=\{a_1,\ldots,a_n\}$, $B=\{b_{1}, \ldots, b_{n-k+2}\}$, and $C=\{c_1,\ldots,c_{n+2}\}$.  
\item \textbf{The set of layers $L$:}
We create two layers, $\alpha$ and $\beta$.
\item \textbf{The set of occurrences of nodes in layers $V_L$:}
Layer $\alpha$ contains all nodes in $\{\vs\} \cup V \cup A \cup C$, while layer $\beta$ contains all nodes in $\{\vs\} \cup V \cup B$.
\item \textbf{The set of edges $E_L$:}
\textbf{In layer $\alpha$} we create an edge between two nodes $v_i,v_j \in V$ if and only if this edge was present in $G$.
We also create an edge $(v_i,a_i)$ for every $v_i$, and an edge between every pair $a_i,a_{i+1}$.
Finally, for every node $c_i\in C:i \leq n$, we create edges $(c_i,c_{n+1})$ and $(c_i,c_{n+2})$.
\textbf{In layer $\beta$} we create an edge $(b_i,b_{n-k+2})$ for every node $b_{i}\in B:i < n-k+2$.
\end{itemize}

Now, consider the following instance of the problem of Multilayer Local Hiding, $(M,\vs,\F,c,\left(d^\alpha\right)_{\alpha \in L})$, where: $M$ is the multilayer network we just constructed; $\vs$ is the evader; $\F=V$ is the set of contacts;
$c$ is the betweenness centrality measure;
and $d^\alpha = 3n+2$ and $d^\beta = 1$ are the safety margins. Given this, let us consider what are the sets of edges that can be added between the evader $\vs$ and the contacts $\F$ in each layer, so that the evader is hidden.

Since $d^\alpha = 3n+2$, then apart from the evader $\vs$, the betweenness centrality of every node in layer $\alpha$ must be greater than that of $\vs$; otherwise the evader $\vs$ would not be hidden in $\alpha$.
Also note that the betweenness centrality of every node $c_i\in C: i \leq n$ equals $\frac{1}{n}$, and all nodes other than $\vs$ have non-zero betweenness centrality. 

Now if $\vs$ gets connected to any two nodes $v_i,v_j\in V$ that are not connected to one another, then $\vs$ controls one shortest path of length 2 between $v_i$ and $v_j$. Note that there can be at most $n-2$ other shortest paths of length 2 between $v_i$ and $v_j$ (each such path goes through some node $v_k\in V\setminus\{v_i,v_j\}$ if and only if $v_k$ is connected to both $v_i$ and $v_j$). Thus, the betweenness centrality of $\vs$ is at least $\frac{1}{n-1}$. Consequently, all nodes that $\vs$ is connected to in layer $\alpha$ must form a clique in order for $\vs$ to be hidden in $\alpha$.

Consider a situation in which the evader $\vs$ is connected to $x$ nodes from $V$ in layer $\beta$ (notice that $x \leq n$).
Its betweenness centrality is then $\frac{x(x-1)}{2}$, as it controls all shortest paths between pairs of its neighbors, but not any other shortest paths.
At the same time, the betweenness centrality of the node $b_{n-k+2}$ is $\frac{(n-k+1)(n-k)}{2}$ (as it controls all shortest paths between pairs of other nodes from $B$), which is greater than the betweenness centrality of $\vs$ if and only if $x \leq n-k$.
All other nodes in the layer have betweenness centrality $0$.
Thus, $\vs$ is hidden in $\beta$ iff it has at most $n-k$ neighbors.

Now we will show that if there exists a solution to the given instance of the problem of Finding $k$-Clique, then there also exists a solution to the constructed instance of the problem of Multilayer Local Hiding.
To this end, let $V^*$ be a group of $k$ nodes forming a clique in $G$.
Let us create $\Add$ by connecting $\vs$ to nodes from $V^*$ in layer $\alpha$ and to nodes from $\F \setminus V^*$ in layer $\beta$.
As argued above, for such $\Add$, the evader $\vs$ is hidden in both layers, hence $\Add$ is a solution to the constructed instance of the Multilayer Local Hiding problem.

To complete the proof we have to show that if there exists a solution $\Add$ to the constructed instance of the problem of Multilayer Local Hiding, then there also exists a solution to the given instance of the problem of Finding $k$-Clique.
Since $\vs$ can be connected in layer $\beta$ to at most $n-k$ nodes from $V$, it has to have at least $k$ neighbors from $V$ in layer $\alpha$
As shown above, in order for $\vs$ to be hidden in $\alpha$, all of its neighbors must form a clique.
Hence, the neighbors of $\vs$ in layer $\alpha$ form a clique in $G$.
This concludes the proof.\hspace*{\fill}$\Box$

\subsection{Approximation Analysis}

In this section we present the analysis of optimization versions of our problems (see Table~\ref{tab:approximation} for a summary).

\begin{table*}[t]
\caption{Summary of our results regarding approximation algorithms.}\smallskip
\centering
\resizebox{.95\linewidth}{!}{
\smallskip\begin{tabular}{lcc}
Centrality	& Maximum Multilayer Global Hiding	& Maximum Multilayer Local Hiding \\
\hline
Degree & can be solved in polynomial time & greedy algorithm is $2$-approximation	\\
Closeness	& - & cannot be approximated within $|\F|^{1-\epsilon}$ for any $\epsilon > 0$	\\
Betweenness	& cannot be approximated within $|\F|^{1-\epsilon}$ for any $\epsilon > 0$				& cannot be approximated within $|\F|^{1-\epsilon}$ for any $\epsilon > 0$	\\
\end{tabular}
}
\label{tab:approximation}
\end{table*}

\begin{theorem}
\label{thrm:appr-degree-global}
The Maximum Multilayer Global Hiding problem can be solved in polynomial time.
\end{theorem}

\begin{proof}
For a given $k \in \N$ it is possible to connect the evader with $k$ contacts if and only if
$\min(k,|\{v \in \F : |N(v)| = k + |N(\vs)|\}|) + |\{v \in V: |N(v)| > k + |N(\vs)|\}| \geq d$.
It is because the only nodes that count towards satisfying the safety margin are those that already have degree greater than $k + |N(\vs)|$, or the contacts that have degree $k$ and their degree will be increased to $k+|N(\vs)|+1$ when they are connected with the evader (notice that since we are adding $k$ edges, there can be at most $k$ such contacts).
\end{proof}

\begin{theorem}
\label{thrm:inappr-closeness-global}
Maximum Multilayer Global Hiding problem given the closeness centrality cannot be approximated within $|\F|^{1-\epsilon}$ for any $\epsilon > 0$, unless P=NP.
\end{theorem}

\begin{figure*}[t]
\centering
\includegraphics[width=\linewidth]{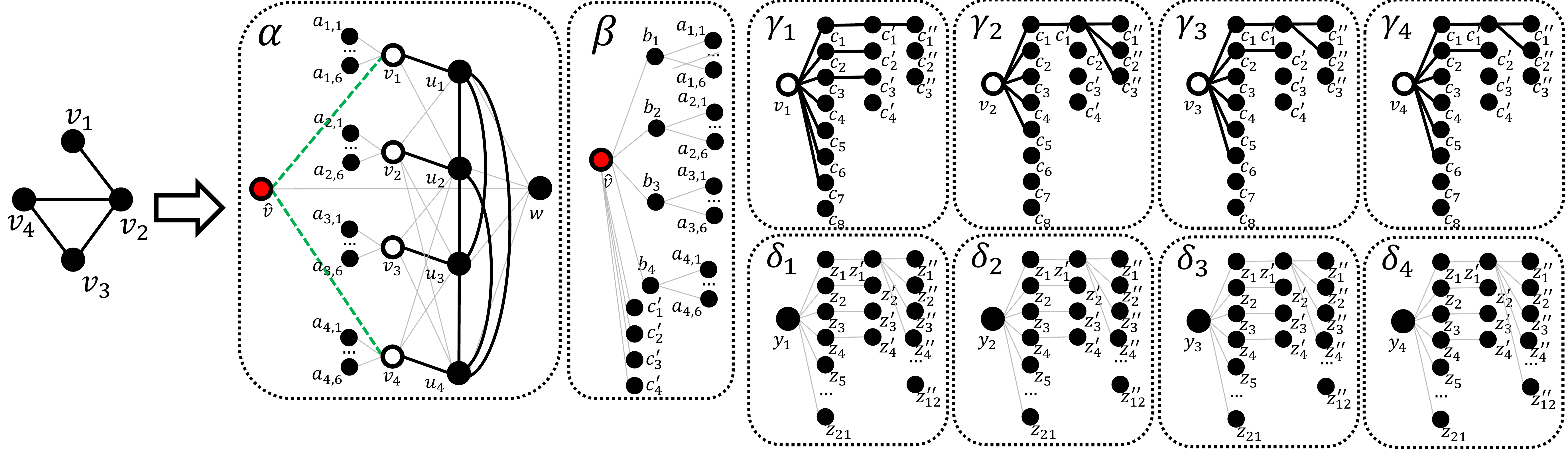}
\caption{An illustration of the network used in the proof of Theorem~\ref{thrm:inappr-closeness-global}.
The red node represents the evader, while the white nodes represent the contacts.
Dashed (green) edges represent the optimal solution to this problem instance.}
\label{fig:inappr-closeness-global}
\end{figure*}

\begin{proof}
In order to prove the theorem, we will use the result by Zuckerman~\cite{zuckerman2006linear} that the \textit{Maximum Independent Set} problem cannot be approximated within $|V|^{1-\epsilon}$ for any $\epsilon > 0$, unless $P=NP$ (notice that the Maximum Independent Set problem is equivalent to the Maximum Clique problem on a complementary network).
The Maximum Independent Set problem is defined by a simple network, $G=(V,E)$.
The goal is to identify the maximum (in terms of size) group of nodes in $G$ that are independent, \ie, they do not induce any edges.

First, we will show a function $f(G)$ that based on an instance of the problem of Maximum Independent Set, defined by a simple network $G=(V,E)$, constructs an instance of the Maximum Multilayer Global Hiding.
In what follows we will assume that $n \geq 4$ (the problem can be easily solve in constant time for $n < 4$).

Let a multilayer network, $M=(V_L, E_L, V', L)$, be defined as follows (Figure~\ref{fig:inappr-closeness-global} depicts an instance of this network):

\begin{itemize}[leftmargin=*]
\item \textbf{The set of nodes $V'$:}
For every node, $v_i \in V$, we create a node $v_i$, a node $u_i$, a node $y_i$, six nodes $a_{i,1},\ldots,a_{i,6}$, and a node $b_i$.
We will denote the set of all nodes $v_i$ by $V$, the set of all nodes $u_i$ by $U$, the set of all nodes $y_i$ by $Y$, the set of all nodes $a_{i,j}$ by $A$, and the set of all nodes $b_i$ by $B$.
Additionally, we create the evader node $\vs$, a node $w$ and six sets of nodes $C=\{c_1,\ldots,c_{3n-4}\}$, $C'=\{c'_1,\ldots,c'_n\}$, and $C''=\{c''_1,\ldots,c''_{n-1}\}$, $Z=\{z_1,\ldots,z_{5n+1}\}$, $Z'=\{z'_1,\ldots,z'_n\}$, and $Z''=\{z''_1,\ldots,z''_{3n}\}$.

\item \textbf{The set of layers $L$:}
We create layers $\alpha$, $\beta$, $n$ layers $\gamma_1, \ldots, \gamma_n$, and $n$ layers $\delta_1, \ldots, \delta_n$.

\item \textbf{The set of occurrences of nodes in layers $V_L$:}
Layer $\alpha$ contains occurrences of nodes $\{\vs,w\} \cup A \cup V \cup U$.
Layer $\beta$ contains occurrences of nodes $\{\vs\} \cup A \cup B \cup C'$.
A given layer $\gamma_i$ contains occurrences of nodes $\{v_i\} \cup C \cup C' \cup C''$.
A given layer $\delta_i$ contains occurrences of nodes $\{y_i\} \cup Z \cup Z' \cup Z''$.

\item \textbf{The set of edges $E_L$:}
\begin{itemize}
\item \textbf{In layer $\alpha$:} for every pair of nodes $v_i \in V, u_j \in U$ we create an edge $(v_i,u_j)$ if and only if the edge $(v_i,v_j)$ was present in $G$.
For every node $u_i$ we create edges $(u_i,v_i)$ and $(u_i,w)$.
For every node $a_{i,j}$ we create an edge $(a_{i,j},v_i)$.
We also create an edge between every pair of nodes $u_i,u_j \in U$.
Finally, we create an edge $(\vs,w)$.

\item \textbf{In layer $\beta$:} for every node $a_{i,j}$ we create an edge $(a_{i,j},b_i)$.
We also create an edge $(\vs,b_i)$ for every node $b_i$ and an edge $(\vs,c'_i)$ for every node $c'_i$.

\item \textbf{In a given layer $\gamma_i$:} for every node $c_j$ we create an edge $(v_i,c_j)$ if and only if $j \leq 3n-4-|N_G(v_i)|$ (\ie, we connect $v_i$ with $3n-4-|N_G(v_i)|$ first nodes from $C$).
For every node $c'_j$ we create an edge $(c'_j,c_j)$ if and only if $j \leq n-|N_G(v_i)|$ (\ie, we connect $n-|N_G(v_i)|$ first nodes from $C'$ with their $C$ counterparts).
For every node $c''_j$ we create an edge $(c''_j,c'_1)$ if and only if $j \leq |N_G(v_i)|$ (\ie, we connect $|N_G(v_i)|$ first nodes from $C''$ with the node $c'_1$).

\item \textbf{In a given layer $\delta_i$:} for every node $z_j$ we create an edge $(y_i,z_j)$.
For every node $z'_j$ we create an edge $(z'_j,z_j)$.
For every node $z''_j$ we create an edge $(z''_j,z'_1)$ if and only if $j \leq n+2i$ (\ie, we connect $n+2i$ first nodes from $Z''$ with the node $z'_1$).
\end{itemize}

\end{itemize}

To complete the constructed instance of the problem let:

\begin{itemize}[leftmargin=*]
\item $\vs$ be the evader;
\item $\F=V$ be the set of contacts;
\item $c$ be the closeness centrality measure;
\item $d = n$ be the safety margin.
\end{itemize}

Hence, the formula of the function $f$ is $f(G)=(M,\vs,\F,c,d)$.
Let $\Add$ be the solution to the constructed instance of the Maximum Multilayer Global Hiding problem.
The function $g$ computing corresponding solution to the instance $G$ of the Maximum Independent Set problem is now $g(\Add)= \{v_i \in V : (\vs^\alpha,v_i^\alpha) \in \Add\}$, \ie, the nodes forming the independent set are the contacts that the evader is connected to.

Now, we will show that $g(\Add)$ is indeed a correct solution to $G$, \ie, that the nodes form an independent set.
Let $x_i = |N_G(v_i) \cap g(\Add)|$, \ie, the number of neighbours of $v_i$ in $G$ connected to $\vs$.
Let us compute the closeness centrality of all nodes in the network and compare it with the closeness centrality of the evader:

\begin{itemize}
\item $c_{clos}(M,\vs) = 2n+1+|\Add| + \frac{6n+n}{2} + \frac{n-|\Add|}{3} = 5\frac{5}{6}n+\frac{2}{3}|\Add|+1$;
\item if $v_i \in g(\Add)$ then $c_{clos}(M,v_i) = 3n+4 + \frac{2n-|N_G(v_i)|+|\Add|-x_i-1}{2} + \frac{n+6|N_G(v_i)|+5|\Add|-5x_i-6}{3} + \frac{6n-6|N_G(v_i)|-6|\Add|+6x_i}{4} = 5\frac{5}{6}n + \frac{2}{3}|\Add| + 1\frac{1}{2} - \frac{2}{3}x_i$;
\item if $v_i \notin g(\Add)$ then $c_{clos}(M,v_i) = 3n+3 + \frac{2n-|N_G(v_i)|}{2} + \frac{n+6|N_G(v_i)|}{3} + \frac{6n-6|N_G(v_i)|-6}{4} = 5\frac{5}{6}n + 1\frac{1}{2}$, hence $c_{clos}(M,v_i) < c_{clos}(M,\vs)$ for $|\Add| > 0$;
\item $c_{clos}(M,w) = n+1 + \frac{n}{2} + \frac{6n}{3} = 3\frac{1}{2}n+1 < c_{clos}(M,\vs)$;
\item $c_{clos}(M,u_i) \leq 2n + \frac{6n}{2} = 5n < c_{clos}(M,\vs)$;
\item $c_{clos}(M,a_{i,j}) \leq 2 + \frac{n+7}{2} + \frac{2n-2}{3} + \frac{6(n-1)}{4} = 2\frac{2}{3}n + 3\frac{1}{3}	< c_{clos}(M,\vs)$;
\item $c_{clos}(M,b_i) = 7 + \frac{n-1}{2} + \frac{(n-1)6}{3} = 2\frac{1}{2}n + 4\frac{1}{2} < c_{clos}(M,\vs)$;
\item $c_{clos}(M,c_i) \leq n + 1 + \frac{4n-6}{2} + \frac{n-1}{3} = 3\frac{1}{3}n - 2\frac{1}{3} < c_{clos}(M,\vs)$;
\item $c_{clos}(M,c'_i) \leq n + 1 + \frac{3n-1}{2} + \frac{8n-4}{3} + \frac{n-1}{4} = 5\frac{5}{12}n - 1\frac{1}{12} < c_{clos}(M,\vs)$;
\item $c_{clos}(M,c''_i) \leq 1 + \frac{n-1}{2} + \frac{n}{3} + \frac{3n-5}{4} + \frac{n-1}{5} = 1\frac{47}{60}n - 1\frac{19}{20} < c_{clos}(M,\vs)$;
\item $c_{clos}(M,y_i) = 5n + 1 + \frac{n}{2} + \frac{n+2i}{3} = 5\frac{5}{6}n + 1 + \frac{2}{3}i$;
\item $c_{clos}(M,z_i) \leq n + 1 + \frac{8n}{2} + \frac{n-1}{3} = 5\frac{1}{3}n + \frac{2}{3} < c_{clos}(M,\vs)$;
\item $c_{clos}(M,z'_i) \leq 3n+1 + \frac{n}{2} + \frac{5n}{3} + \frac{n-1}{4} = 5\frac{5}{12}n + \frac{3}{4} < c_{clos}(M,\vs)$;
\item $c_{clos}(M,z''_i) = 1 + \frac{3n}{2} + \frac{n}{3} + \frac{5n}{4} + \frac{n-1}{5} = 3\frac{17}{60}n + \frac{4}{5} < c_{clos}(M,\vs)$.
\end{itemize}

Notice that the only nodes that can have greater closeness centrality than the evader are the nodes in $V$ that the evader is connected to and the nodes in $Y$.
For a given node $y_i$ it has a greater closeness centrality score than the evader if and only if the evader is connected to less than $i$ nodes from $V$ (\ie, when $|\Add| < i$), as:
$$
c_{clos}(M,\vs) - c_{clos}(M,y_i) = \frac{2}{3}(|\Add|-i).
$$
Hence, since the safety margin is $d=n$ and the only other nodes that can have greater closeness centrality than the evader are the nodes in $V$ that the evader is connected to.
Hence, every node in $V$ that the evader connects to must have greater closeness centrality than the evader, in order for the safety margin to be maintained (there are exactly $n$ nodes in $Y$, and every edge additional edge in $\Add$ causes the evader to get greater closeness centrality than one of the nodes in $Y$).
However, node $v_i \in V$ that the evader is connected to has greater closeness centrality than the evader if and only if $x_i=0$, \ie, no neighbors of $v_i$ are connected to $\vs$, as:
$$
c_{clos}(M,\vs) - c_{clos}(M,v_i) = \frac{2}{3}x_i - \frac{1}{2}.
$$
This implies that, if $\vs$ is hidden then all nodes from $V$ that are connected to $\vs$ must form an independent set.

Therefore, the optimal solution to the constructed instance of the Maximum Multilayer Global Hiding problem is returning nodes from $V$ forming in $G$ an independent set of the maximum size.
Hence, the optimal solution corresponds to the optimal solution to the given instance of the Maximum Independent Set problem.

Now, assume that there exists an approximation algorithm for the Maximum Multilayer Global Hiding problem with ratio $|\F|^{1-\epsilon}$ for some $\epsilon > 0$.
Let us use this algorithm to solve the constructed instance $f(G)$, acquiring solution $\Add$. and consider solution $g(\Add)$ to the given instance of the Maximum Clique problem.
Since the size of the optimal solution is the same for both instances, we obtained an approximation algorithm that solves Maximum Independent Set problem to within $|V|^{1-\epsilon}$ for $\epsilon > 0$.
However, Zuckerman~\cite{zuckerman2006linear} shown that the Maximum Independent Set problem cannot be approximated within $|V|^{1-\epsilon}$ for any $\epsilon > 0$, unless $P=NP$.
Therefore, such approximation algorithm for the Maximum Multilayer Global Hiding problem cannot exist, unless $P=NP$.
This concludes the proof.
\end{proof}

\begin{figure*}[t]
\centering
\includegraphics[width=.6\linewidth]{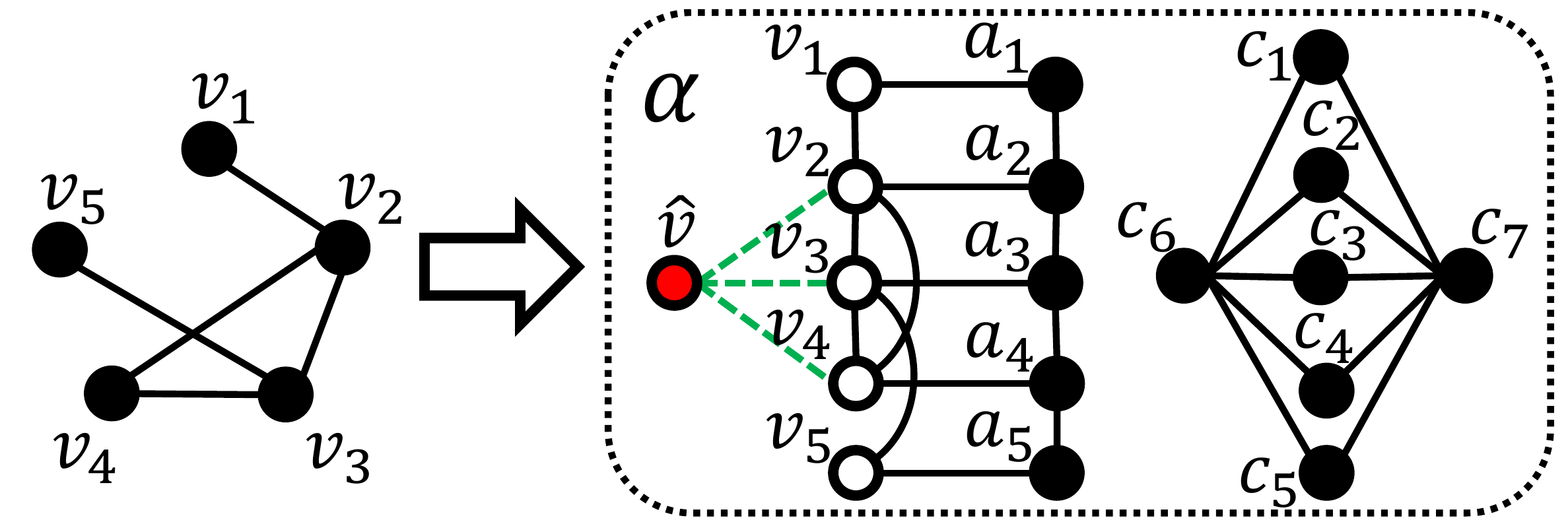}
\caption{An illustration of the network used in the proof of Theorem~\ref{thrm:inappr-betweenness}.
The red node represents the evader, while the white nodes represent the contacts.
Dashed (green) edges represent the optimal solution to this problem instance.}
\label{fig:inappr-betweenness}
\end{figure*}

\begin{theorem}
\label{thrm:inappr-betweenness}
Both Maximum Multilayer Global Hiding and Maximum Multilayer Local Hiding problems given the betweenness centrality cannot be approximated within $|\F|^{1-\epsilon}$ for any $\epsilon > 0$, unless P=NP.
\end{theorem}

\begin{proof}
In order to prove the theorem, we will use the result by Zuckerman~\cite{zuckerman2006linear} that the \textit{Maximum Clique} problem cannot be approximated within $|V|^{1-\epsilon}$ for any $\epsilon > 0$, unless $P=NP$.
The Maximum Clique problem is defined by a simple network, $G=(V,E)$.
The goal is to identify the maximum (in terms of size) group of nodes in $G$ that form a clique.

First, we will show a function $f(G)$ that based on an instance of the problem of Maximum Clique, defined by a simple network $G=(V,E)$, constructs either an instance of the Maximum Multilayer Global Hiding or an instance of the Maximum Multilayer Local Hiding.

Let a multilayer network, $M=(V_L, E_L, V', L)$, be defined as follows (Figure~\ref{fig:inappr-betweenness} depicts an instance of this network):

\begin{itemize}[leftmargin=*]
\item \textbf{The set of nodes $V'$:}
For every node, $v_i \in V$, we create a node $v_i$ and a node $a_i$.
Additionally, we create the evader node $\vs$ and the set of nodes $C=\{c_1,\ldots,c_{n+2}\}$.

\item \textbf{The set of layers $L$:}
We create only a single layer $\alpha$.

\item \textbf{The set of occurrences of nodes in layers $V_L$:}
All nodes occur in layer $\alpha$.

\item \textbf{The set of edges $E_L$:}
In layer $\alpha$ we create an edge between two nodes $v_i,v_j \in V$ if and only if this edge was present in $G$.
We also create an edge $(v_i,a_i)$ for every $v_i$, and an edge between every pair $a_i,a_{i+1}$.
Finally, for every node $c_i\in C:i \leq n$, we create edges $(c_i,c_{n+1})$ and $(c_i,c_{n+2})$.
\end{itemize}

To complete the constructed instance of the problem let:

\begin{itemize}[leftmargin=*]
\item $\vs$ be the evader;
\item $\F=V$ be the set of contacts;
\item $c$ be the betweenness centrality measure;
\item $d = 3n+2$ be the safety margin in the global version;
\item $d^\alpha = 3n+2$ be the safety margin in the local version.
\end{itemize}

Hence, the formula of the function $f$ is $f(G)=(M,\vs,\F,c,d)$ for the global version of the problem and $f(G)=(M,\vs,\F,c,\left(d^\alpha\right)_{\alpha \in L})$ for the local version of the problem.
Notice that since network $M$ has only one layer, both problems are equivalent.
In the following we will focus on the global version of the problem.

Let $\Add$ be the solution to the constructed instance of the Maximum Multilayer Global Hiding problem.
The function $g$ computing corresponding solution to the instance $G$ of the Maximum Clique problem is now $g(\Add)= \{v \in V : (\vs,v) \in \Add\}$, \ie, the nodes forming the clique are the contacts that the evader is connected to.

Now, we will show that $g(\Add)$ is indeed a correct solution to $G$, \ie, that the nodes form a clique.
Notice that, since $d = 2n+2k+3$, all other nodes must have greater betweenness centrality than the evader in order for $\vs$ to be hidden.
Notice also that the betweenness centrality of every node $c_i$ for $i \leq n$ is $\frac{1}{n}$.
Moreover, after adding $\Add$ all nodes other than $\vs$ have non-zero betweenness centrality.
If $\vs$ gets connected to at least two nodes from $\F$ that are not connected to each other, then $\vs$  controls one of at most $n-1$ shortest path between them (other paths can only go through nodes in $V$) and thus the betweenness centrality of $\vs$ is at least $\frac{1}{n-1}$.
Therefore, in order to get hidden, $\vs$ cannot control any shortest paths in the network.
This implies that, if $\vs$ is hidden then all nodes that are connected to $\vs$ must form a clique.

Therefore, the optimal solution to the constructed instance of the Maximum Multilayer Global Hiding problem is returning nodes from $V$ forming a clique of maximum size.
Since the structure of connections between the nodes $V$ is the same as in the network $G$, the optimal solution corresponds to the optimal solution to the given instance of the Maximum Clique problem.

Now, assume that there exists an approximation algorithm for the Maximum Multilayer Global Hiding problem with ratio $|\F|^{1-\epsilon}$ for some $\epsilon > 0$.
Let us use this algorithm to solve the constructed instance $f(G)$, acquiring solution $\Add$. and consider solution $g(\Add)$ to the given instance of the Maximum Clique problem.
Since the size of the optimal solution is the same for both instances, we obtained an approximation algorithm that solves Maximum Clique problem to within $|V|^{1-\epsilon}$ for $\epsilon > 0$.
However, Zuckerman~\cite{zuckerman2006linear} shown that the Maximum Clique problem cannot be approximated within $|V|^{1-\epsilon}$ for any $\epsilon > 0$, unless $P=NP$.
Therefore, such approximation algorithm for the Maximum Multilayer Global Hiding problem cannot exist, unless $P=NP$.
This concludes the proof.
\end{proof}

\begin{theorem}
\label{thrm:2appr-degree-local}
The greedy algorithm is a $2$-approximation for the Maximum Multilayer Local Hiding problem given the degree centrality.
The bound is tight.
\end{theorem}

\begin{proof}

First, let us analyze the structure of a solution to the Maximum Multilayer Local Hiding problem given the degree centrality.
Let $\delta^\alpha$ be the degree of the $d^\alpha$-th node in the degree centrality ranking of the nodes in $V^\alpha$, let $\delta_0^\alpha$ be the initial (\ie, before any edges to the contacts are added) degree of the evader in layer $\alpha$, and let $\F^\alpha$ be the set of occurrences of contacts in layer $\alpha$, \ie, $\F^\alpha = \{v^\alpha : v \in \F\}$.
An algorithm solving the Maximum Multilayer Local Hiding problem can either:

\begin{enumerate}[label=\alph*)]
\item \label{pt:greedy-1} connect the evader to at most $k^\alpha = \delta^\alpha - 1 - \delta_0^\alpha$ of freely selected nodes from $\F^\alpha$, as this way the degree of the evader is increased to at most $\delta^\alpha - 1$, and the nodes from the first $d^\alpha$ positions of the degree ranking before the addition continue to have greater degree than the evader when the new edges are added;

\item \label{pt:greedy-2} connect the evader to exactly $\delta^\alpha - \delta_0^\alpha$ nodes from $\F^\alpha$ (notice that $\delta^\alpha - \delta_0^\alpha = k^\alpha + 1$).
This increases the degree of the evader to $\delta^\alpha$, hence the new connections must include at least $d^\alpha - |\{v^\alpha \in V^\alpha : |N^\alpha(v)| > \delta^\alpha \}|$ nodes with degree exactly $\delta^\alpha$.
As a result, there will now exist $d^\alpha$ nodes with degree at least $\delta^\alpha+1$ and the safety margin will be maintained.
\end{enumerate}

First, notice that the sets of potential connections in both \ref{pt:greedy-1} and \ref{pt:greedy-2} can be easily computed in polynomial time, hence the greedy algorithm can use them to optimize the choice of edges added in a single layer.

Notice also that the evader can never add more than $k^\alpha + 1$ edges in layer $\alpha$, as her degree will then increase to at least $\delta^\alpha+2$.
Since adding a set of connections between the evader and the contacts cannot increase the degree of any contact by more than one, the $d^\alpha$-th node in the degree centrality ranking of the nodes in $V^\alpha$ will have degree at most $\delta^\alpha+1$.
Hence, the safety margin cannot be maintained.

Finally, notice that if $k^\alpha < 0$, then the degree of the evader is at least $\delta^\alpha$ before adding any edges, which puts her within the top $d^\alpha$ positions of the degree centrality ranking.
Since increasing the degree of any other nodes can be realized only by adding an edge to the evader (which in turn increases the evader's degree even more), the problem does not have a solution if $k^\alpha < 0$ for any layer $\alpha$.

The greedy algorithm iterates over the layers and for each layer it connects the evader with maximum possible number of contacts that the evader has not been connected with yet.
Notice that it is never beneficial to connect the evader with a given contact in more than one layer, hence any solution doing so has an equivalent solution without the redundant edge(s).
In what follows, we will only consider solutions without the redundant edges.

Let us now compare a solution $A^\$$ returned by the greedy algorithm with an optimal solution $\Add$.
We will denote by $A^\$_\alpha$ the set of contacts connected to the evader by the greedy algorithm in layer $\alpha$, \ie, $A^\$_\alpha = \{v \in V^\alpha : (\vs^\alpha , v^\alpha) \in A^\$\}$, and by $\Add_\alpha$ the set of contacts connected to the evader by the optimal algorithm in layer $\alpha$, \ie, $\Add_\alpha = \{v \in V^\alpha : (\vs^\alpha , v^\alpha) \in \Add\}$.
We iterate over the layers of the network in the same order as the greedy algorithm; let this order be $\alpha_1,\ldots,\alpha_{|L|}$.
Contacts that the optimal solution connects the evader to in a given layer $\alpha_i$ can be grouped into three pairwise disjoint sets:
\begin{itemize}
\item Contacts that could not have been selected in layer $\alpha_i$ by the greedy algorithm, as they were selected by it in one of the previous layers, \ie:
$$
X^{\alpha_i} = \{ v \in \Add_{\alpha_i} : v \notin A^\$_{\alpha_i} \land \exists_{j < i} v \in A^\$_{\alpha_j} \};
$$
\item Contacts that are not selected by the greedy algorithm in layer $\alpha_i$, but they could have been selected, \ie:
$$
Y^{\alpha_i} = \{ v \in \Add_{\alpha_i} : v \notin A^\$_{\alpha_i} \land \neg\exists_{j < i} v \in A^\$_{\alpha_j} \};
$$
\item Contacts that are selected by both the greedy algorithm and the optimal solution in layer $\alpha_i$, \ie:
$$
Z^{\alpha_i} = \{ v \in \Add_{\alpha_i} : v \in A^\$_{\alpha_i} \}.
$$
\end{itemize}

We will  show that $|\Add_{\alpha_i}| - |A^\$_{\alpha_i}| \leq |X^{\alpha_i}|$, \ie, the difference between the number of edges added in layer $\alpha_i$ by the optimal solution and by the greedy algorithm cannot be greater than $|X^{\alpha_i}|$.
We will prove this by contradiction.
To this end, assume that in some layer $\alpha_i$ the said difference is greater than $|X^{\alpha_i}|$, \ie, $|\Add_{\alpha_i}| > |A^\$_{\alpha_i}| + |X^{\alpha_i}|$.
Since $|\Add_{\alpha_i}| = |X^{\alpha_i}| + |Y^{\alpha_i}| + |Z^{\alpha_i}|$, we get that in this layer: 
$|A^\$_{\alpha_i}| < |Y^{\alpha_i}| + |Z^{\alpha_i}|$.
However, since none of the nodes from $Y^{\alpha_i} \cup Z^{\alpha_i}$ were selected by the greedy algorithm in the previous layers, the greedy algorithm would have chosen to connect the evader with contacts from $Y^{\alpha_i} \cup Z^{\alpha_i}$, as it connects the evader with a greater number of nodes in layer $\alpha_i$ than the solution $A^\$$.
Therefore, the difference between the number of edges added in layer $\alpha_i$ by the optimal solution and by the greedy algorithm cannot be greater than $|X^{\alpha_i}|$, \ie, $|\Add_{\alpha_i}| - |A^\$_{\alpha_i}| \leq |X^{\alpha_i}|$. Summing over all layers yields:
$$
\sum_{\alpha_i \in L} |\Add_{\alpha_i}| \leq \sum_{\alpha_i \in L} |A^\$_{\alpha_i}| + \sum_{\alpha_i \in L} |X^{\alpha_i}|.
$$
Since any $v$ is a member of only a single set $X^{\alpha_i}$ (as we assumed that the optimal solution does not contain any redundant edges) and since from the definition of $X^{\alpha_i}$ we have that $\exists_{j < i} v \in A^\$_{\alpha_j}$, we get that $\sum_{\alpha_i \in L} |X^{\alpha_i}| \leq |A^\$|$.
Given that $\sum_{\alpha_i \in L} |\Add_{\alpha_i}| = |\Add|$ and $\sum_{\alpha_i \in L} |X^{\alpha_i}| = |A^\$|$ we get:
$$
|\Add| \leq 2|A^\$|.
$$
Since we consider solution without redundant edges, the size of each solution is equal to the number of contact connected with the evader by each solution.
Therefore, the greedy algorithm is a $2$-approximation.

\begin{figure}[t]
\centering
\includegraphics[width=.7\linewidth]{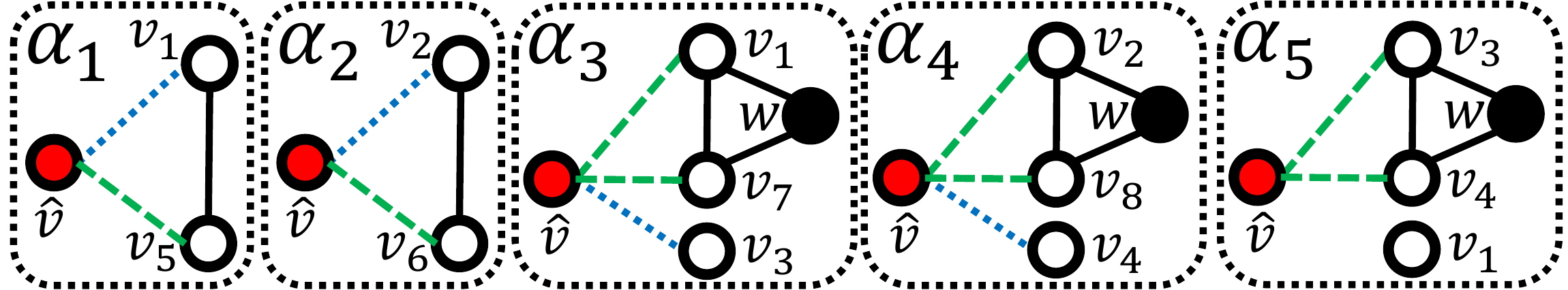}
\caption{An illustration of the network showing the tightness of the bound given in Theorem~\ref{thrm:2appr-degree-local}.
The red node represents the evader, while white the nodes represent the contacts.
Dashed (green) edges represent the optimal solution to this problem instance, while dotted (blue) edges represent the solution returned by the greedy algorithm.}
\label{fig:2appr-degree-local}
\end{figure}

Figure~\ref{fig:2appr-degree-local} presents an example of the network, where the bound is tight, \ie, the optimal solution connects the evader with exactly twice as many contacts as the greedy algorithm.
The green edges represent the optimal solution, connecting the evader with all eight contacts, while the greedy algorithm (the result of which is represented by the blue edges) connects the evader to only four contacts.

\end{proof}

\section{Heuristics \& Empirical Analysis}
\label{sec:simulations}

Given that most computational results are negative, we shift now our attention towards developing heuristic algorithms that provide efficient, albeit not optimal, solutions.

\subsection{Network Datasets}

In our experiments we use both randomly generated and real-life multilayer networks. As for the randomly generated ones, we use the following standard network generation models:

\begin{itemize}
\item \emph{Random graphs}, generated using the Erd{\H{o}}s-R{\'e}nyi model \cite{erdds1959random}.
We denote by ER$(n,k)$ a network with $n$ nodes with expected degree of $k$.

\item \emph{Small-world} networks, generated using the Watts-Strogatz model \cite{watts1998collective}.
We denote by WS$(n,k)$ a network with $n$ nodes, an average degree of $k$, and a rewiring probability of $\frac{1}{4}$.

\item \emph{Scale-free} networks, generated using the Barab{\'a}si-Albert model \cite{barabasi1999emergence}. We denote by BA$(n,k)$ a network with $n$ nodes, with $k$ edges added with each new node.
The size of the initial clique is $k$.
\end{itemize}

To construct a multilayer network with $n$ nodes and $l$ layers using a network generation model $X \in \{ER,WS,BA\}$, we perform the following steps:

\begin{enumerate}

\item We create the set of node occurrences such that, for every node $v$ and every layer $\alpha$, the node $v$ occurs in $\alpha$ with probability $p_O$.
If, at the end of this process, $v$ does not occur on any layer, then we create one occurrence of $v$ in a layer chosen uniformly at random; this ensures that every node occurs in at least one layer.

\item For every layer, $\alpha$, we generate a network $X(|V^\alpha|,k)$ whose set of nodes consists of all the nodes that occur in $\alpha$ from the previous step.

\item For every two occurrences of the same node, we create an inter-layer edge between them with probability $p_C$. This results in a network with diagonal couplings.

\end{enumerate}

The network created using these steps will be denoted by $X^l (n,k)$.
In our experiments, we use networks where $p_O = p_C = \frac{1}{2}$ and where the number of layers is $l=3$ (the choice of 3 layers was inspired by the work of Gera~\etal~\cite{gera2017three}). Additionally, we use the following real-life networks:

\begin{itemize}

\item \textit{FF-TW-YT} dataset \cite{MagnaniSBP10} consisting of a network of connections between $1722$ individuals that have an account in each of the following social media site: Friendfeed, Twitter, and YouTube; each of these sites is represented by a separate layer in the network.

\item \textit{Provisional IRA} dataset \cite{gill2014lethal} consisting of the network of connections between members of the Provisional Irish Republic Army in the period between 1970 and 1998.
The network consists of $937$ nodes and $5$ layers, where layers correspond to contacts between organization members in different time periods.

\item \textit{Lazega law firm} dataset \cite{lazega2001collegial} consisting of connections between attorneys working for a US corporate law firm.
The network consists of $71$ nodes and $3$ layers, where each layer corresponds to a different type of relationship, \ie, friendship, professional cooperation and mentorship.

\item \textit{CS Aarhus} dataset \cite{magnani2013combinatorial} consisting of connections between the employees of the Computer Science department at Aarhus.
The network consists of $61$ nodes and $5$ layers, where each layer corresponds to a different type of relationship between employees, \eg, Facebook friendships, co-authorship of papers, \etc

\end{itemize}

\begin{table}[t!]
\caption{Characteristics of the considered datasets.}
\smallskip
\centering
\smallskip
\begin{tabular}{lcccc}
Dataset	& $|V|$ & $|L|$ & $|V_L|$ & $|E_L|$ \\
\hline
FF-TW-YT & 574 & 3 & 1722 & 5681 \\
Provisional IRA & 937 & 5 & 1570 & 3398 \\
Lazega law firm & 71 & 3 & 211 & 2051 \\
CS Aarhus & 61 & 5 & 224 & 948 \\
\end{tabular}
\label{tab:datasets}
\end{table}

Table~\ref{tab:datasets} presents detailed characteristics of the considered datasets.

\subsection{Heuristic Algorithms}

Recall that the ``group of contacts'' refers to the set of individuals whom the evader wishes to connect to. We will refer to each member of this group as a ``contact''.
Notice that a typical member of a social network does not have complete knowledge about the network's structure.
Hence, we assume that the evader's knowledge is limited to the connections between the contacts, as well as the degree of each contact.
All of our heuristic algorithms take only this information into account.
Specifically:

\begin{itemize}

\item \textit{Random}---This heuristic connects the evader to every contact in a layer chosen uniformly at random out of all layers in which both the evader and that contact occur.

\item \textit{All in one}---This heuristic (Algorithm~\ref{alg:all-in-one}) focuses on creating edges between the evader and her contacts in as few layers as possible. The intuition is that, by focusing all activities of the evader in a small number of layers (if possible, in only one layer), the global centrality measures would assign low importance to the evader.
Even though this heuristic might seem overly simplified, we include it as a reasonable baseline---a ``rule of thumb'' that could be readily implemented by members of the general public.

\item \textit{Fringe}---This heuristic (Algorithm~\ref{alg:fringe}) focuses on minimizing the number of nodes that are in close vicinity of the evader.
The main idea behind this heuristic is to maximize the average distance between the evader and other nodes, in the hope of achieving low ranking according to closeness centrality.
Given the limited knowledge of the evader about the network topology, the heuristic cannot analyze any nodes whose distance from the evader is greater than $2$.
Therefore, the heuristic simply focuses on minimizing the number of neighbors of the contacts.

\item \textit{Density}---This heuristic (Algorithm~\ref{alg:density}) is meant to link the evader to densely connected groups in each layer.
Here, the underlying idea is that edges between the contacts act as ``shortcuts'', preventing the shortest paths in the network from running through the evader, thus reducing her betweenness centrality.
Intuitively, the heuristic prefers to connect the evader to a contact $v$ in layers where $v$ is connected to many nodes that are already connected to the evader (the term $|\{w \in \F : (\vs^\alpha,w^\alpha) \in \Add\} \cap N^\alpha(v)|$ in the numerator), as well as layers where $v$ has many connections with other contacts (the term $|\F \cap N^\alpha(v)|$ in the numerator) to increase the chance of creating additional ``shortcuts''.
Finally, the heuristic prefers layers with fewer contacts connected to the evader (the term $|\{w \in \F : (\vs^\alpha,w^\alpha) \in \Add\}|$) to distribute the evader's connections among layers more uniformly, thereby helping her hide from local centrality measures.

\end{itemize}

\begin{algorithm}[tb!]
\caption{``All in one'' heuristic}
\label{alg:all-in-one}
\textbf{Input}:
Multilayer network $M$, the evader $\vs$, contacts $\F$ \\
\textbf{Output}: Edges to be added to the network, i.e., the set $\Add$ \\
\begin{algorithmic}[1]
\STATE $\Add \gets \emptyset$
\STATE $\F^* \gets \F$
\STATE $L^* \gets \{\alpha \in L : \vs \in V^\alpha\}$
\WHILE{$|\F^*| > 0$}
	\STATE $\alpha^* \gets \argmax_{\alpha \in L^*}|\F^* \cap V^\alpha|$
	\FOR{$v \in \F^* \cap V^\alpha$}
		\STATE $\Add = \Add \cup \{(\vs^{\alpha^*}, v^{\alpha^*})\}$
	\ENDFOR
	\STATE $\F^* = \F^* \setminus V^\alpha$
\ENDWHILE
\STATE \textbf{return} $\Add$
\end{algorithmic}
\end{algorithm}

\begin{algorithm}[tb!]
\caption{Fringe heuristic}
\label{alg:fringe}
\textbf{Input}:
Multilayer network $M$, the evader $\vs$, contacts $\F$ \\
\textbf{Output}: Edges to be added to the network, i.e., the set $\Add$  \\
\begin{algorithmic}[1]
\STATE $\Add \gets \emptyset$
\STATE $L^* \gets \{\alpha \in L : \vs \in V^\alpha\}$
\FOR{$v \in \F \cap V^\alpha$}
	\STATE $\alpha^* \gets \argmin_{\alpha \in L^*} |N^\alpha(v) \setminus \F|$
	\STATE $\Add = \Add \cup \{(\vs^{\alpha^*}, v^{\alpha^*})\}$
\ENDFOR
\STATE \textbf{return} $\Add$
\end{algorithmic}
\end{algorithm}

\begin{algorithm}[t]
\caption{Density heuristic}
\label{alg:density}
\textbf{Input}:
Multilayer network $M$, the evader $\vs$, contacts $\F$ \\
\textbf{Output}: Edges to be added to the network, i.e., the set $\Add$  \\
\begin{algorithmic}[1]
\STATE $\Add \gets \emptyset$
\STATE $L^* \gets \{\alpha \in L : \vs \in V^\alpha\}$
\FOR{$v \in \F \cap V^\alpha$}
	\STATE $\alpha^* \gets \argmax_{\alpha \in L^*}\frac{|\{w \in \F : (\vs^\alpha,w^\alpha) \in \Add\} \cap N^\alpha(v)| + |\F \cap N^\alpha(v)|}{\max(1,|\{w \in \F : (\vs^\alpha,w^\alpha) \in \Add\}|)}$
	\STATE $\Add = \Add \cup \{(\vs^{\alpha^*}, v^{\alpha^*})\}$
\ENDFOR
\STATE \textbf{return} $\Add$
\end{algorithmic}
\end{algorithm}

\begin{figure*}[t!]
\centering
\setlength\tabcolsep{1pt}
\renewcommand{\arraystretch}{0.01}
\begin{tabular}{m{.03\textwidth}m{.195\textwidth}m{.195\textwidth}m{.195\textwidth}m{.195\textwidth}m{.195\textwidth}}
& \multicolumn{1}{c}{Global Closeness}
& \multicolumn{1}{c}{Global Betweenness}
& \multicolumn{1}{c}{Local Degree}
& \multicolumn{1}{c}{Local Closeness}
& \multicolumn{1}{c}{Local Betweenness}\\
\rotatebox{90}{ER$^3(2000,10)$} &
\includegraphics[width=\linewidth,height=2.3cm]{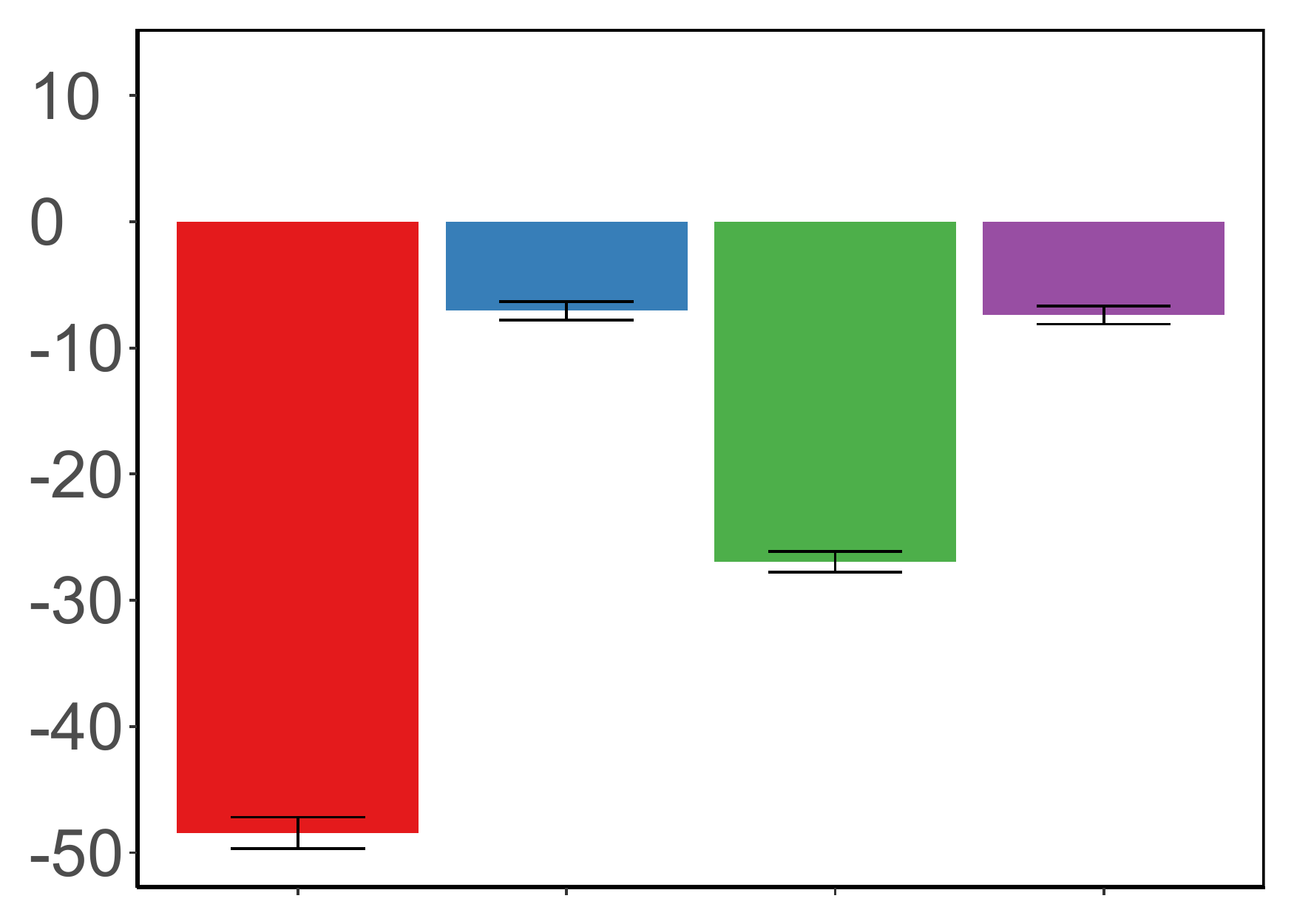} &
\includegraphics[width=\linewidth,height=2.3cm]{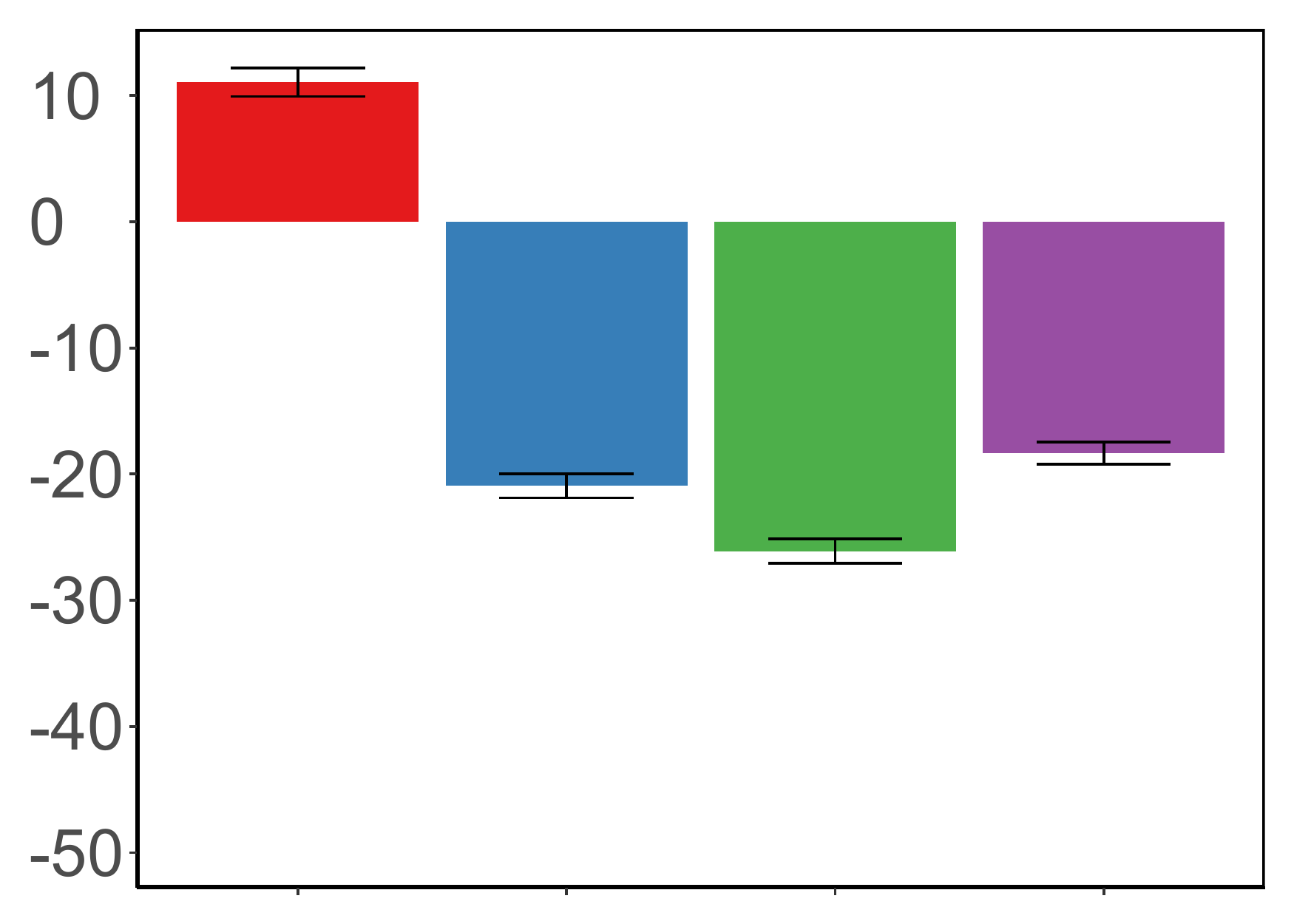} &
\includegraphics[width=\linewidth,height=2.3cm]{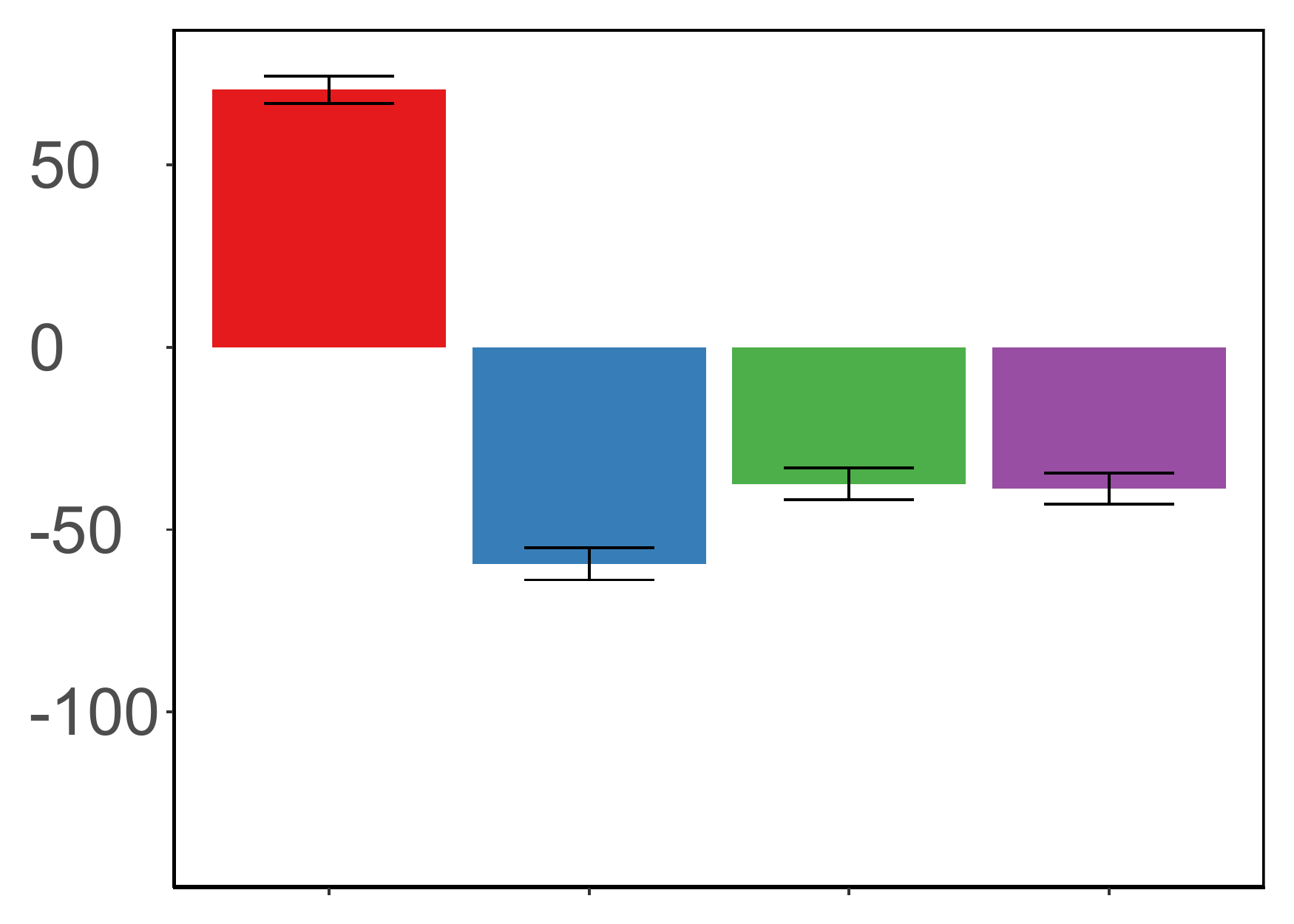} &
\includegraphics[width=\linewidth,height=2.3cm]{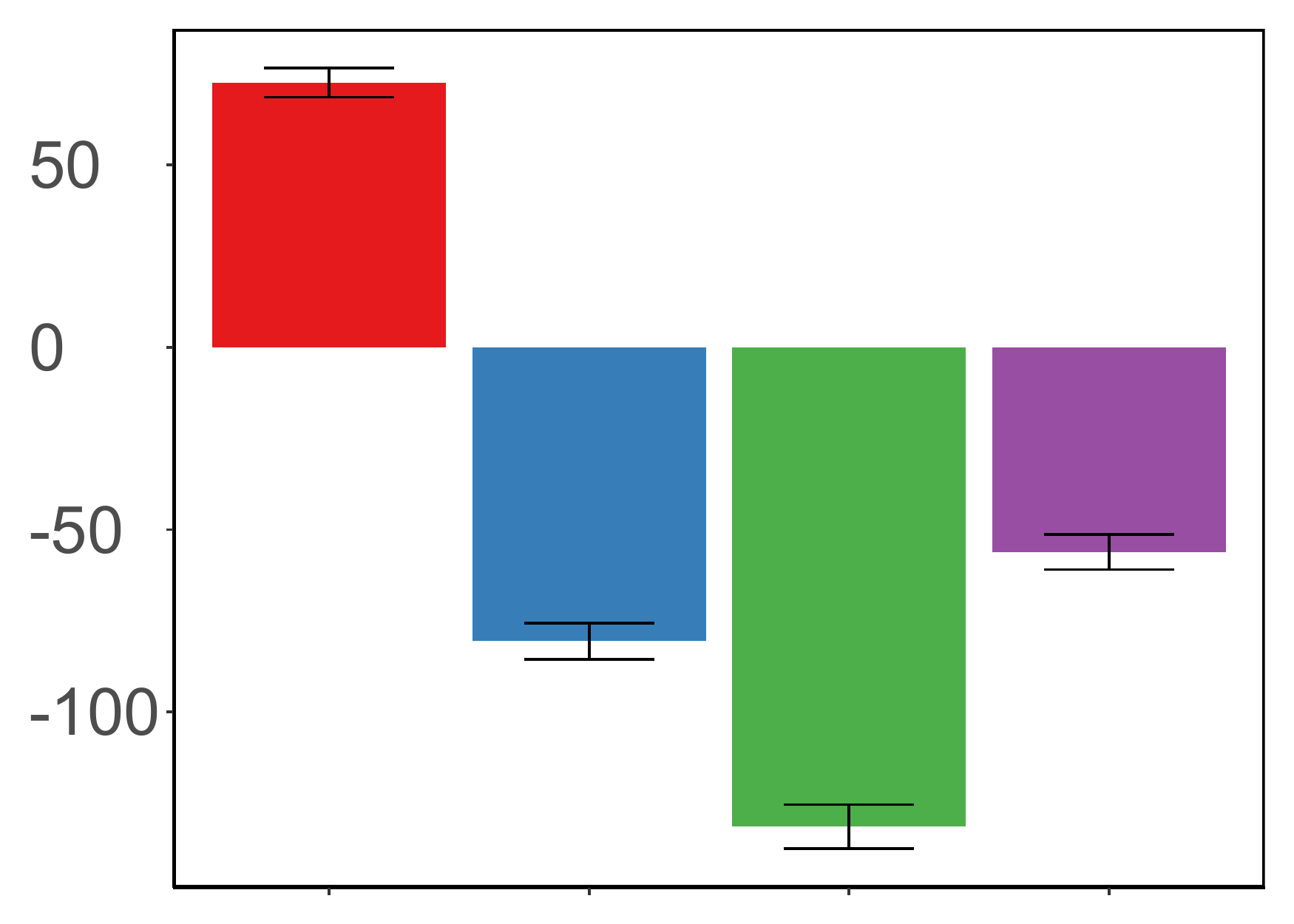} &
\includegraphics[width=\linewidth,height=2.3cm]{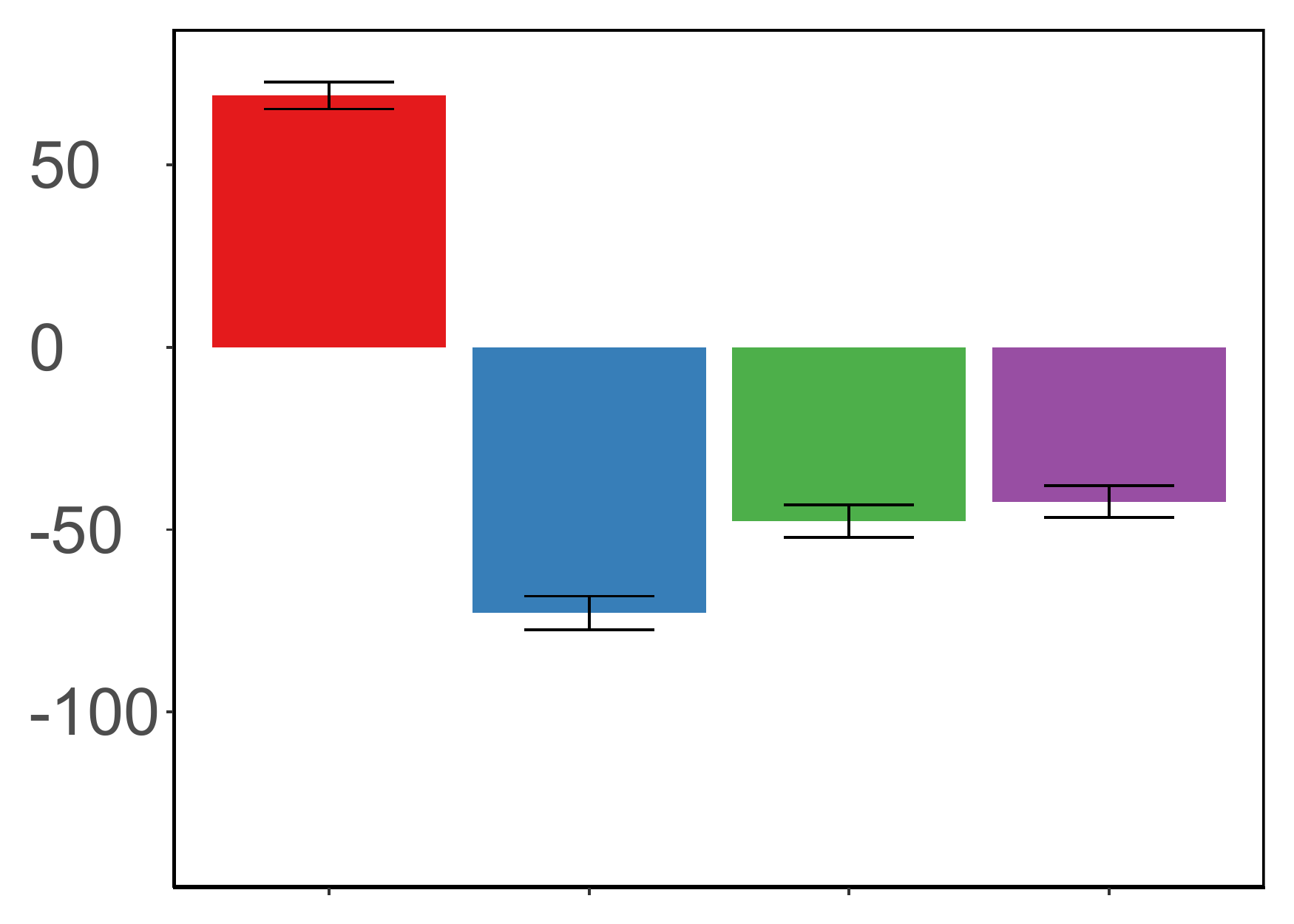} \\
\rotatebox{90}{WS$^3(2000,10)$} &
\includegraphics[width=\linewidth,height=2.3cm]{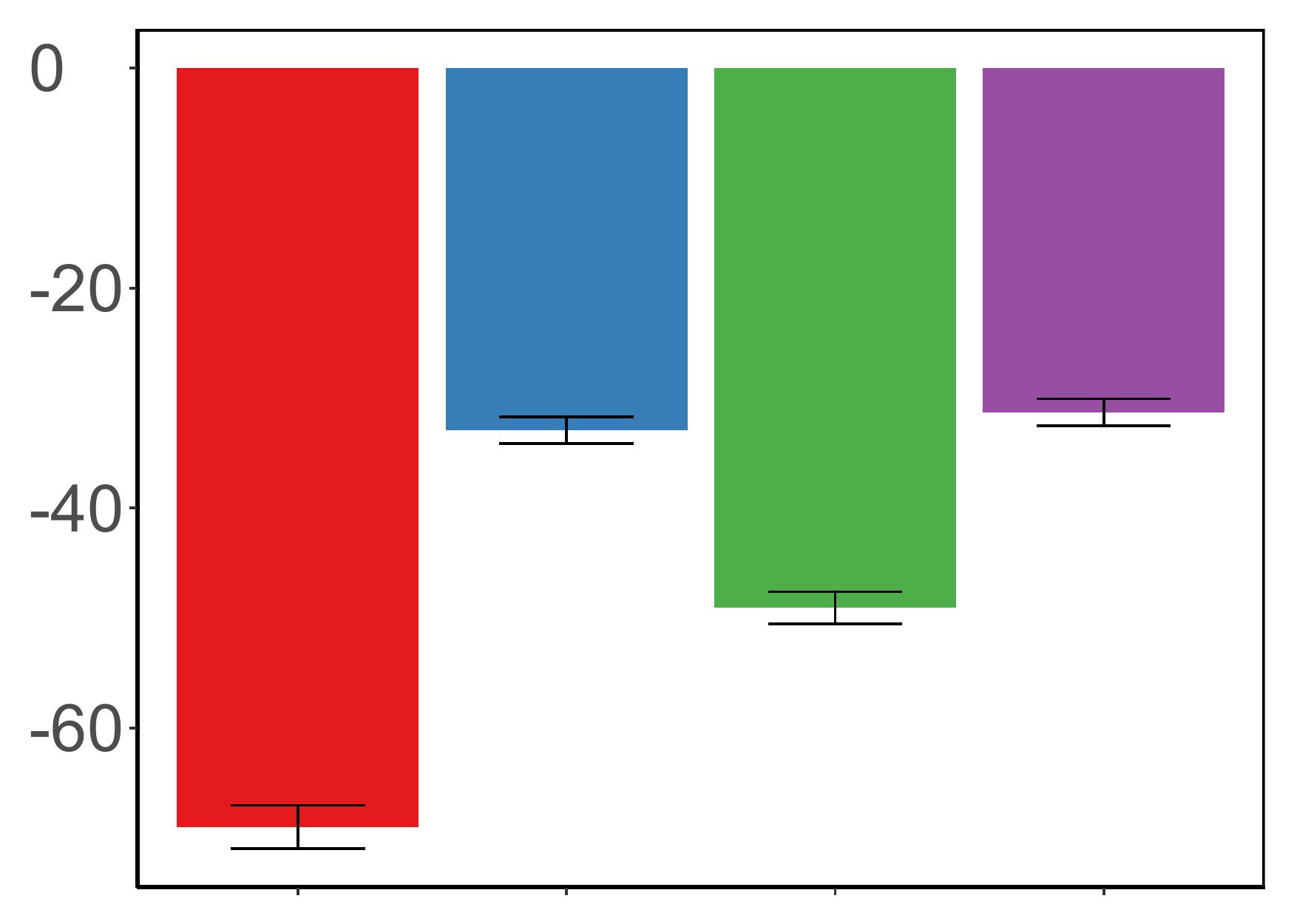} &
\includegraphics[width=\linewidth,height=2.3cm]{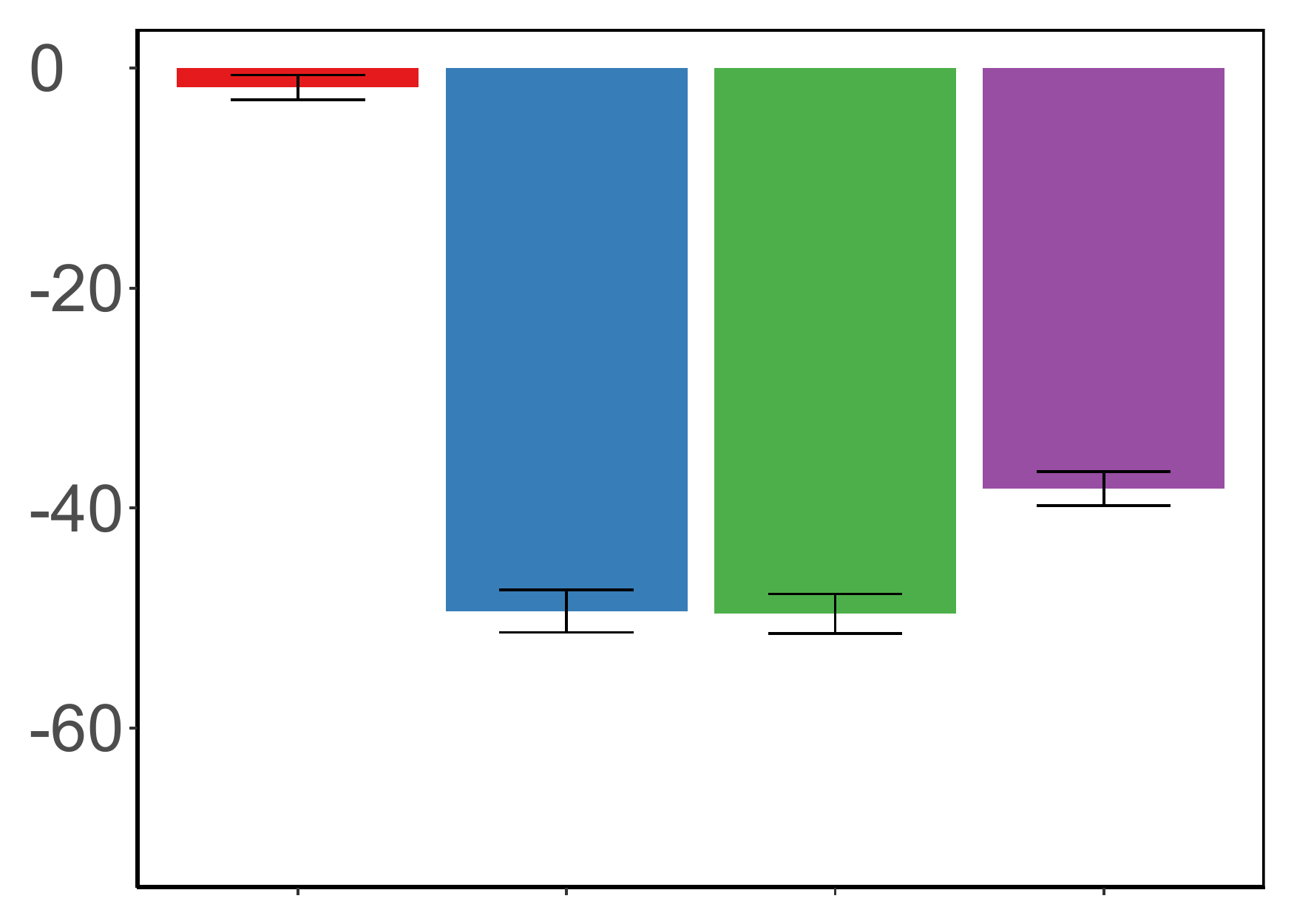} &
\includegraphics[width=\linewidth,height=2.3cm]{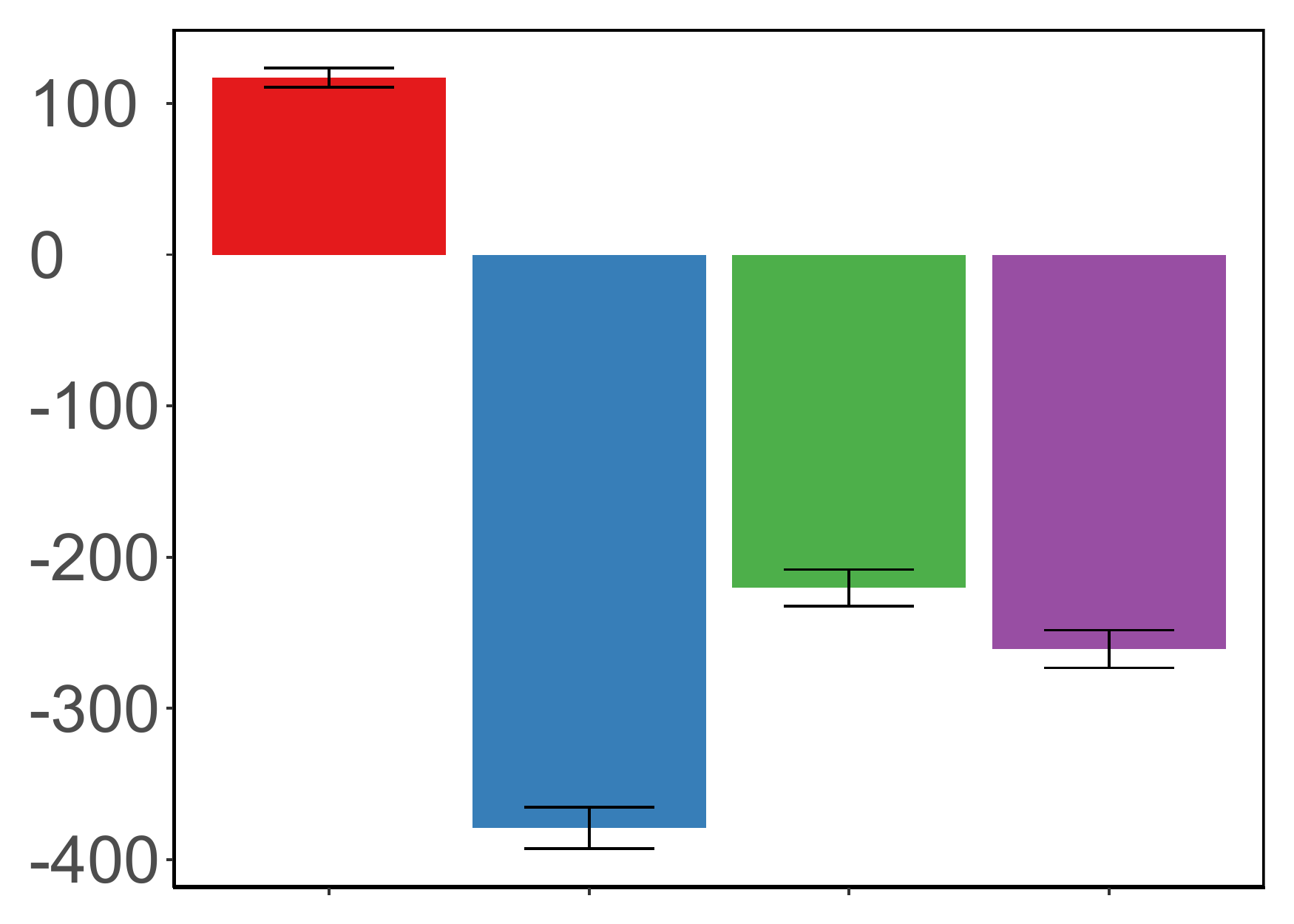} &
\includegraphics[width=\linewidth,height=2.3cm]{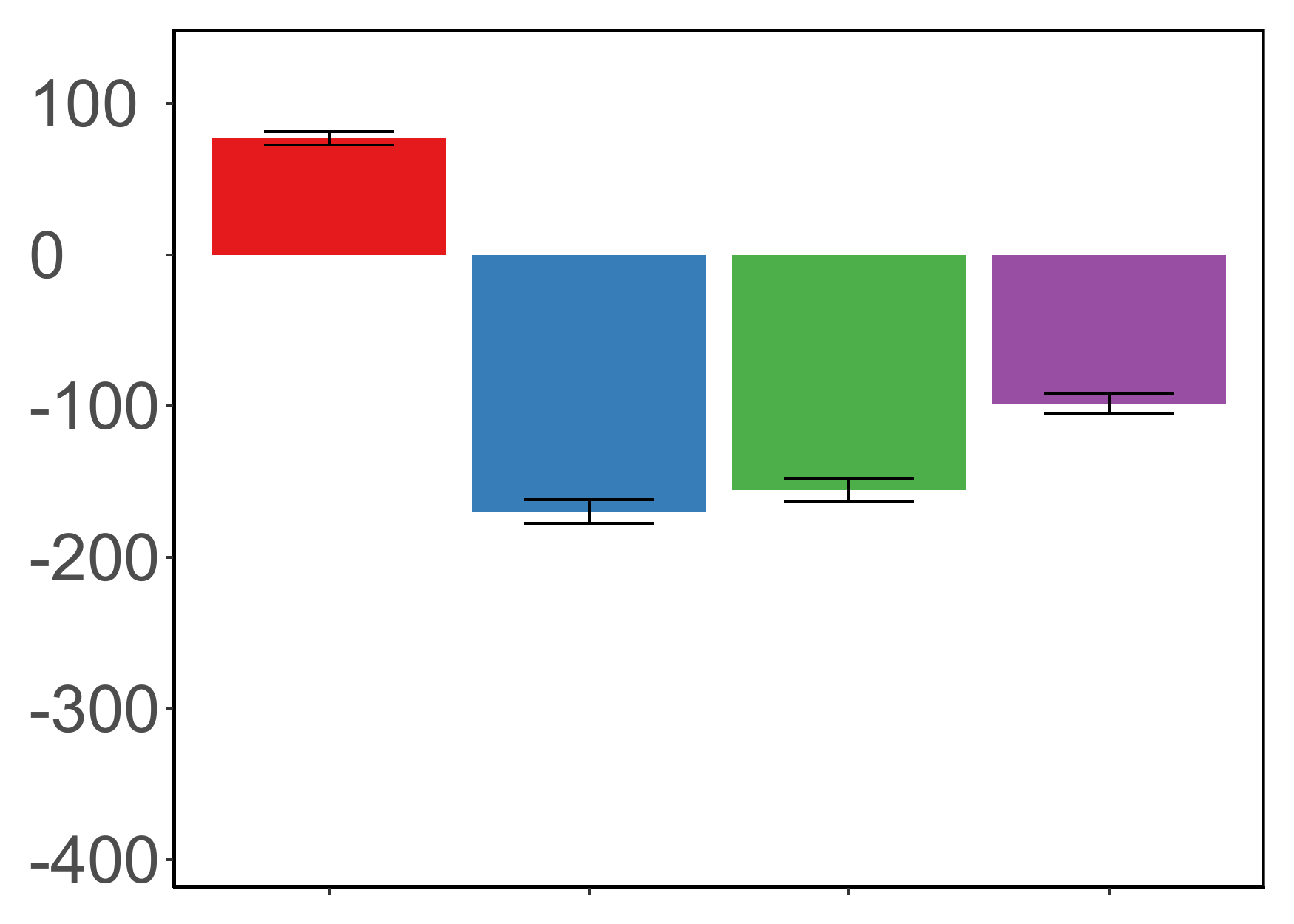} &
\includegraphics[width=\linewidth,height=2.3cm]{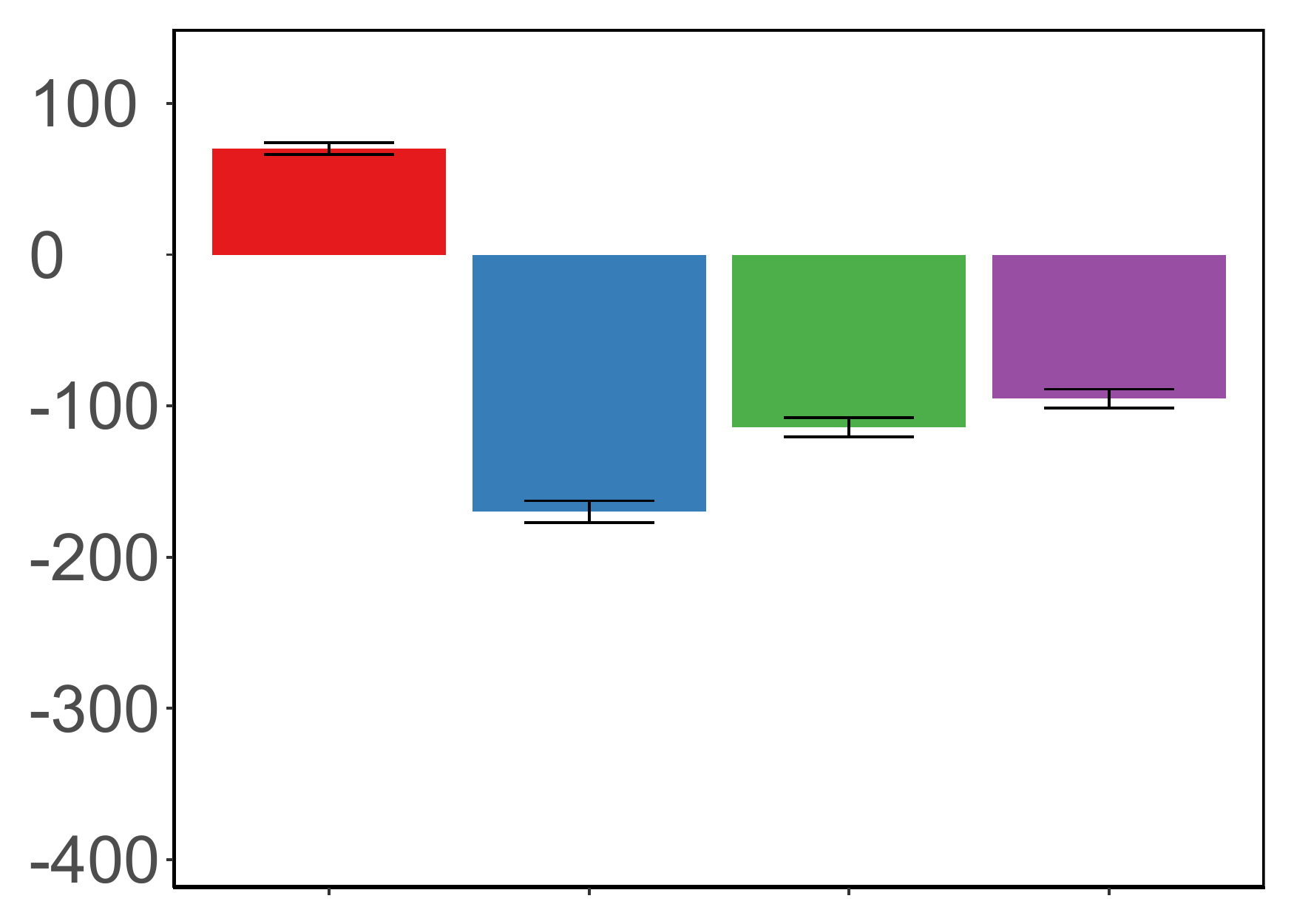} \\
\rotatebox{90}{BA$^3(2000,5)$} &
\includegraphics[width=\linewidth,height=2.3cm]{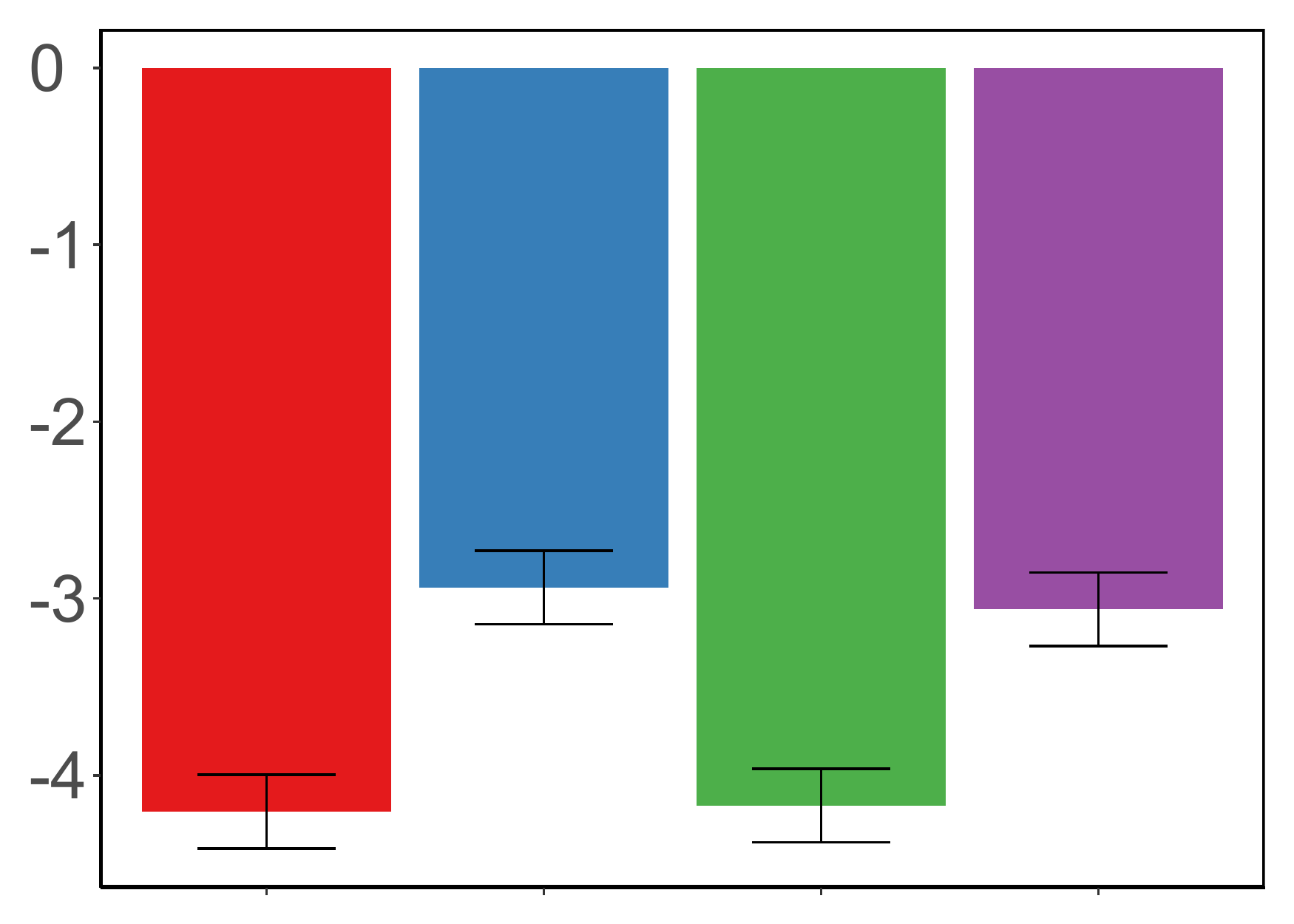} &
\includegraphics[width=\linewidth,height=2.3cm]{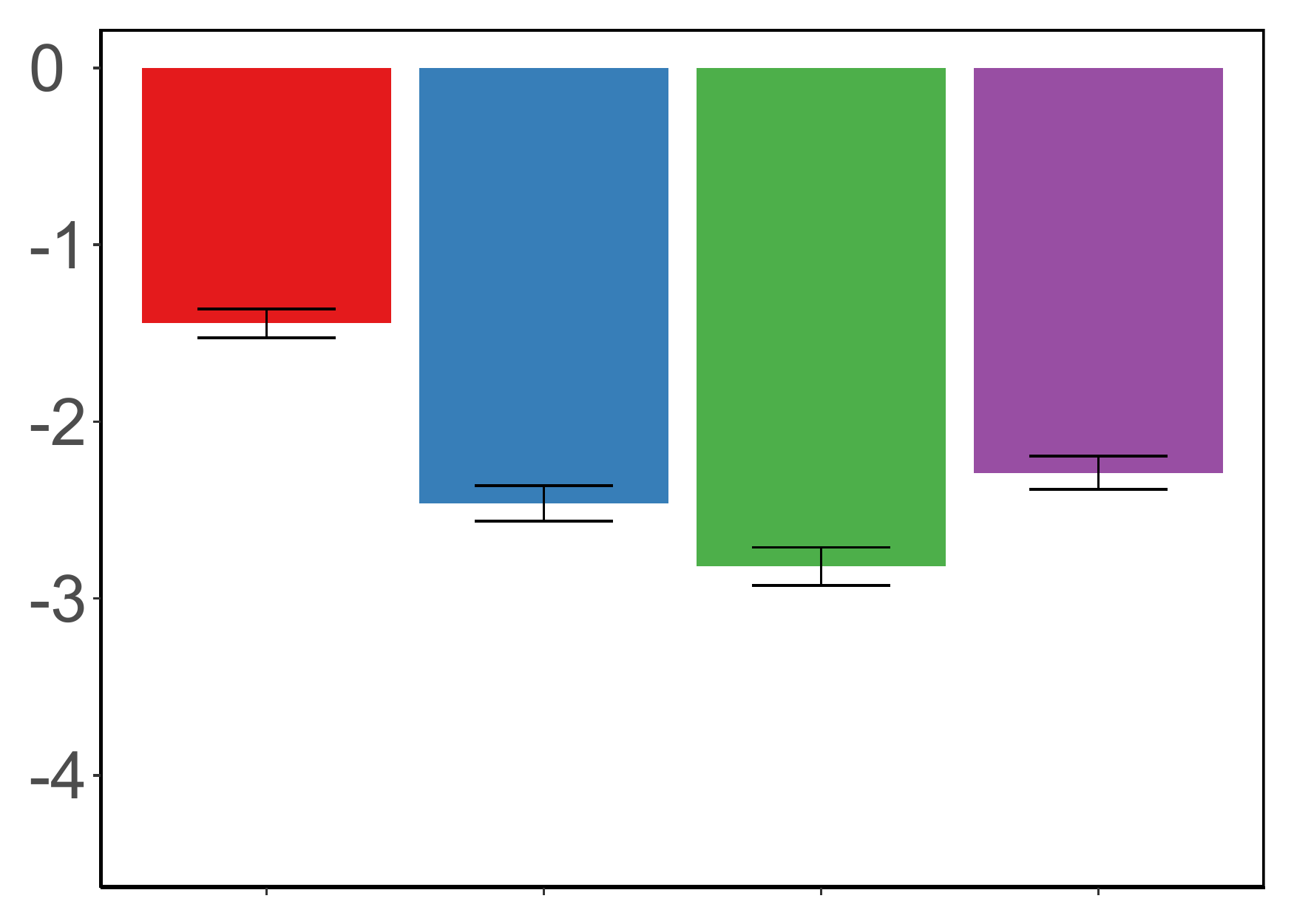} &
\includegraphics[width=\linewidth,height=2.3cm]{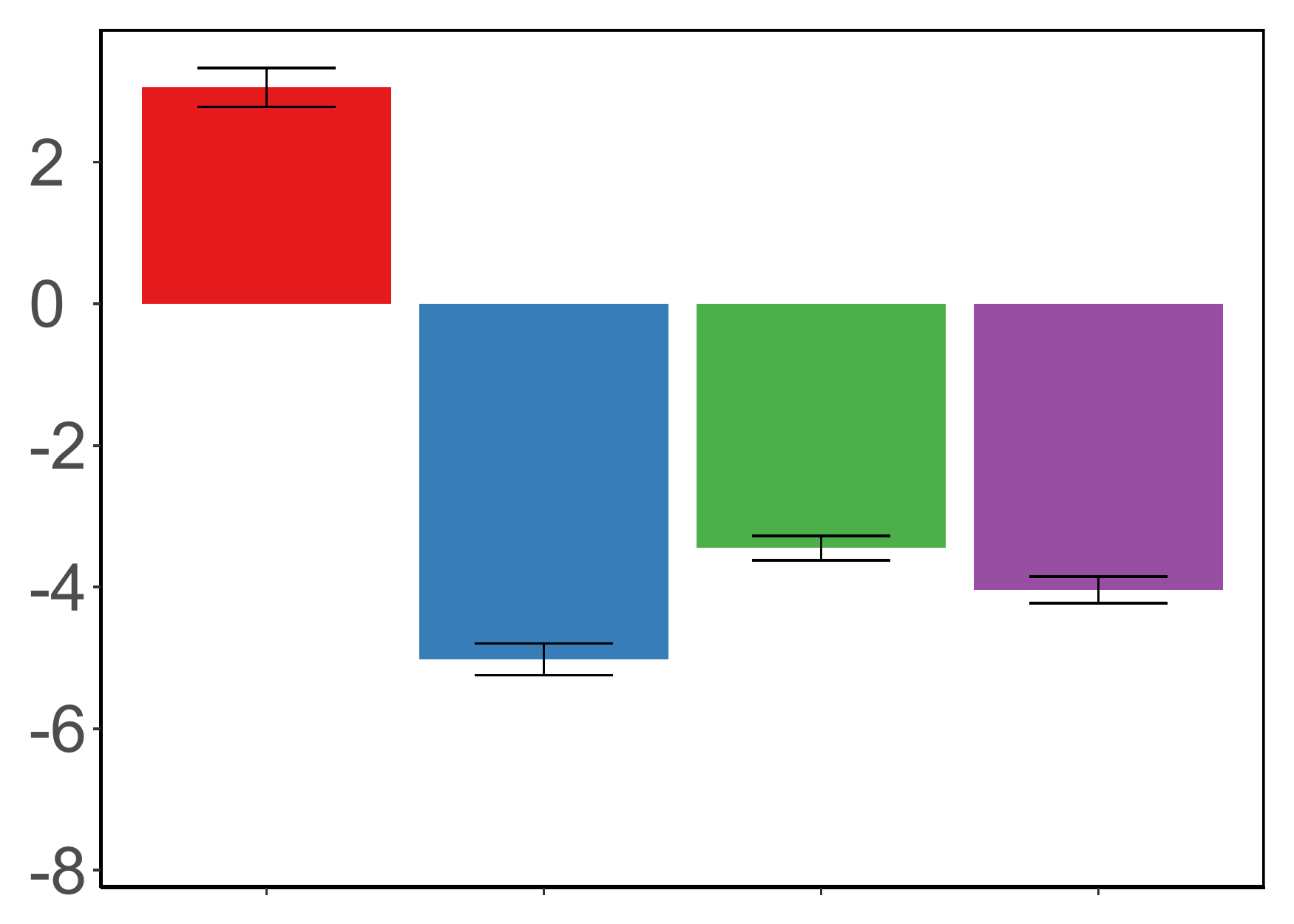} &
\includegraphics[width=\linewidth,height=2.3cm]{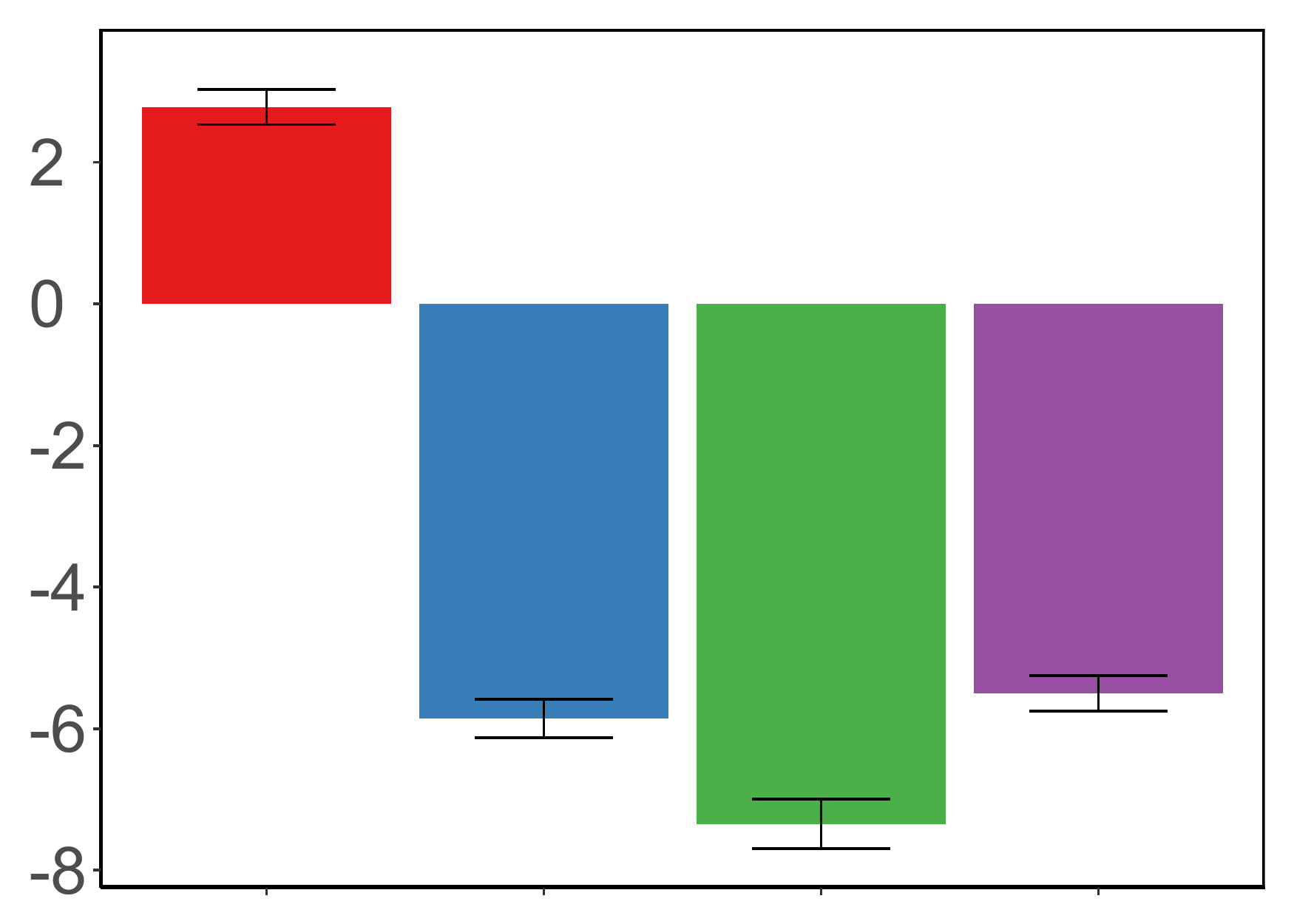} &
\includegraphics[width=\linewidth,height=2.3cm]{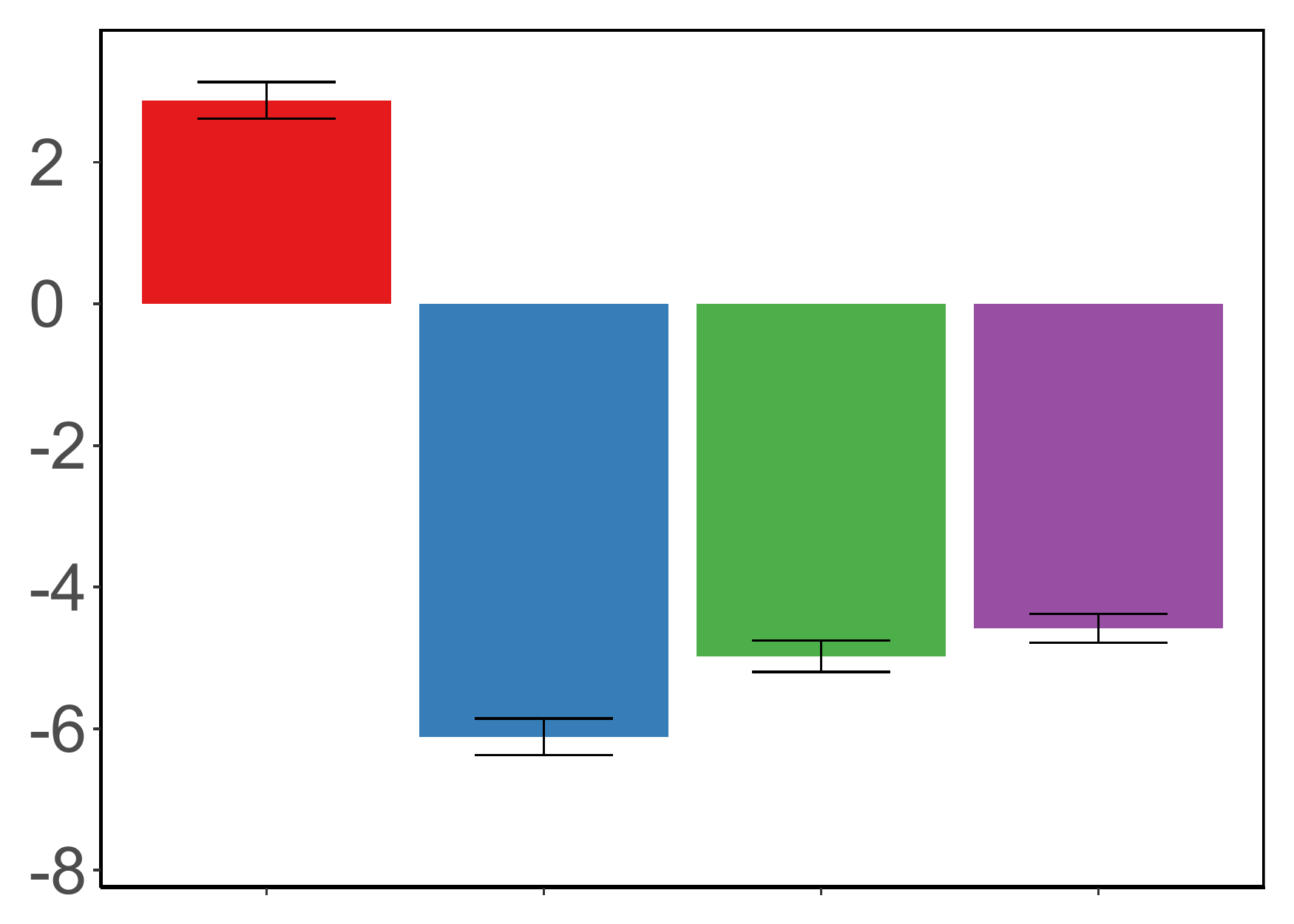} \\
\multicolumn{6}{c}{\includegraphics[width=.3\linewidth]{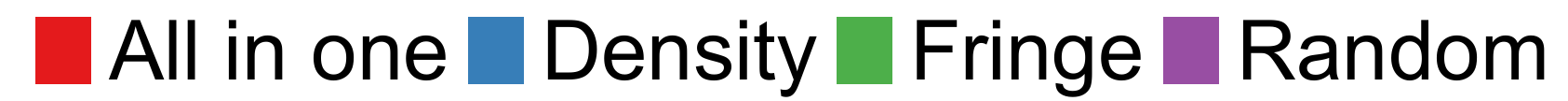}}
\end{tabular}
\caption{
Given different centrality measures and different networks with $2000$ nodes and $3$ layers (ER---Erd{\H{o}}s-R{\'e}nyi, WS---Watts-Strogatz, BA---Barab{\'a}si-Albert), the figure depicts the average change in centrality ranking of $10$ different evaders as a result of execution of hiding heuristics.
The experiment is repeated $100$ times, with a new network generated each time.
Error bars represent $95\%$ confidence intervals.
}
\label{fig:simulations-random}
\end{figure*}

\begin{figure*}[ht!]
\centering
\setlength\tabcolsep{1pt}
\renewcommand{\arraystretch}{0.01}
\begin{tabular}{m{.05\textwidth}m{.195\textwidth}m{.195\textwidth}m{.195\textwidth}m{.195\textwidth}m{.195\textwidth}}
& \multicolumn{1}{c}{Global Closeness}
& \multicolumn{1}{c}{Global Betweenness}
& \multicolumn{1}{c}{Local Degree}
& \multicolumn{1}{c}{Local Closeness}
& \multicolumn{1}{c}{Local Betweenness}\\
\rotatebox{90}{FF-TW-YT} &
\includegraphics[width=\linewidth]{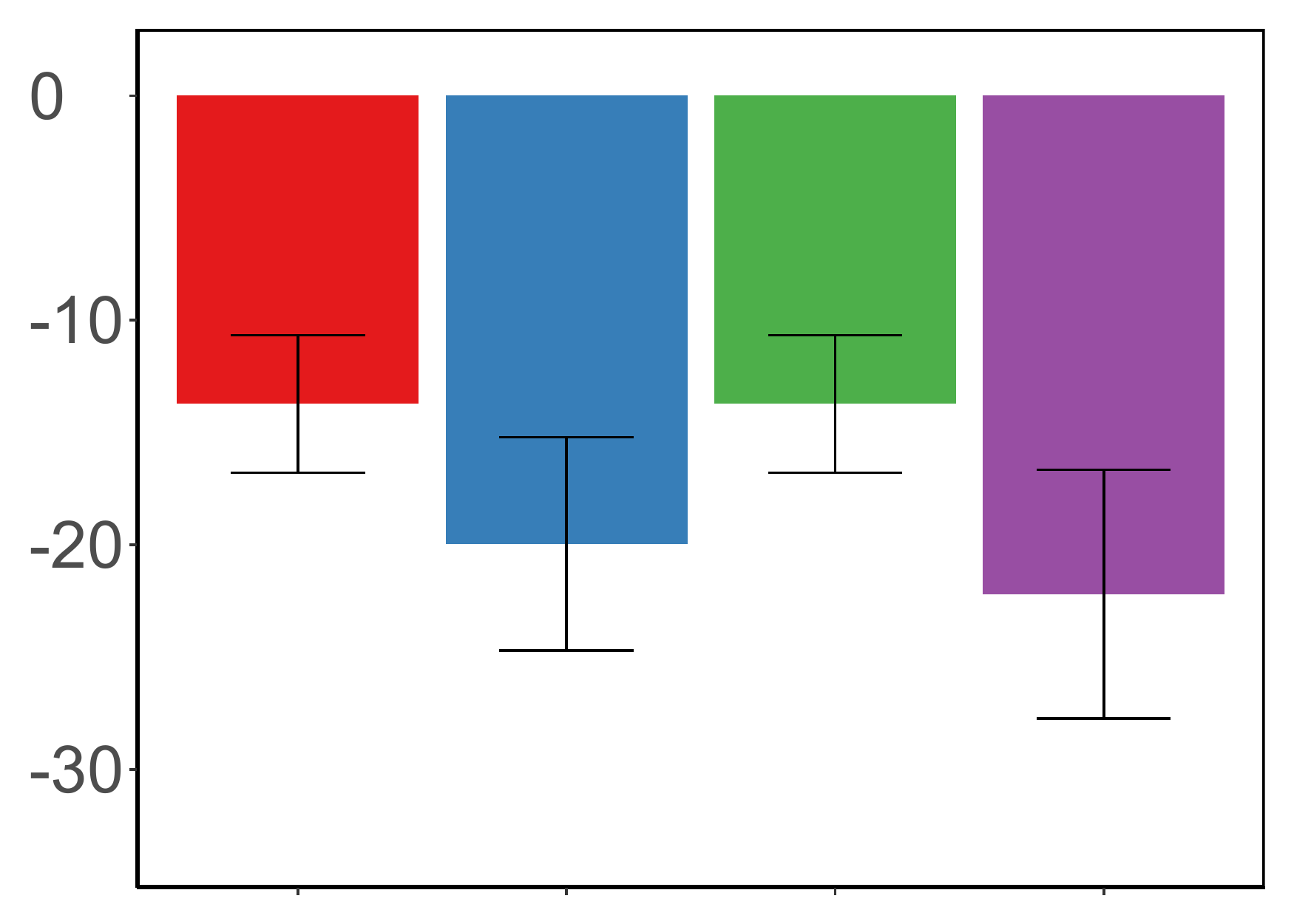} &
\includegraphics[width=\linewidth]{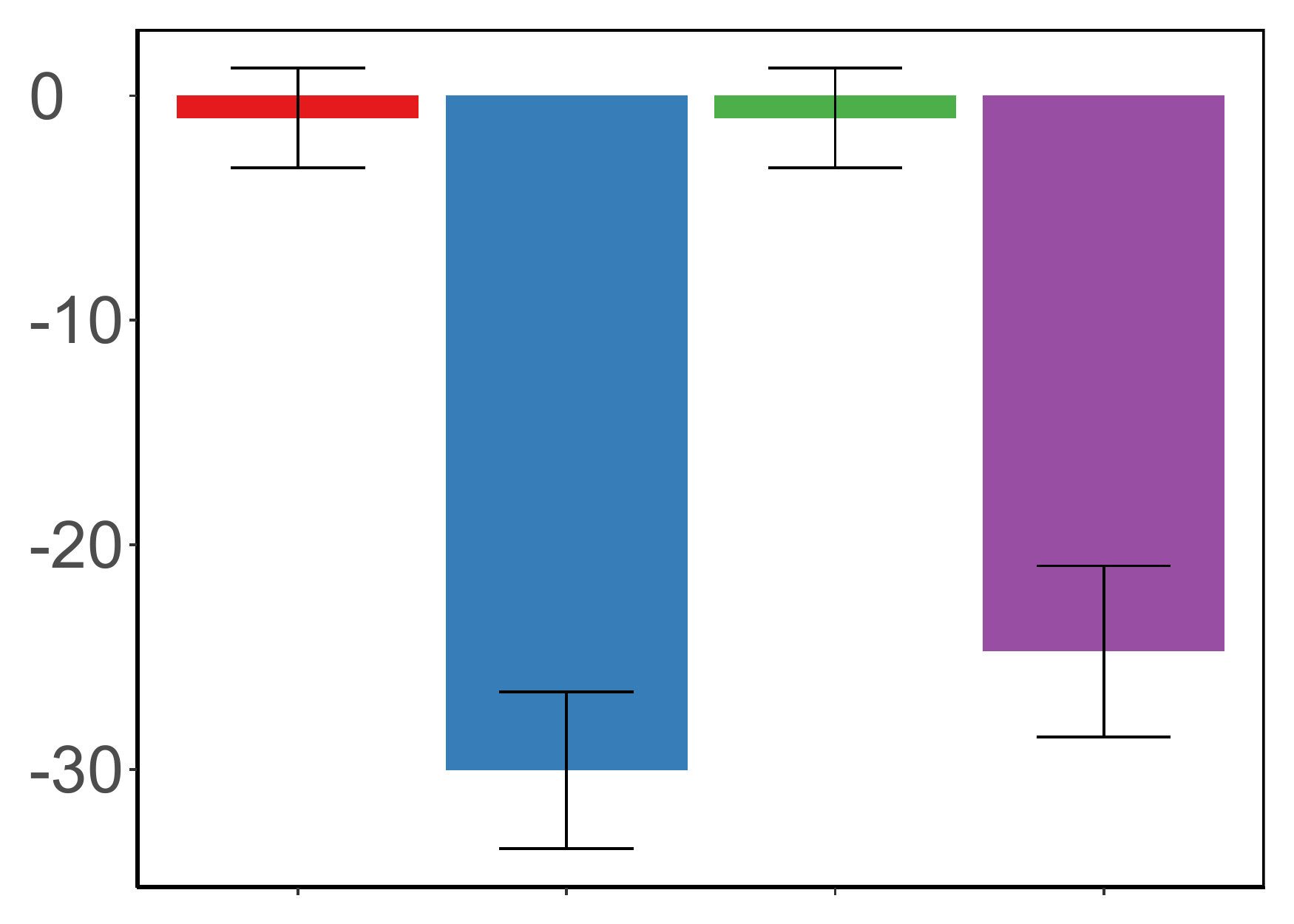} &
\includegraphics[width=\linewidth]{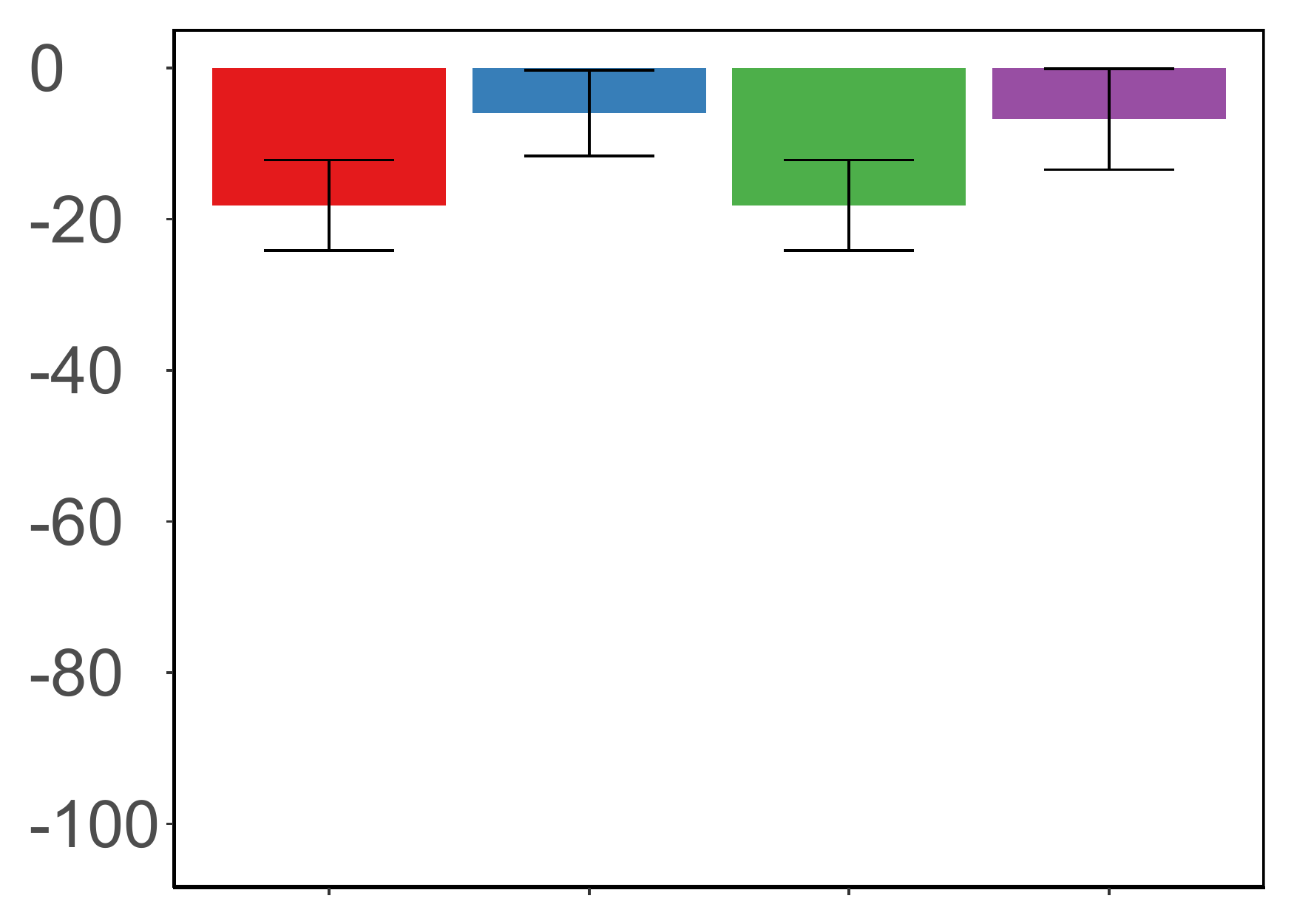} &
\includegraphics[width=\linewidth]{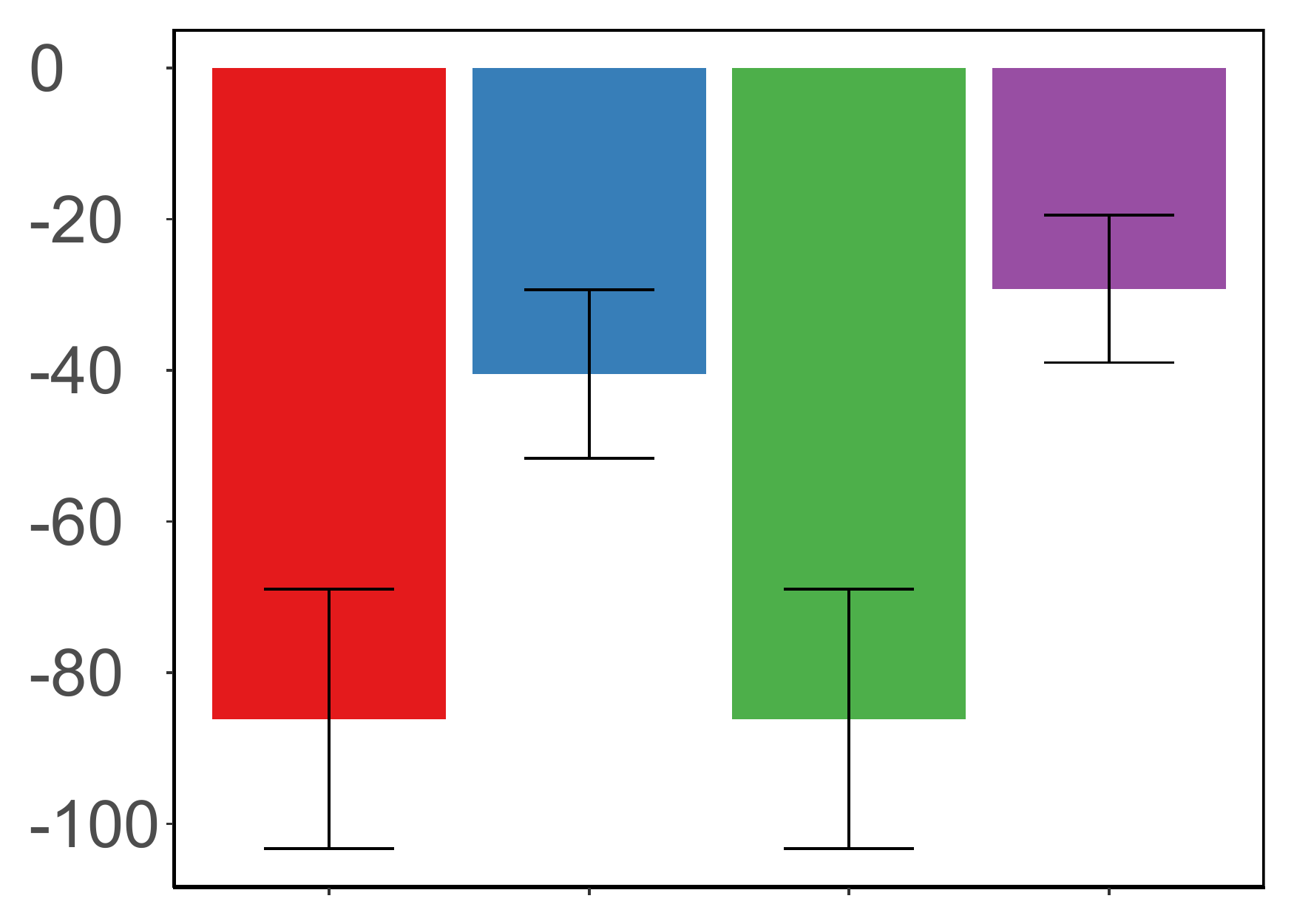} &
\includegraphics[width=\linewidth]{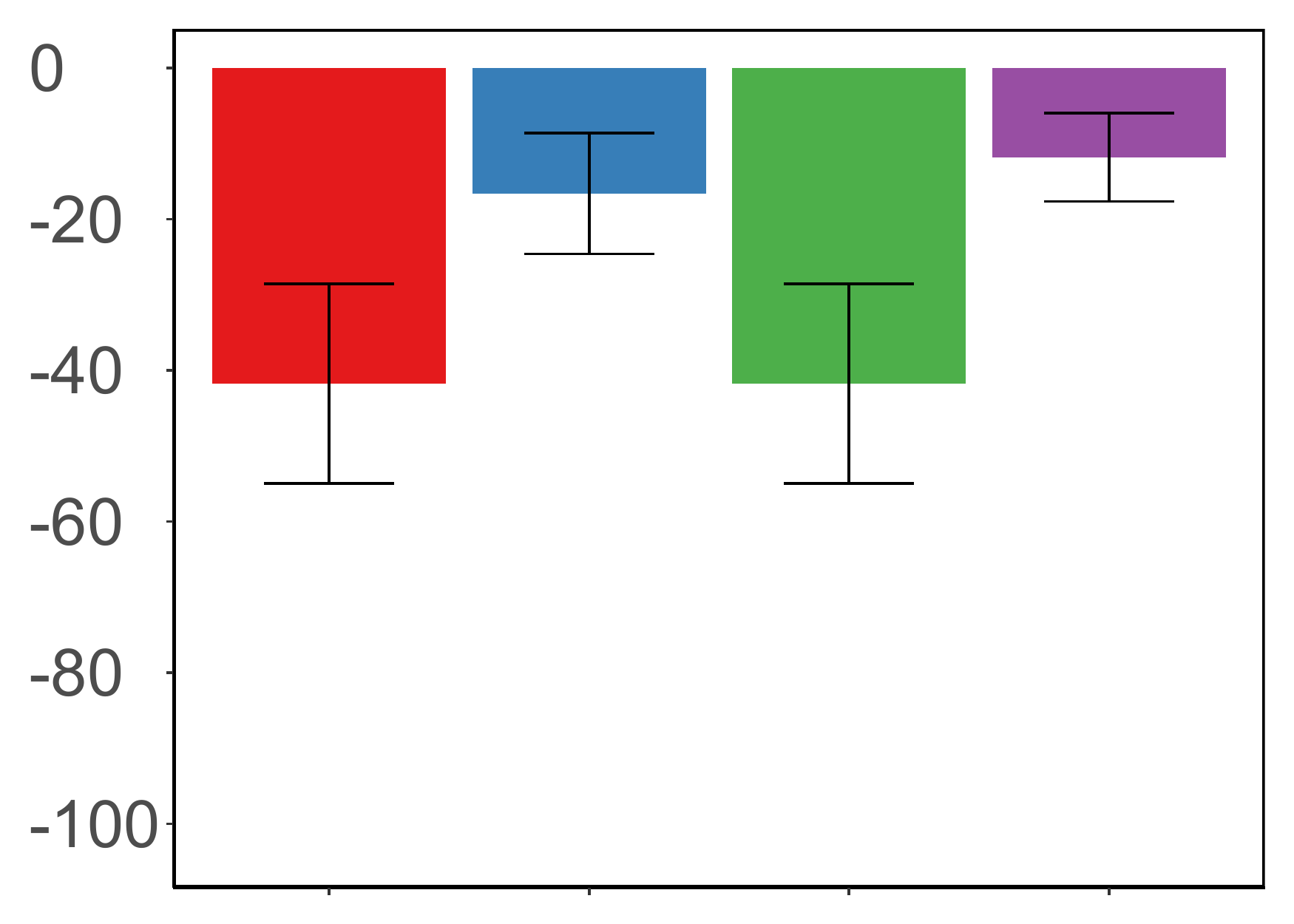} \\
\rotatebox{90}{Provisional IRA} &
\includegraphics[width=\linewidth]{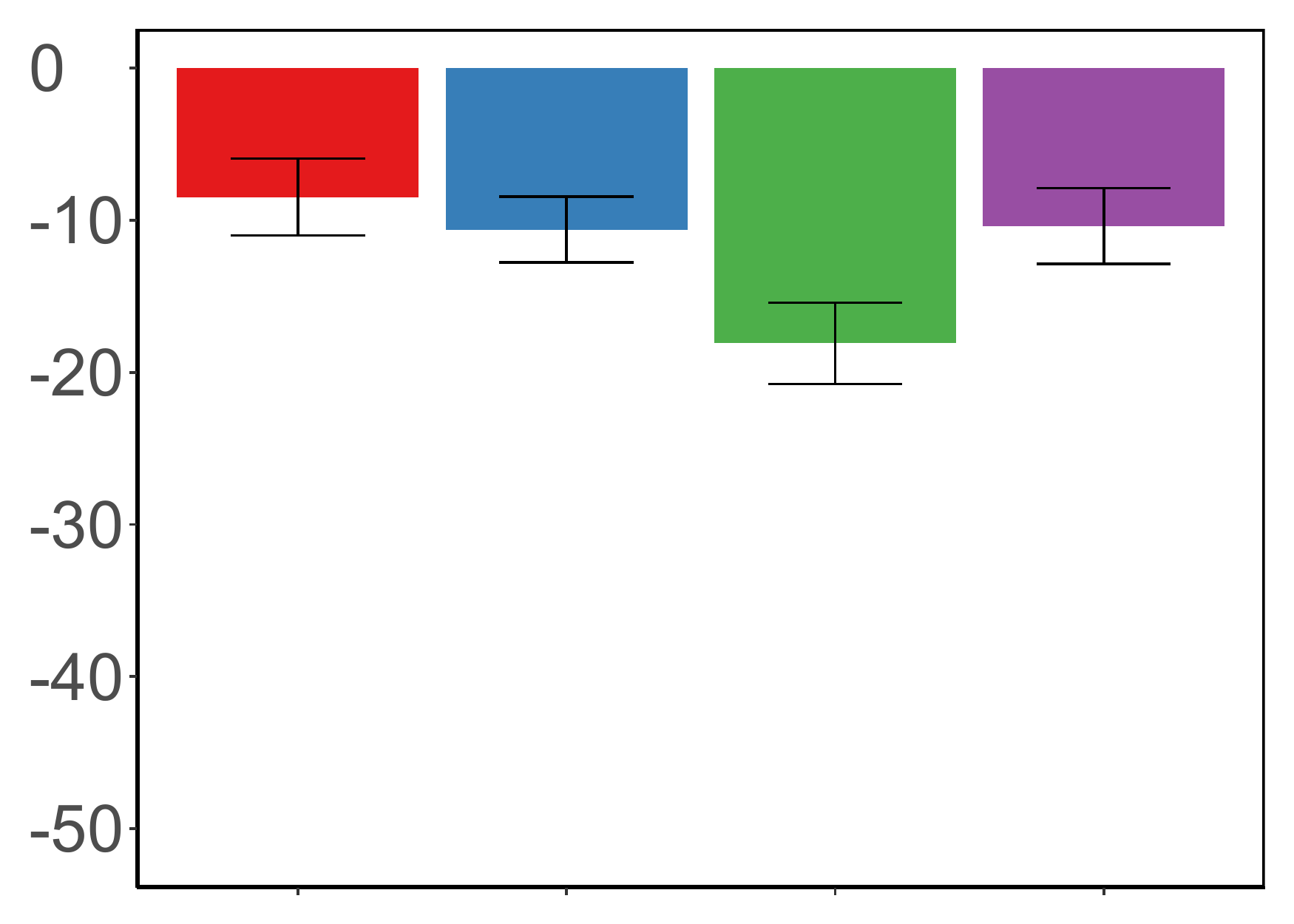} &
\includegraphics[width=\linewidth]{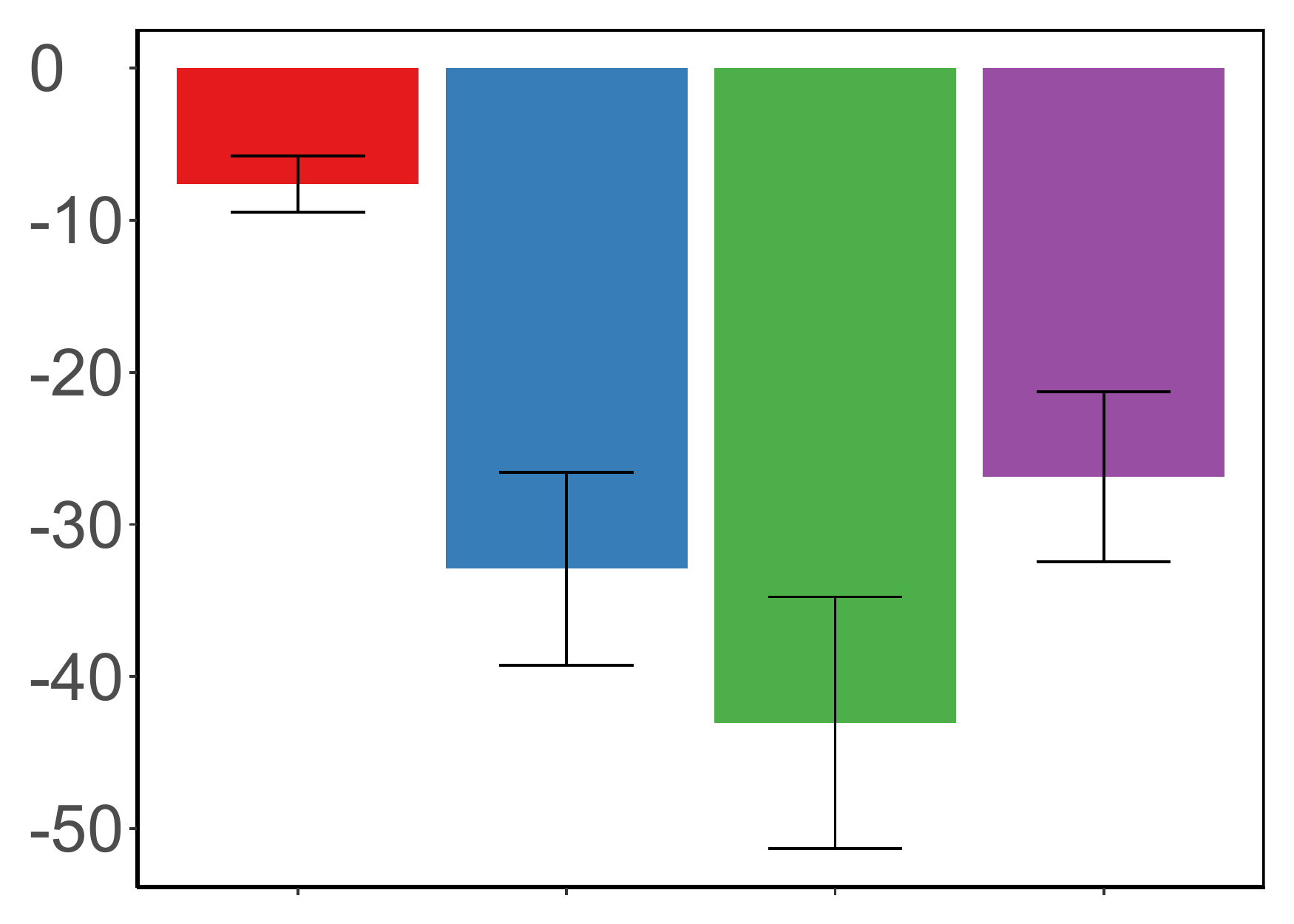} &
\includegraphics[width=\linewidth]{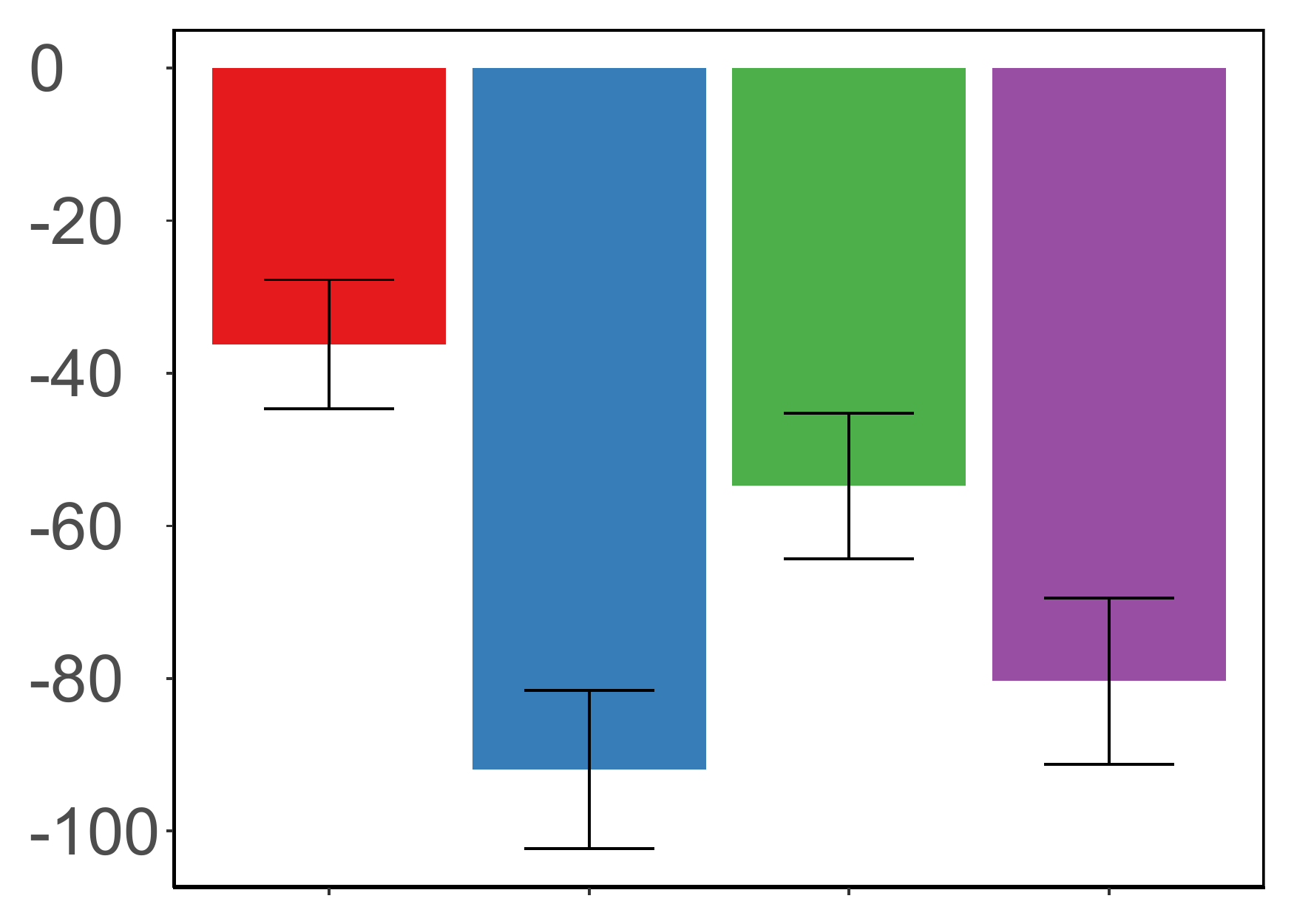} &
\includegraphics[width=\linewidth]{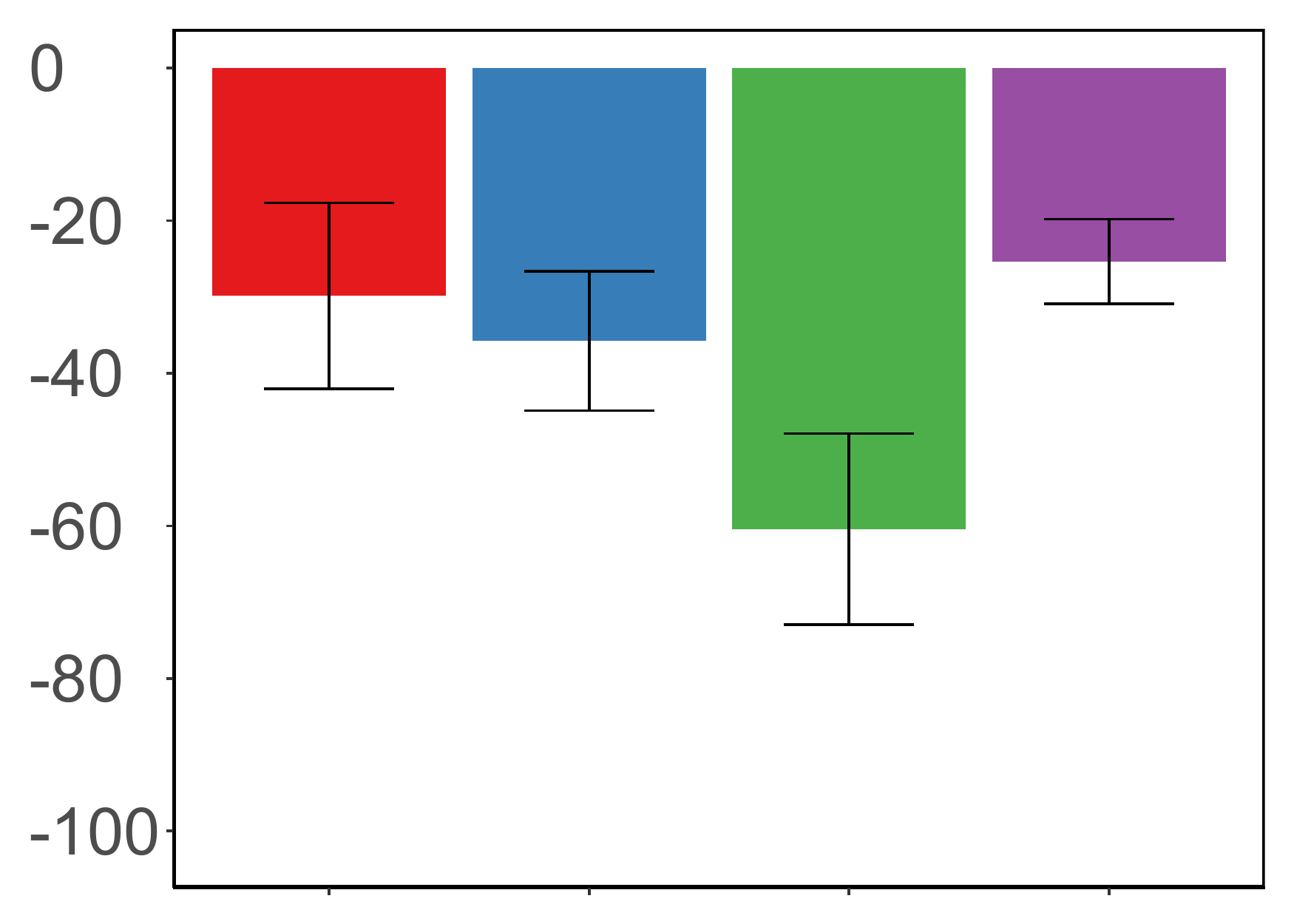} &
\includegraphics[width=\linewidth]{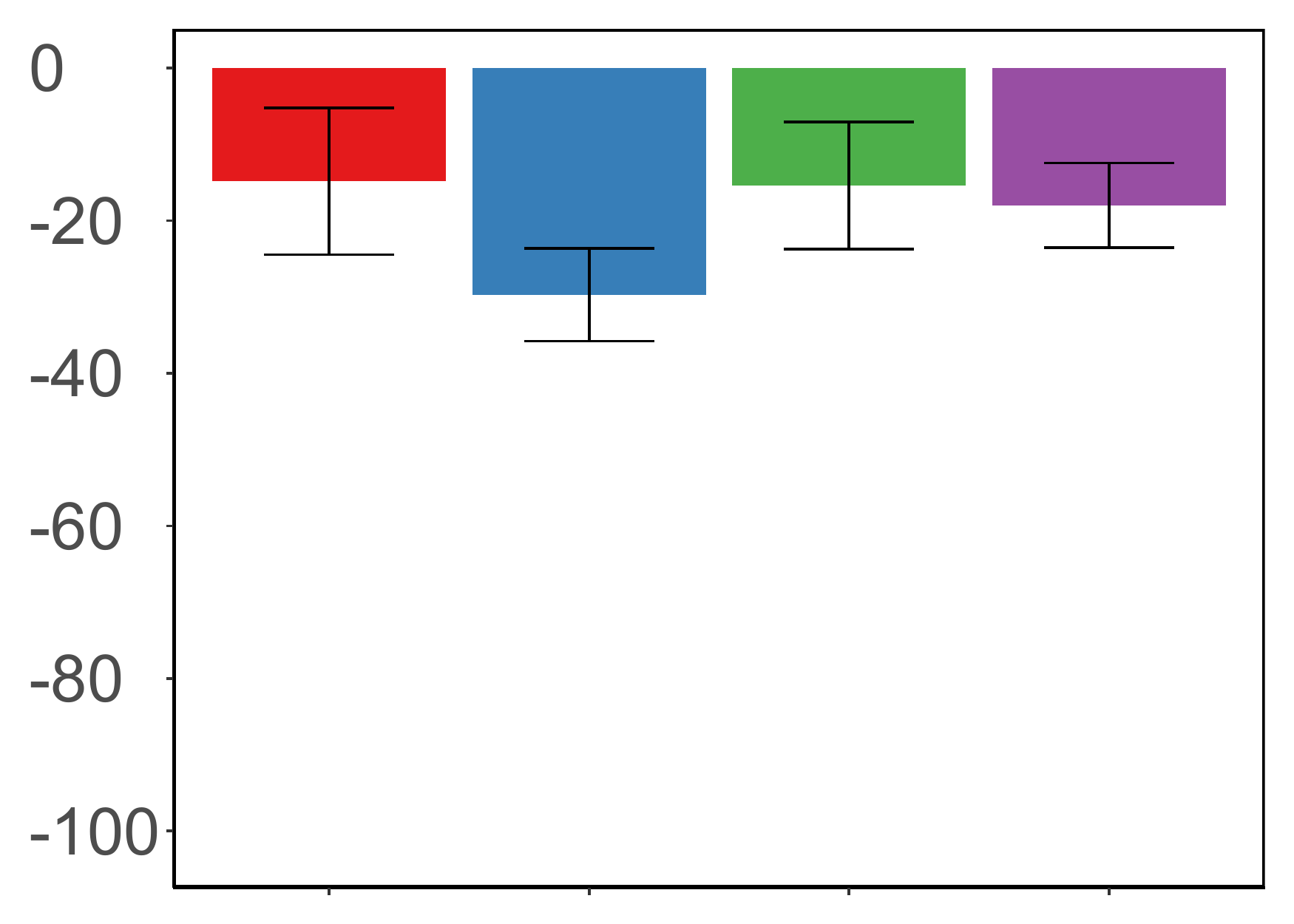} \\
\rotatebox{90}{Law firm} &
\includegraphics[width=\linewidth]{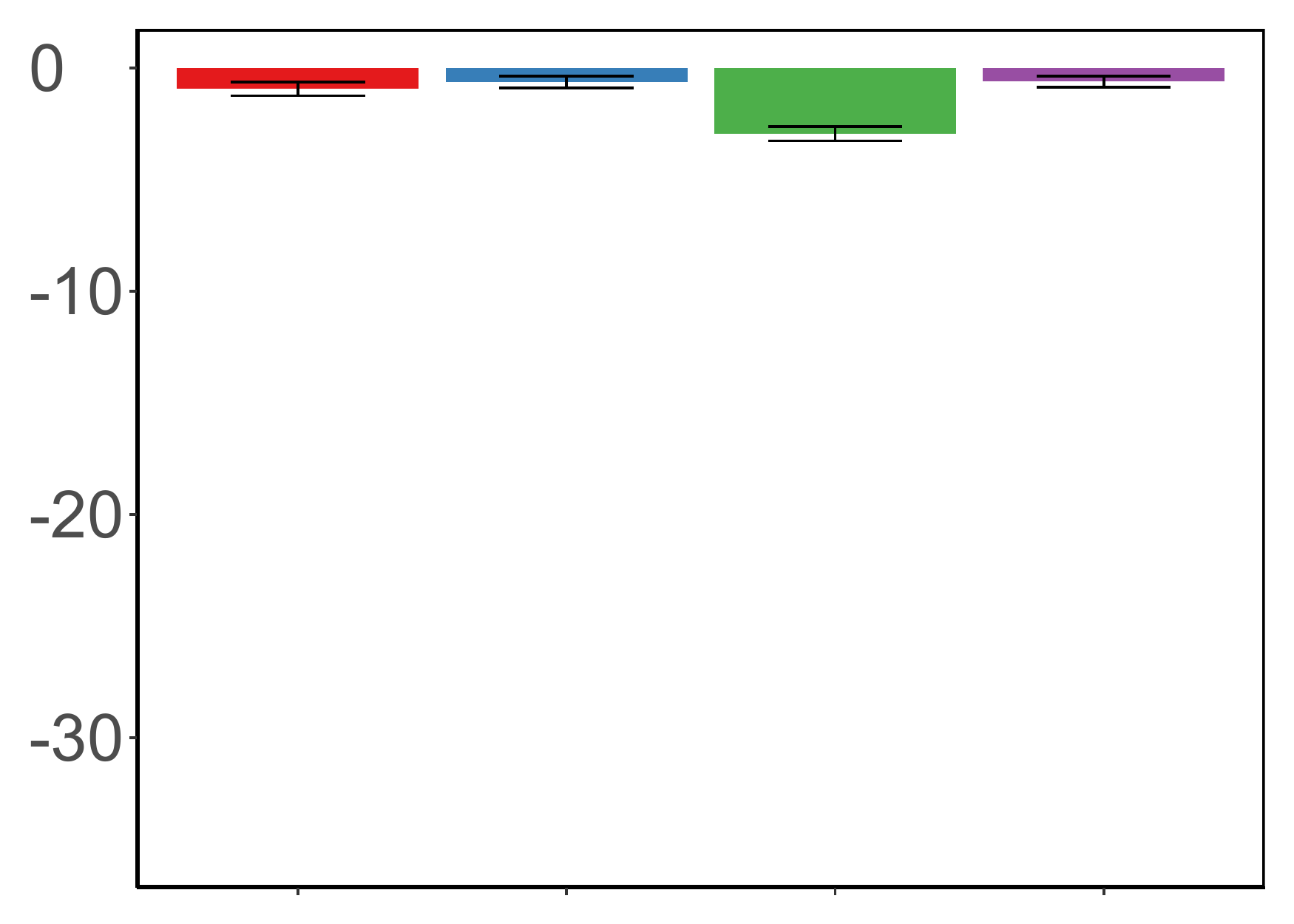} &
\includegraphics[width=\linewidth]{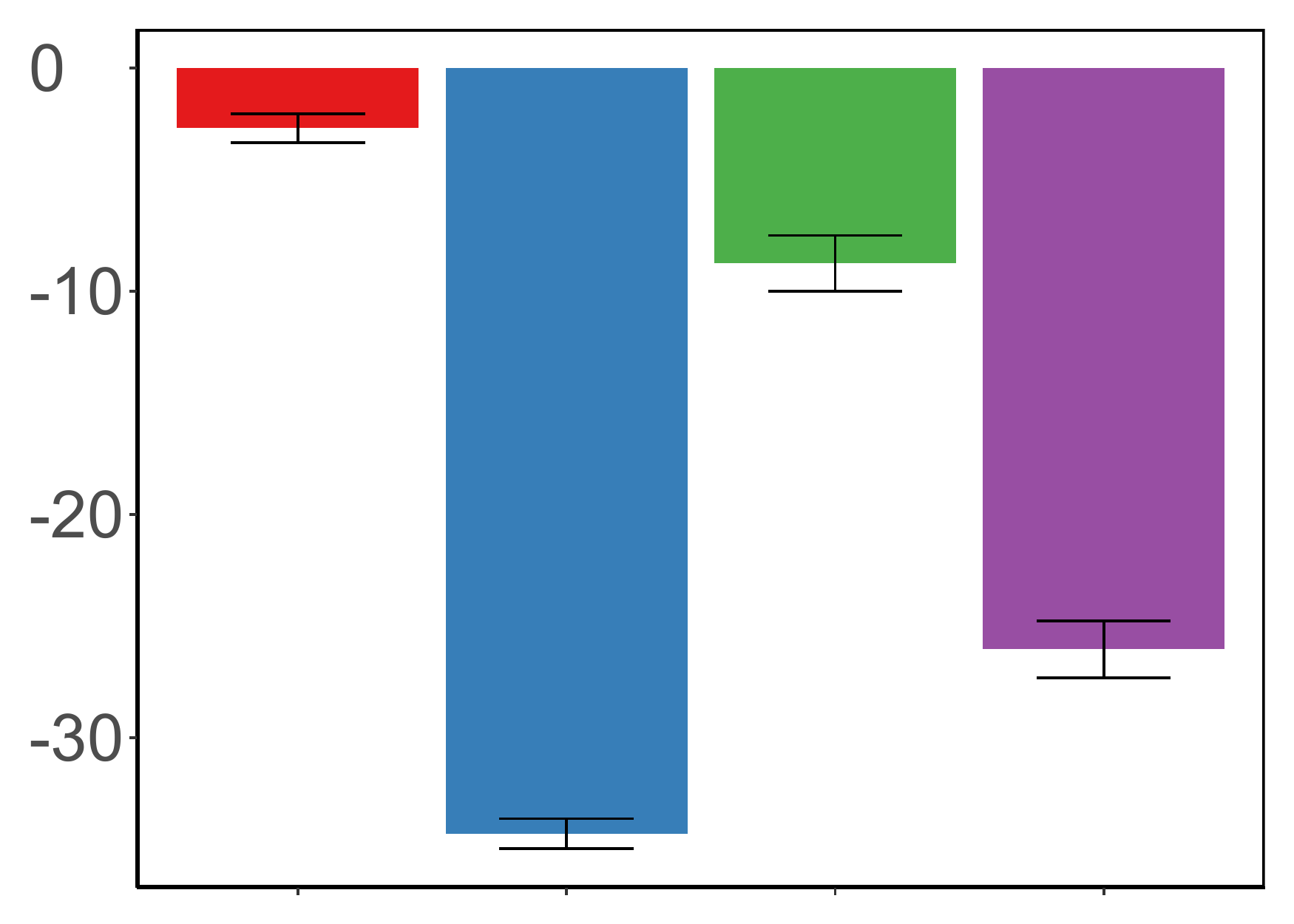} &
\includegraphics[width=\linewidth]{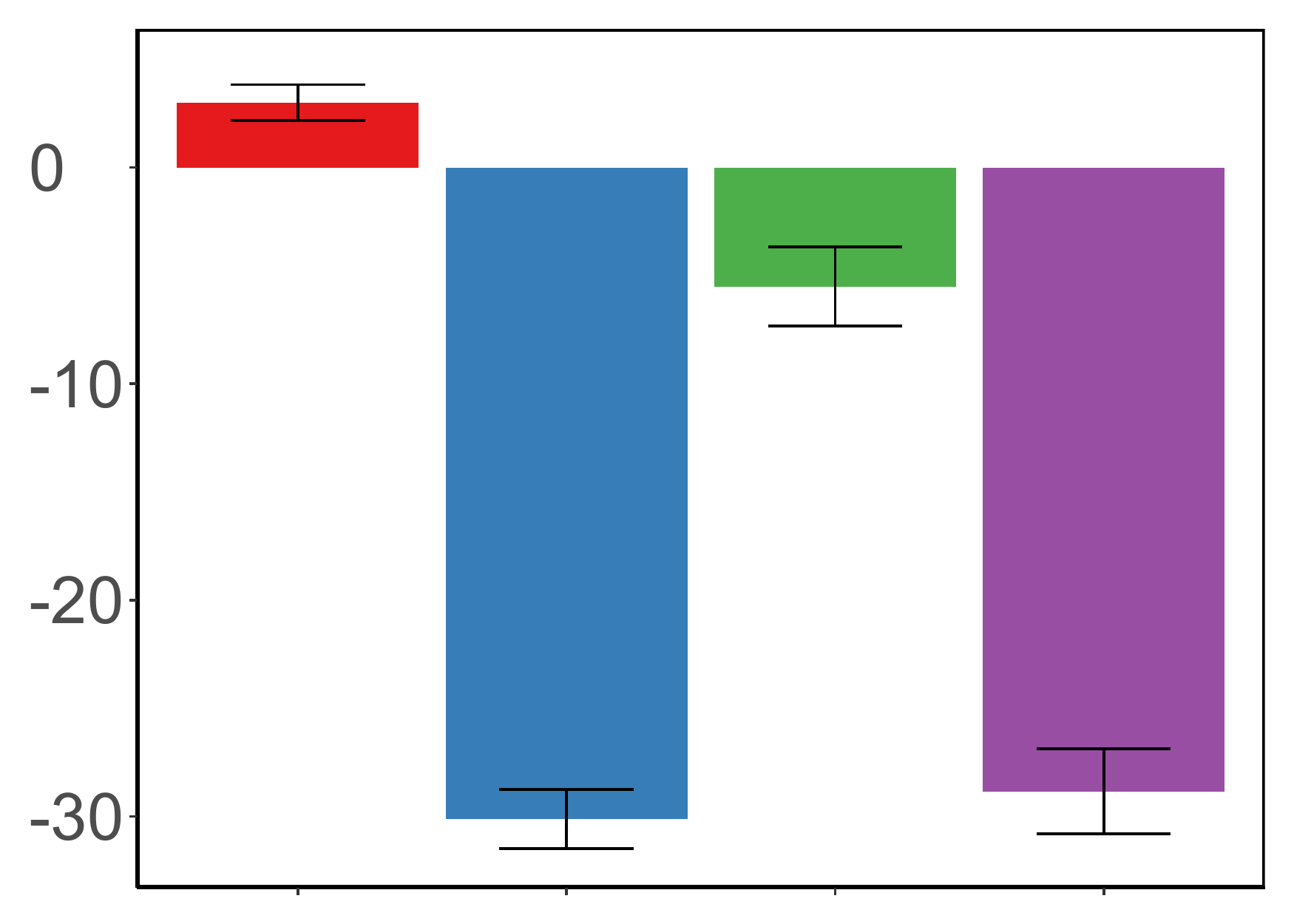} &
\includegraphics[width=\linewidth]{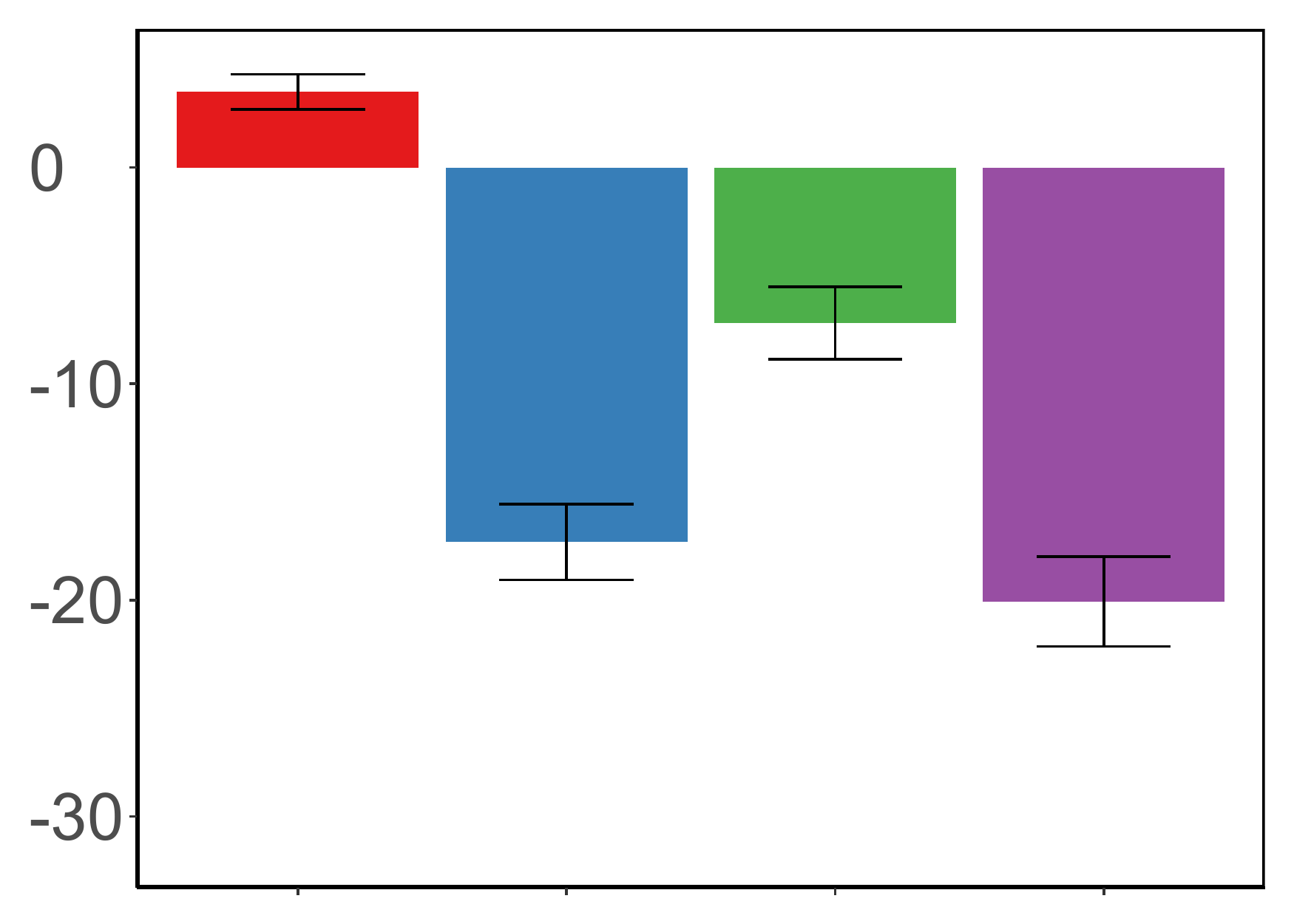} &
\includegraphics[width=\linewidth]{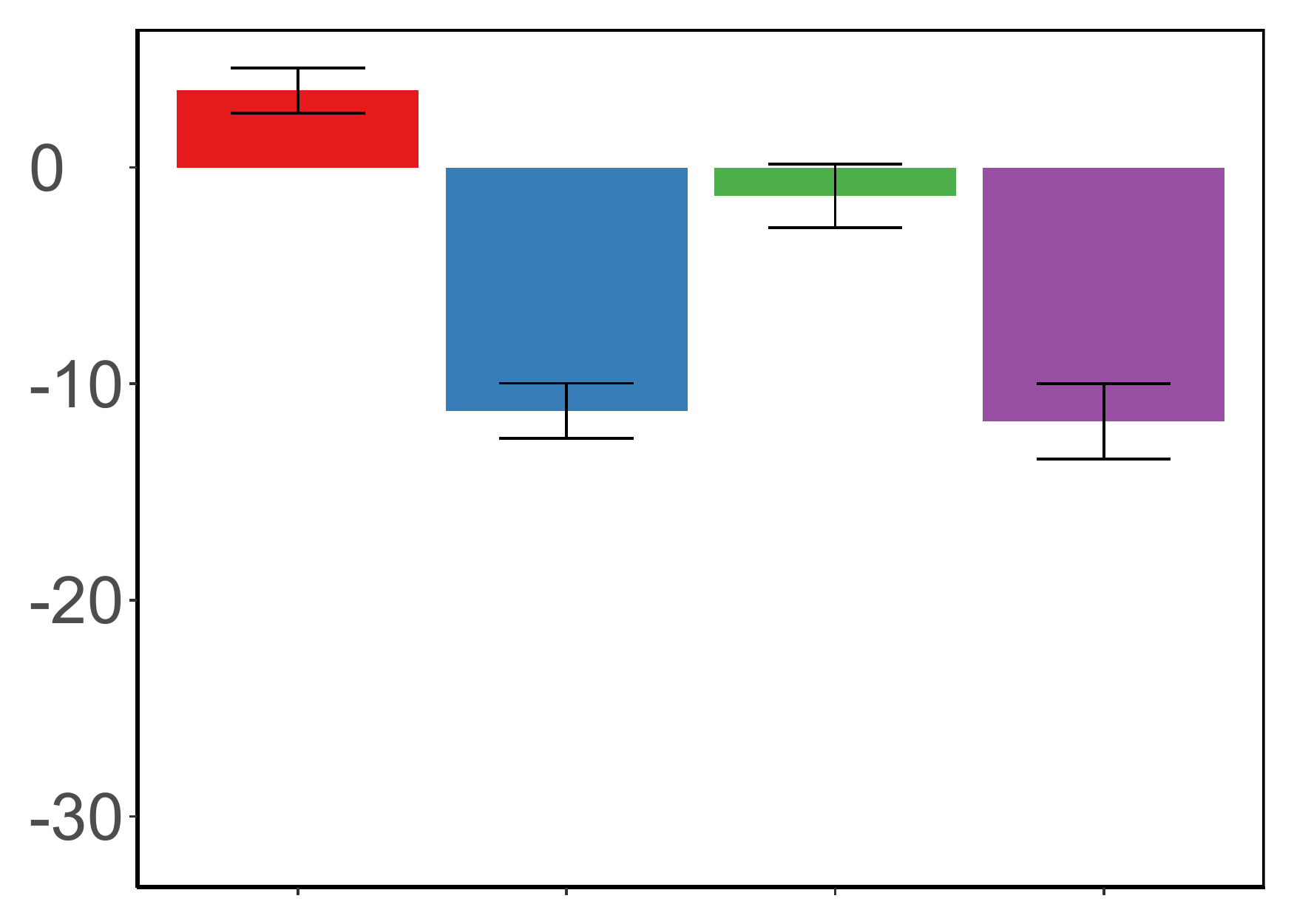} \\
\rotatebox{90}{CS Aarhus} &
\includegraphics[width=\linewidth]{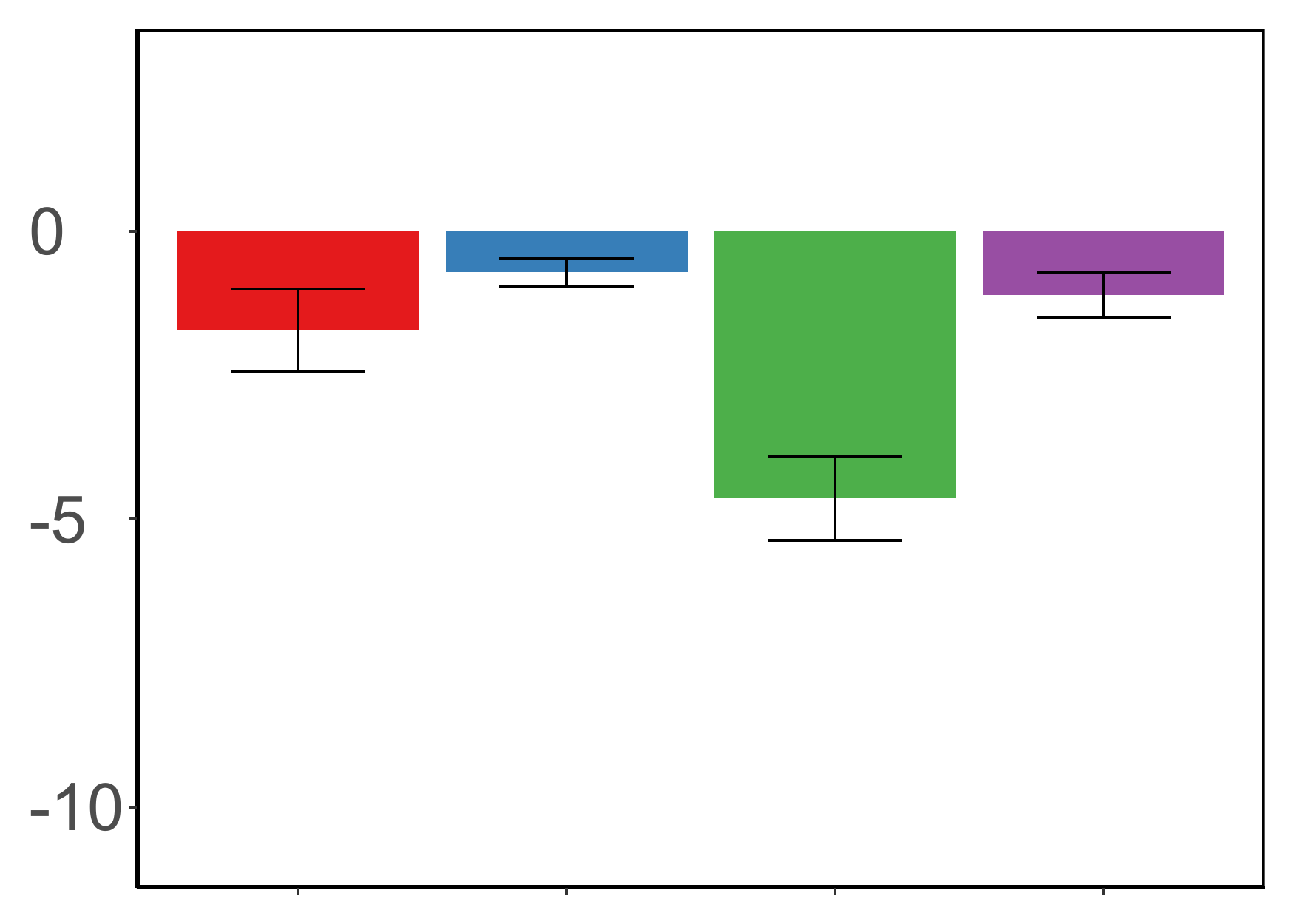} &
\includegraphics[width=\linewidth]{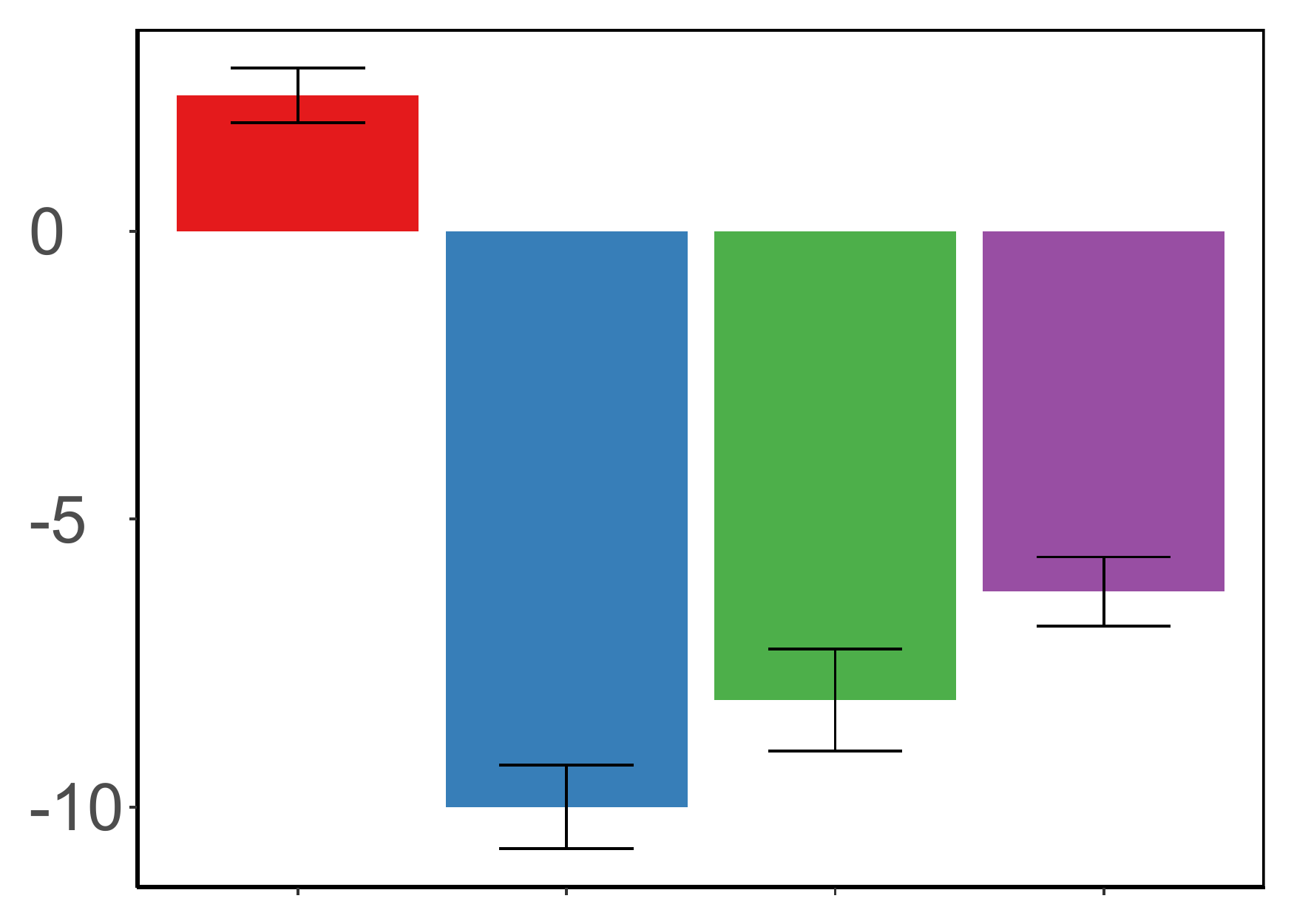} &
\includegraphics[width=\linewidth]{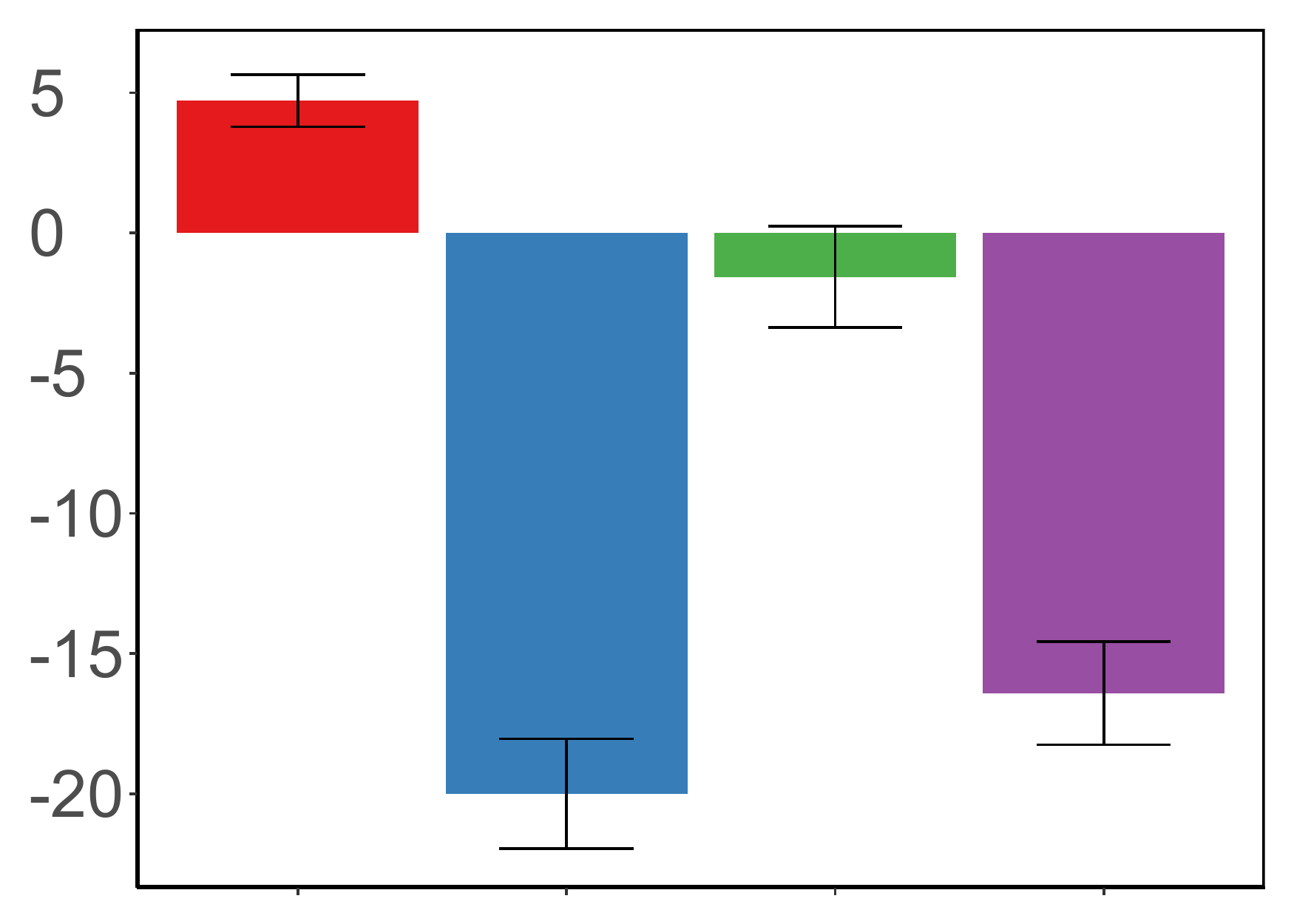} &
\includegraphics[width=\linewidth]{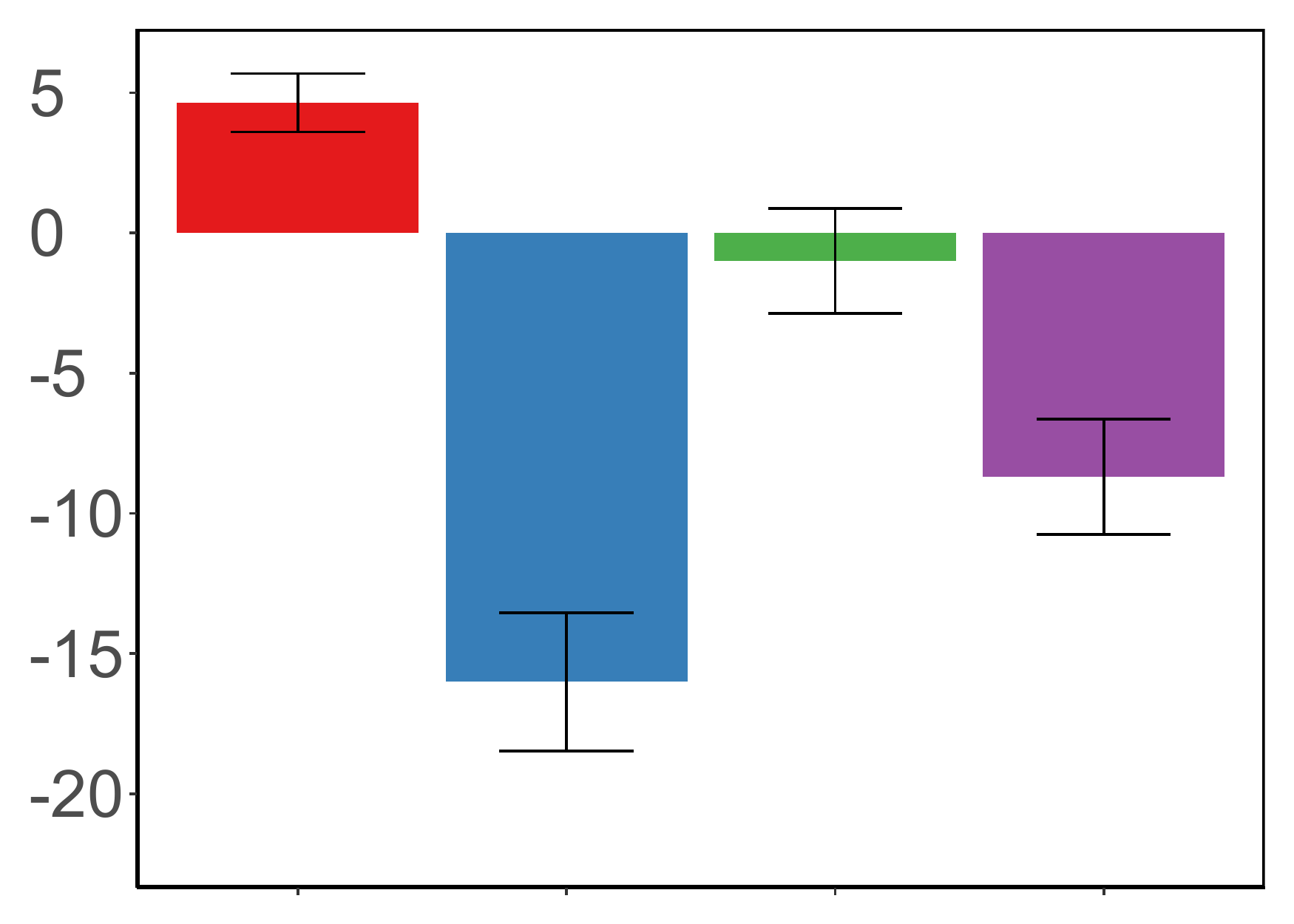} &
\includegraphics[width=\linewidth]{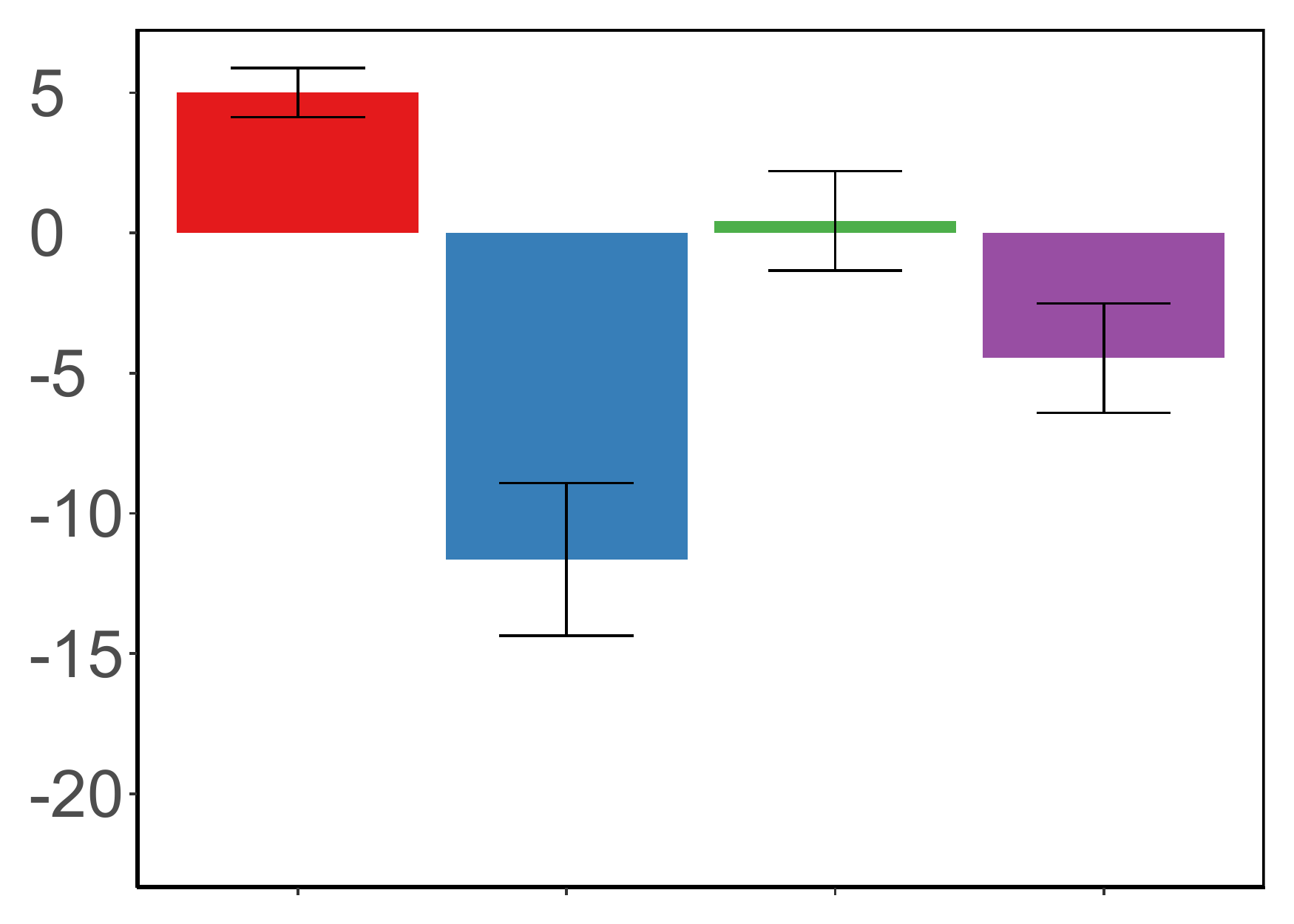} \\
\multicolumn{6}{c}{\includegraphics[width=.35\linewidth]{figures/plots/mhiding/legend-rankDelta}}
\end{tabular}
\caption{
Given different centrality measures and different networks, the figure depicts the average change in centrality ranking of $10$ different evaders as a result of execution of different hiding heuristics.
For the randomly generate networks the experiment is repeated $100$ times, with a new network generated each time.
Error bars represent $95\%$ confidence intervals.
}
\label{fig:simulations-reallife}
\end{figure*}

\subsection{The Simulation Process}
In our simulations, we consider \textit{local} degree, closeness and betweenness centrality, as well as \textit{global} closeness and betweenness centrality. The reason behind excluding global degree centrality is that, as stated in Observation~\ref{thrm:npc-degree-global}, for any given group of contacts, the centrality ranking of the evader does not depend on the way in which connections are distributed across the different layers.
The simulation process is as follows.
For every network, we pick as potential evaders the nodes that are ranked among the top $10$ according to at least one of the five considered centrality measures.
We then simulate the hiding process for each one of those evaders separately.
To this end, we choose the group of contacts to be the neighbors of the evader in the original network. After that, we remove all original edges between the evader and those contacts, and act as if the evader was never connected to those individuals, but rather wants to connect to them while remaining hidden from centrality analysis. Finally, we connect the evader to the contacts using edges chosen by one of our heuristics.
We record the difference between the ranking of the evader in the original, unchanged network, and in the network after running the heuristic. In so doing, we quantify the impact of strategically choosing the relationships to be formed with the group of contacts.
Note that for the local centrality measures, we need to aggregate the centrality scores for each layer into a single ranking for the entire network.
We do so by assigning to each node $v$ the following centrality score: $\frac{1}{\min_{\alpha \in L}r^\alpha(v)}$, where $r^\alpha(v)$ is the ranking of $v$ in layer $\alpha$.

\subsection{Simulation Results}

The results of our simulations are presented in Figures~\ref{fig:simulations-random} and~\ref{fig:simulations-reallife}. Each row corresponds to a network, and each column corresponds to centrality measure.
Each bar represents the change in the evader's ranking after using a particular heuristic (the color of the bar corresponds to the heuristic being used).
A negative change implies that the ranking of the evader decreased, \ie, she became more hidden. In contrast, a positive change implies that the heuristic backfired, \ie, the evader actually became more exposed.

As can be seen, there is no heuristic that dominates the others, i.e., no heuristic is superior against all centrality measures.
The ``All in one'' heuristic proves to be effective in hiding from global closeness centrality in many cases.
Unfortunately, if the network is analyzed with one of the local centrality measures, the evader may become even more exposed.
For every considered centrality measure, either the Density or the Fringe heuristic is among the most effective methods for hiding, and they never make the evader more exposed.
Finally, commenting on the results of the Random heuristic, they demonstrate that it is relatively effective to simply get rid of excess links (\ie, avoid connecting with each node in more than one layer) and spread the remaining connections uniformly.

Our results show also that the global centrality measures are on average much harder to hide from than their local counterparts. This demonstrates the importance of analyzing the entire structure of a multilayer network, rather than focusing on each layer separately.

Regarding the size of the networks used in the simulations, note that the heuristics use only local information and can be easily applied in much larger networks.
However, the cost of computing complete rankings of the multilayer centrality measures, which is necessary for us to present our results, grows quickly with the size of the network.
Hence, we present results for the networks of moderate size.

\section{Conclusions}
\label{sec:conclusions}

We studied the problem of evading centrality analysis in multilayer networks, and analyzed this problem both theoretically and empirically, thereby initiating the study of evading social network analysis tools in multilayer networks. Interesting future directions include developing more sophisticated heuristics for evading centrality measures, and analyzing the problem of evading link-prediction algorithms in multilayer networks.

\section{ Acknowledgments}
Marcin Waniek was supported by the Polish National Science Centre grant 2015/17/N/ST6/03686.
Tomasz Michalak was supported by the Polish National Science Centre grant 2016/23/B/ST6/03599.

\bibliographystyle{abbrv}
\bibliography{bibliography-multilayer}

\end{document}